\documentclass[a4paper,11pt]{article}
\usepackage[left=2cm,right=2cm,top=2.5cm,bottom=2.5cm]{geometry}
\usepackage{amsmath,amsfonts,amssymb,bm,mathrsfs}
\usepackage{comment}
\usepackage{amsthm}
\newtheorem{proposition}{Proposition}
\newtheorem{lemma}{Lemma}
\newtheorem{remark}{Remark}
\newtheorem{ex}{Example}
\theoremstyle{definition}

\usepackage{caption}
\allowdisplaybreaks[1]
\usepackage{xcolor}
\usepackage{mathtools}
\usepackage{graphicx}
\usepackage{algorithm,algpseudocode, algorithmicx}
\floatname{algorithm}{Algorithm}
\algrenewcommand{\algorithmicrequire}{\textbf{Input:}}
\algrenewcommand{\algorithmicensure}{\textbf{Output:}}
\algrenewcommand{\algorithmicprocedure}{\textbf{Step:}}
\algrenewtext{EndProcedure}{\algorithmicend\;\textbf{Step}}

\usepackage{authblk}

\usepackage{enumerate,booktabs,multirow,multicol,adjustbox}

\usepackage{setspace}
\usepackage[sectionbib]{natbib}
\usepackage[colorlinks=true, allcolors=blue]{hyperref}
\usepackage{placeins}


\renewcommand{\vec}[1]{\mathbf{#1}}
\newcommand{\mat}[1]{\mathrm{#1}}
\newcommand{\parvec}[1]{\bm{#1}}

\newcommand{\pr}{\mathbb{P}}
\newcommand{\E}{\mathbb{E}}

\renewcommand{\sp}[1]{\mathfrak{#1}}
\newcommand{\bigO}{\mathcal O}
\DeclareMathOperator{\argmin}{\arg\min}
\newcommand{\smallO}{{\scriptstyle\mathcal O}}
\newcommand{\iid}{\mathrm{i.i.d.}}
\newcommand{\se}[1]{\scriptsize{(\textit{#1})}}
\newcommand{\best}[1]{\underline{#1}}

\definecolor{soham}{RGB}{0,167,159}

\definecolor{annesha}{RGB}{148,0,211}

\title{A scalable version of MADD for big-data classification}
\author[1]{Annesha Ghosh}
\author[2]{Adrija Saha}
\author[3]{Soham Sarkar\footnote{Corresponding author. Email id: \url{sohamsarkar@isical.ac.in}}}
\affil[1]{Theoretical Statistics and Mathematics Unit, Indian Statistical Institute, \protect\\ 203, B.\ T.\ Road, Kolkata 700108, India.}
\affil[2]{Chief Model Risk Officer Department, UBS Business Solutions (India) Pvt. Ltd., \protect\\ Airoli Knowledge Park, Navi Mumbai 400708, India}
\affil[3]{Interdisciplinary Statistical Research Unit, Indian Statistical Institute, \protect\\ 203, B.\ T.\ Road, Kolkata 700108, India.}

\date{}

\usepackage{float}
\begin{document}

\maketitle

\singlespacing

\begin{abstract}
Distance-based classifiers are very popular, and the Euclidean distance is one of the most commonly used metrics in distance-based classifiers. However, classifiers based on the Euclidean distance often suffer in high-dimensional setups due to issues such as distance concentration, violation of neighborhood structures, and the presence of hubs. In high-dimension, low-sample-size (HDLSS) situations, a data-driven semi-metric called the Mean Absolute Difference of Distances (MADD) is known to circumvent these issues. But one major problem with MADD is that its computational complexity increases quadratically with the training sample size. As a result, the application of MADD becomes computationally challenging for big datasets that have both a high dimension as well as a large number of observations. In this paper, we propose a scalable version of MADD that significantly reduces its computational complexity while retaining its advantages. This speed-up is achieved by selecting a representative set during the computation of MADD. Further speed-ups are achieved by using the idea of Random Fourier Features, particularly when the sample size is very large. We establish that our proposed methods achieve performances similar to MADD but only at a fraction of its computing time, both theoretically as well as numerically. Our approach broadens the scope of MADD, allowing its use to big-data with a very large number of observations.

\end{abstract}

\section{Introduction}
\label{sec:intro}

In a classification problem, our goal is to assign one of $J (\ge 2)$ labels to an observation based on a collection of labelled observations (called a training sample). These problems form one of the cores of modern machine learning, and they draw significant attention from the entire scientific community. There are several classification techniques available in the literature. Among them, distance-based classifiers have been extensively used due to their simplicity and interpretability. This type of classifiers assign a label to a data point based on its proximity to the training observations. The proximity is measured by some distance between the observations, and the most commonly used distance is the Euclidean distance \citep{duda2001pattern,hastie2009elements}.

Although Euclidean-distance-based classifiers work well in low dimensions, their performance often deteriorates when the dimension becomes high. Consider a two-class classification problem (i.e., $J=2$), where the $j$-th population has mean $\parvec\mu_j$ and variance $\parvec\Sigma_j$, $j=1,2$. Let $\vec X_1,\vec X_1^\prime$ and $\vec X_2,\vec X_2^\prime$ be indepdent and identically distributed (i.i.d.) $d$-dimensional random vectors from class-1 and class-2, respectively. Then, under some regularity conditions, the following concentrations of Euclidean distances occur:

\begin{align*}
\frac{1}{d}\|\vec X_1 - \vec X_1^\prime\|^2 \overset{P}{\to} 2\sigma_1^2,\; \frac{1}{d}\|\vec X_2 - \vec X_2^\prime\|^2 \overset{P}{\to} 2\sigma_2^2, \text{ and } \frac{1}{d}\|\vec X_1 - \vec X_2\|^2 \overset{P}{\to} \sigma_1^2 + \sigma_2^2 + \nu^2 \quad \text{as } d \to \infty,
\end{align*}
where $\sigma_j^2 = \lim_{d \to \infty} \mathrm{trace}(\parvec\Sigma_j)/d$, $j=1,2$, and $\nu^2 = \lim_{d \to \infty} \|\parvec\mu_1-\parvec\mu_2\|^2/d$ \citep[][]{hall2005geometric,chan2009scale}. Note that $\nu^2$ quantifies the location difference between the two populations, while $\sigma_1^2$ and $\sigma_2^2$ quantify the scales of the two populations, respectively. Now, it is easy to see that if $\nu^2 < \sigma_1^2 - \sigma_2^2$, then $\|\vec X_1 - \vec X_1^\prime\|$ becomes larger than $\|\vec X_1 - \vec X_2\|$ with probability tending to $1$. In other words, when the scale difference between the two populations becomes larger than their location difference, the Euclidean distance between two observations from the same class becomes larger than the Euclidean distance between two observations from different classes. This violation of neighborhood structure can severely impact the performance of classifiers based on the Euclidean distance, as shown by \cite{hall2005geometric,chan2009scale,pal2016high,roy2022generalizations}. Moreover, the distance concentration phenomenon results in the occurrences of hubs, which further affect the performance of Euclidean-distance-based classifiers \citep{radovanovic2010hubs,tomavsev2014hubness}.

To overcome these issues with the Euclidean distance, \cite{pal2016high} proposed a data-driven semi-metric called the Mean Absolute Difference of Distances or MADD. For $j=1,\ldots,J$, let $\sp X_j = \{\vec X_{j1},\ldots,\vec X_{jn_j}\}$ denote the training sample (of size $n_j$) from class-$j$, and let $\sp X = \cup_{j=1}^J \sp X_j$ denote the entire training sample (of size $n=\sum_{j=1}^J n_j$). Then, MADD between two points $\vec u$ and $\vec v$ is defined as
\begin{align}\label{eq:MADD}
\rho(\vec u,\vec v) := \frac{1}{|\sp X \setminus \{\vec u, \vec v\}|} \sum_{\vec z \in \sp X \setminus \{\vec u,\vec v\}} \Big| \|\vec u - \vec z\| - \|\vec v - \vec z\| \Big|.
\end{align}
Here, for a set $S$, $|S|$ denotes its cardinality. It can be shown that under the same set of regularity conditions as mentioned above, $\rho(\vec X_1,\vec X_1^\prime)/\sqrt{d}$ and $\rho(\vec X_2,\vec X_2^\prime)/\sqrt{d}$ converge to $0$, whereas $\rho(\vec X_1,\vec X_2)/\sqrt{d}$ converges to a constant which is strictly positive as long as either $\nu^2 > 0$ or $\sigma_1^2 \ne \sigma_2^2$ \citep{pal2016high}. In other words, when the underlying populations differ in either their locations or their scales, MADD maintains the correct neighborhood structure with high probability. Consequently, classifiers based on MADD outperform those based on the Euclidean distance, especially in high-dimensional situations. This has been well articulated in \cite{pal2016high,roy2022generalizations}. MADD has also been successfully used in two-sample testing problems \citep{sarkar2020some} and clustering \citep{sarkar2019perfect} with high-dimension, low-sample-size data.

Despite the advantages of MADD, its computational complexity poses a significant challenge, especially when the sample size is large. To see this, notice that the calculation of MADD between two observations requires $\bigO(nd)$ operations, compared to $\bigO(d)$ for the Euclidean distance. Moreover, in a classification problem, we need to compute these pairwise distances for all the training sample observations, making the computing time even higher. For instance, consider the nearest neighbor (NN) classifier. In order to classify an observation using the NN classifier, we need to compute the distance of that observation from all the training sample observations. For MADD, this requires computations of the order $\bigO(n^2d)$, which can become very challenging, and sometimes even prohibitive. To demonstrate this, in Table~\ref{tab:Introduction}, we report the computing times for classifying  $5000$ test observations (with $d=100$) using the NN classifier with MADD (henceforth, NN-MADD) for a range of training sample sizes. For comparison, we also report the corresponding computing times for the NN classifier with the Euclidean distance (henceforth, NN). Clearly, the computing time for NN-MADD increases exponentially with the training sample size. For $n=16384$, the classification takes more than $6.5$ hours to complete. Datasets with such large number of training samples are quite common now-a-days, e.g., with large-scale gene expressions or high-resolution images. 

\begin{table}[t!]
\centering
\caption{Computing times (in seconds) for NN classifier with the Euclidean distance and MADD for different training sample sizes ($n$). The experiments were conducted on a system with AMD Ryzen 7 CPU @ 3.20 GHz, 16 GB RAM, an RTX 30-series GPU, and Windows 11 (64-bit) OS.}
\label{tab:Introduction}
\begin{tabular}{l|rrrrrr} \toprule
Training sample size ($n$) &  $512$ & $1024$ & $2048$ & $4096$ & $8192$ & $16384$\\ \hline
NN with Euclidean distance & $0.80$ & $1.50$ & $2.89$ & $5.73$ & $11.39$ & $22.52$ \\
NN with MADD & $19.89$ & $81.26$ & $357.108$ & $1249.31$ & $4758.40$ & $24003.98$ \\\bottomrule
\end{tabular}
\end{table}

To address this issue, we revisit the theoretical behavior of MADD. As mentioned above, for two observations $\vec U$ and $\vec V$ from the same population, $\rho(\vec U,\vec V)/\sqrt{d}$ converges to $0$ as $d \to \infty$. In fact, each summand in \eqref{eq:MADD} scaled by $\sqrt{d}$ converges to $0$ as $d \to \infty$. On the other hand, if $\vec U$ and $\vec V$ are from different populations, and the populations differ in their locations and/or scales (in the sense mentioned above), then some of the summands in \eqref{eq:MADD} scaled by $\sqrt{d}$ converge to a strictly positive quantity. This ensures that MADD is able to discriminate between populations which differ in their locations and/or scales. However, in order to achieve this separation, we do not necessarily need to consider all the summands, or equivalently all the $\vec Z$'s in $\sp X \setminus \{\vec U, \vec V\}$, while defining \eqref{eq:MADD}. It is enough to select a few observations which ensure that the sum is positive (asymptotically, after proper scaling) for two observations $\vec U$ and $\vec V$ from different populations.

Based on the above observation, we propose a scalable version of MADD by redefining it using a small representative set of observations. In order to retain the advantages of MADD, the representative set needs to capture the landscape of the entire training sample. We achieve this by employing a determinantal point process \citep[DPP,][]{kulesza2012determinantal} on the training sample. The details of our proposal are given in the next section. It is observed that the proposed scalable MADD is able to achieve superior performance even when the representative set is of a small size. However, the choice of the number of representative observations is crucial in the performance of the classifier, and we propose a data-adaptive way of choosing this number. We also utilize the notion of random Fourier features (RFF) to further speed up our procedure, particularly for massive/huge datasets.

The paper is organized as follows. We present our methodology in Section~\ref{sec:methodology}. In the same section, we also propose a novel incremental cross-validation technique to select the number of representative samples in a data-adaptive manner. An extensive simulation study is conducted in Section~\ref{sec:Other classifiers} to compare the performance of the proposed classifier with existing methods. In Section~\ref{sec:RFF}, we define the RFF-based speed-up of the proposed method, which can handle huge number of samples. A scalable version of a generalized version of MADD is proposed in Section~\ref{sec:gmadd}. In Section~\ref{sec:Real_Data}, we demonstrate the utility of the proposed method on several benchmark datasets. The theoretical behavior of the proposed classifier is studied in Section~\ref{sec:Theorem}. A brief summary of the work and some concluding remarks are given in Section~\ref{sec:Conclusion}. Implementation details of the proposed method, proofs of mathematical results, and some additional numerical results are provided in the appendices. 

\section{Proposed Method\label{sec:methodology}}

To motivate the idea behind our method, we revisit the definition of MADD (Eq.~\eqref{eq:MADD}). For two random vectors $\vec U$ and $\vec V$, we have
\begin{align*}\
d^{-1/2} \rho(\vec U, \vec V) = \frac{1}{|\sp X \setminus \{\vec U, \vec V\}|} \sum_{j=1}^J \sum_{\vec Z \in \sp X_j: \vec Z \ne \vec U, \vec V} d^{-1/2} \Big| \|\vec U - \vec Z\| - \|\vec V - \vec Z\| \Big|.
\end{align*}
If $\vec U$ and $\vec V$ are from the same population, then all the summands in the right converge to $0$, resulting in $\rho(\vec U,\vec V)/\sqrt{d}$ to converge to $0$. In fact, for this convergence to hold, it is enough to consider a single observation $\vec Z$ from the training sample which is different from $\vec U$ and $\vec V$. The more interesting case is when $\vec U$ and $\vec V$ are from different populations. In this case, each summand within the inner-sum on the right converges to the same non-negative quantity. The sum of these quantities (and hence, the overall sum) is positive if and only if the underlying populations differ in their locations and/or scales, i.e, $\nu_{jj^\prime}^2 > 0$ and/or $\sigma_j^2 \ne \sigma_{j^\prime}^2$ \citep{pal2016high,sarkar2019perfect}. Here, $\nu_{jj^\prime}^2 = \lim_{d \to \infty} \|\parvec\mu_j - \parvec\mu_{j^\prime}\|^2/d$ and $\sigma_j^2 = \lim_{d \to \infty} \mathrm{trace}(\parvec\Sigma_j)/d$ are as defined before. Therefore, for $\rho(\vec U,\vec V)/\sqrt{d}$ to converge to a positive quantity, it is enough to consider one observation from each class other than $\vec U$ and $\vec V$. In essence, the class separation property of MADD remains intact if we redefine it using a representative set consisting of at least two observations from the training samples of each class.

More generally, let $\sp X^\ast$ be a collection of representative observations from the training data, i.e., $\sp X^\ast \subset \sp X$. Using this representative set, we re-define MADD as
\begin{align}\label{eq:MADD_scalable}
\rho_{\rm sc}(\vec u,\vec v) := \frac{1}{|\sp X^\ast \setminus \{\vec u, \vec v\}|} \sum_{\vec z \in \sp X^\ast \setminus \{\vec u,\vec v\}} \Big| \|\vec u - \vec z\| - \|\vec v - \vec z\| \Big|.
\end{align}
If $\sp X^\ast$ consists of $k$ observations, then the calculation of $\rho_{\rm sc}$ between two points requires a computation of the order $\bigO(kd)$, compared to $\bigO(nd)$ for the usual MADD. Similarly, NN classification using $\rho_{\rm sc}$ requires $\bigO(knd)$ computations compared to $\bigO(n^2d)$ for the usual MADD. When $k$ is significantly smaller than $n$, this modified MADD results in a substantial decrease in the computational cost. Therefore, we call $\rho_{\rm sc}$ the \emph{scalable version of MADD} (in short, MADD$_\mathrm{sc}$).

\subsection{Selection of Representative Observations}

As we noted above, the theoretical advantages of MADD can be retained by considering as small as $2$ observations from each class as representatives (i.e., with $k=2J$). The practical performance of the method, however, depends on the choice of the representative set, and a judicial choice is needed to achieve good performance. One can use popular methods like condensed nearest neighbors \citep[CNN,][]{hart1968condensed} or reduced nearest neighbors \citep[RNN,][]{gates1972reduced} for this purpose. In fact, \cite{pal2016high} suggested their use in a similar but different context. Both CNN and RNN methods perform repeated nearest-neighbor searches over the training set, requiring significant computational times. Instead, one can use simple random sampling without replacement \citep[SRSWOR,][]{cochran1977sampling} -- a foundational method for selecting a subset of units -- where every unit has an equal probability of being selected. However, SRSWOR completely disregards any intrinsic structure present in the data, and may fail to capture its diversity.

To demonstrate this, we consider two binary classification problems with bivariate populations. In Example~A, the two competing populations have the bivariate normal distribution $N_2(\vec 0, 3 \mat I)$ and the standard bivariate $t$ distribution with $3$ degrees of freedom, respectively. In Example~B, the populations are unequal mixtures of bivariate normal distributions. For the first population, we take a mixture of $N_{2}((0,5)^\top,\,2 \mat I)$ and $N_{2}((0,-5)^\top,\,\mat I)$ with mixing proportions $0.25$ and $0.75$ respectively, while the second population is a mixture of $N_{2}((5,0)^\top,\,\mat I)$ and $N_{2}((-5,0)^\top,\,2\mat I)$ with mixing proportions $0.75$ and $0.25$ respectively. For both of these examples, we generate $500$ observations from the two competing populations. Then, we employ SRSWOR to select $50$ representative observations from these $500$ observations. In Figure~\ref{fig:Intro}, we show the scatter plots of the original set of $500$ observations and the sample of $50$ observations selected using SRSWOR, separately from both the populations. Clearly, the representative samples selected using SRSWOR miss intricate structures present in the data, especially when the underlying populations are multimodal, and the data is heterogeneous.

\begin{figure}[!t]
	\centering
	\begin{tabular}{c c c c c}
	\multicolumn{5}{c}{(a) Example~A} \\ [5pt]
	~~~~Original data && ~~~~SRSWOR sample && ~~~~DPP sample \\
	\includegraphics[width=0.23\linewidth]{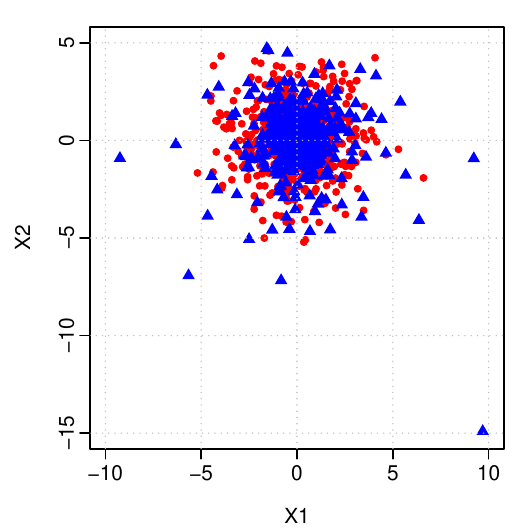} &&
	\includegraphics[width=0.23\linewidth]{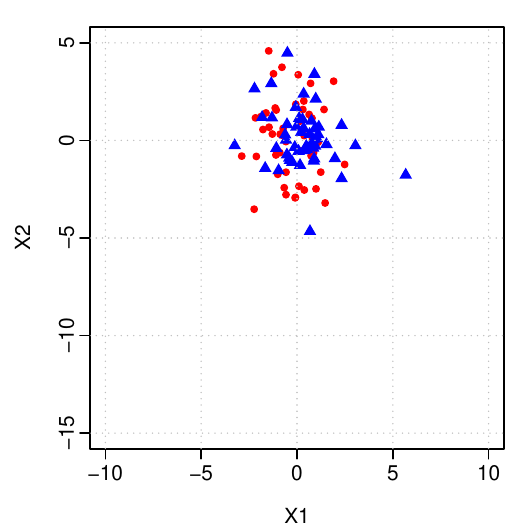} &&	
	\includegraphics[width=0.23\linewidth]{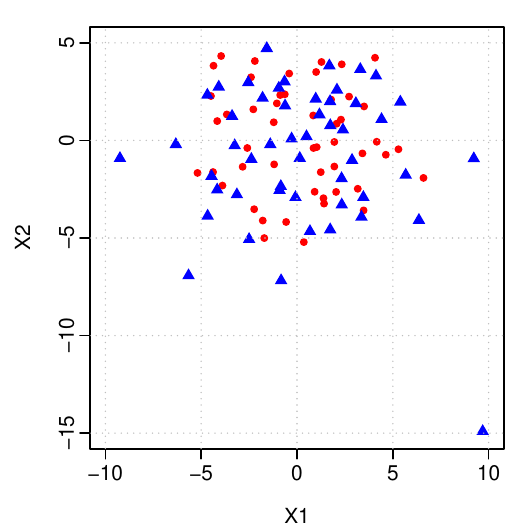} \\ [3pt]
	\multicolumn{5}{c}{(b) Example~B} \\ [5pt]
	~~~~Original data && ~~~~SRSWOR sample && ~~~~DPP sample \\
	\includegraphics[width=0.23\linewidth]{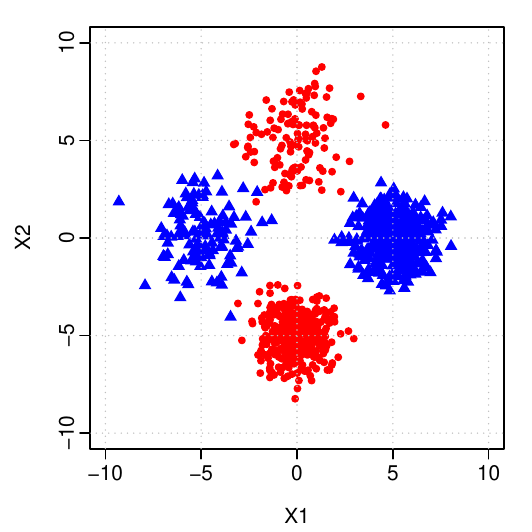} &&
	\includegraphics[width=0.23\linewidth]{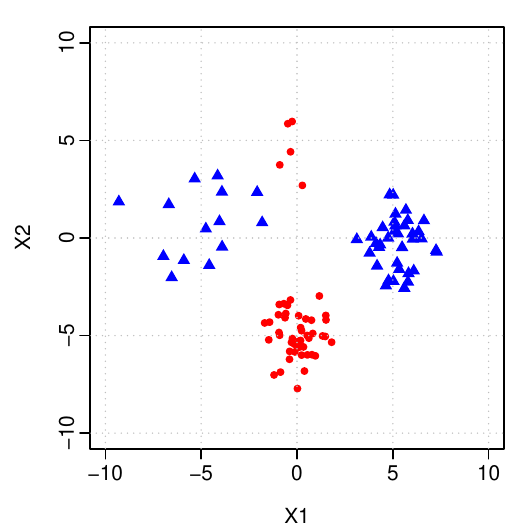} &&	
	\includegraphics[width=0.23\linewidth]{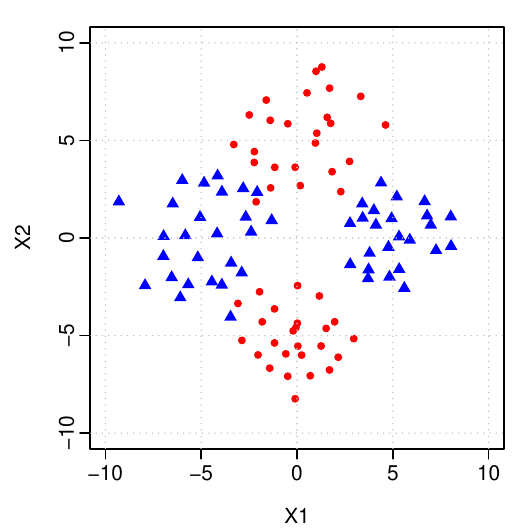}
	\end{tabular}

    \caption{Scatter plots of $500$ observations from the two populations (left), and $50$ observations selected using simple random sampling without replacement (SRSWOR, middle) and determinantal point process (DPP, right) in Examples~A and B. Here, \textcolor{red}{red $\bullet$} (respectively, \textcolor{blue}{blue $\scriptstyle\blacktriangle$}) represents observations from Class 1 (respectively, Class 2).\label{fig:Intro}}
\end{figure}

Thus, SRSWOR samples are not suitable as representative observations for defining MADD$_{\rm sc}$. To appreciate this, recall the original definition of MADD (Eq.~\eqref{eq:MADD}). In order to have the same information as MADD, the representative set should be chosen in such a way that all the possible variations in the pairwise distances (as captured by the summands in Eq.~\eqref{eq:MADD}) are retained. This is ensured if the chosen representative observations resemble the structure of the original data, which may not be the case with an SRSWOR sample.

Based on the above discussion, it is clear that our representative set needs to capture the structure of the original data. Also, to address our computational concerns, the selection strategy should be fast to implement. In this article, we employ Determinantal Point Process \citep[DPP,][]{kulesza2012determinantal} to select the representative set. The representative observations ($50$ from each class) selected using DPP in Examples~A and B are depicted in the right panels of Figure~\ref{fig:Intro}. Clearly, DPP samples exhibit a much better representation of the original data compared to SRSWOR samples. It selects fewer observations from denser regions and more observations from sparse regions, essentially filling the entire data cloud. We describe the representative set selection technique below, starting with a brief description of DPP.

Point processes provide probabilistic models for random subsets of a finite set. More precisely, for a finite set $\mathcal A = \{1,\ldots,N\}$, a point process defines the distribution of a random element $\Psi$ taking values among all possible subsets of $\mathcal A$. For determinantal point processes (DPPs), the distribution is governed by the determinant of submatrices of a real, symmetric matrix. In particular, given an $N \times N$ real, symmetric matrix $\mat K = ((k_{ij}))$, $\pr(\Psi \supseteq A) = \det(\mat K_A)$ for $A \subseteq \mathcal A$, where $\mat K_A = ((k_{ij}))_{i,j \in A}$ is the submatrix of $\mat K$ with indices restricted to $A$. Consequently, for two indices $1\le i \ne j \le N$, we get $\pr(\Psi \ni \{i,j\}) = k_{ii}k_{jj} - k_{ij}^2$. Therefore, for random subsets selected using DPP, the probability of simultaneous inclusion of $i,j$ becomes low if $k_{ij}$ is high -- introducing a repulsion effect in the selected subset \citep{kulesza2012determinantal}.

For a DPP to be a valid probabilistic model, the matrix $\mat K$ needs to satisfy $0 \le \det(\mat K_A) \le 1$ for all $A \subseteq \mathcal A$, which is restrictive. Therefore, in practical implementations of DPPs, a related concept known as $\mat L$-ensemble is used. For an $\mat L$-ensemble, we have $\pr(\Psi = A) \propto \det(\mat L_A)$ for some matrix $\mat L = ((l_{ij}))$. This gives us a valid probabilistic model as long as $\mat L$ is positive semi-definite. Moreover, an $\mat L$-ensemble is a DPP with $\mat K = \mat L(\mat L+\mat I)^{-1}$ \citep[cf.][for more details]{kulesza2012determinantal}. More importantly, similar to DPP, $\mat L$-ensemble assigns lower probabilities to simultaneous inclusion of indices with high $l_{ij}$ values. In particular, if $\mat L$ is a similarity matrix, then subsets selected using the corresponding $\mat L$-ensemble has a low probability of containing highly similar elements. This repulsion among similar elements results in diversity in the selected subset.

In the present context, we are interested in selecting representative sets from each of the competing classes $\sp X_j = \{\vec X_{j1},\ldots,\vec X_{jn_j}\}$, $j=1,\ldots,J$. For each $j=1,\ldots,J$, we can select a subset of $\{1,\ldots,n_j\}$ using an $\mat L$-ensemble. If the selected subset is $\{\pi_1,\ldots,\pi_{n_j^{\ast}}\}$, then the representative set for class-$j$ would be $\sp X_j^\ast = \{\vec X_{j\pi_1},\ldots,\vec X_{j\pi_{n_j^\ast}}\}$. An important issue here is the choice of the matrix $\mat L$. Since we want to capture the diversity among the pairwise Euclidean distances, a natural choice is to use the Gaussian radial basis kernel: $\mat L^{(j)} = ((l_{ii^\prime}^{(j)}))$ for class-$j$ with $l^{(j)}_{ii^\prime} = \exp\{-\|\vec X_{ji} - \vec X_{ji^\prime}\|^2/(2\sigma_j^2)\}$, where $\sigma_j$ is a bandwidth parameter \citep{Bishop2006PRML,hastie2009elements}. In our implementation, we use the highly popular median heuristic: $\sigma_j = \operatorname{median}\{\|\vec X_{ji} - \vec X_{ji^\prime}\|: 1 \le i < i^\prime \le n_j\}$ \citep{garreau2018median}. This choice gives the right scaling in high dimensions \citep{beyer1999nearest}. 

With the Gaussian kernel, the similarity values lie between $0$ and $1$, with $l_{ii}^{(j)} = 1$ for $i=1,\ldots,n_j$. Also, $l_{ii^\prime}^{(j)}$ is a decreasing function of the Euclidean distance $\|\vec X_{ji} - \vec X_{ji^\prime}\|$. Therefore, a subset of $\sp X_j$ selected using $\mat L^{(j)}$-ensemble tend to contain observations which are more distant from each other. However, for a subset selected using $\mat L$-ensemble, the size of the subset is random. This is problematic for us, since we want to avoid both the situations where the selected subset contains either too few (and hence unrepresentative) or too many (leading to higher computing cost) observations. To circumvent this issue, we utilize $k$-DPP \citep[cf.][]{kulesza2012determinantal}, which selects a subset of exactly $k$ points while maintaining repulsion among the selected points. The algorithm for selecting an $k$-DPP sample is given in the Appendix (see Algorithm~\ref{alg:k-DPP}).

After selecting representative observations from each class, we construct the representative set $\sp X^\ast = \sp X_1^\ast \cup \cdots \cup \sp X_J^\ast$. To classify an observation $\vec Z$, we compute the MADD$_{\rm sc}$ values between $\vec Z$ and all the training observations. Then, we classify the observation using the nearest neighbor method with $\rho_{\rm sc}$ as the distance. The entire procedure of nearest neighbor classification using MADD$_{\rm sc}$ (referred to as NN-MADD$_{\rm sc}$ classification) is formalized in Algorithm~\ref{alg:SMNN}.

\begin{algorithm}[t!]
\caption{Nearest Neighbor Classification Using Scalable MADD\label{alg:SMNN}}
	\begin{algorithmic}[1]
	\Require $\sp X_1,\ldots,\sp X_J$ -- Training sample, where $\sp X_j = \{\vec X_{j1},\ldots,\vec X_{jn_j}\}$.\newline
	$k_1,\ldots,k_J$ -- Number of representative samples to be selected.\newline
	$\vec Z$ -- Test observation.
	
	\Ensure $\delta_{\rm sc}(\vec Z)$ -- The predicted class label for $\vec Z$.
	
	\bigskip

	\For {$j \in \{1,\dots, J\}$}
	\State Find $\sigma_j = \operatorname{median}\{\|\vec X_{ji} - \vec X_{ji^\prime}\|: 1 \le i < i^\prime \le n_j\}$.
	\State Compute $l^{(j)}_{ii^\prime} = \exp\{-\|\vec X_{ji} - \vec X_{ji^\prime}\|^2/(2\sigma_j^2)\}$, $i,i^\prime = 1,\ldots,n_j$. Construct $\mat L^{(j)} = ((l^{(j)}_{ii^\prime}))$.
	\State Select a subset $\{\pi_1,\ldots,\pi_{k_j}\}$ of $\{1,\ldots,n_j\}$ using $k_j$-DPP (Algorithm~\ref{alg:k-DPP}) with the matrix $\mat L^{(j)}$.
	\State The representative observations from class-$j$ are $\sp X_j^\ast = \{\vec X_{j\pi_1},\ldots,\vec X_{j\pi_{k_j}}\}$.
	\EndFor
	\State Define $\sp X^\ast = \sp X_1^\ast \cup \cdots \cup \sp X_J^\ast$ to be the representative set.
	\State Compute $\rho_{\rm sc}(\vec Z,\vec X_{ji})$, $j=1,\ldots,J$, $i=1,\ldots,n_j$ using \eqref{eq:MADD_scalable}.
	\State For $j=1,\ldots,J$, compute $\rho_{\rm sc}(\vec Z,\sp X_j) = \min_{i=1,\ldots,n_j} \rho_{\rm sc}(\vec Z,\vec X_{ji})$.
	\State Assign $\delta_{sc}(\vec Z) = \argmin_{j=1,\ldots,J} \rho_{\rm sc}(\vec Z,\sp X_j)$.
	\end{algorithmic}
\end{algorithm}

\subsection{Choice of the Number of Representative Observations\label{sec:cross-validation}}

With the method described above, we are left with the choice of the number of representative observations from each class, viz., $k_1,\ldots,k_J$, and this choice plays a crucial role on the performance of the classifier. To demonstrate this, as well as the effectiveness of the proposed NN-MADD$_{\rm sc}$ classifier, we consider the following binary classification problems with varying sample sizes ($n$) and dimensions ($d$). These are generalizations of Examples~A and B defined earlier.

\begin{ex}\label{example1}
The two population distributions are $N_d(\vec 0,\,3\mat I)$ and the standard multivariate $t$ distribution with $3$ degrees of freedom.
\end{ex}

\begin{ex}\label{example2}
The two population distributions are unequal mixtures of multivariate normal distributions. The observations in class-$1$ are from $0.25\,N_d(\parvec\mu_{11},\,2\mat I) + 0.75\,N_d(\parvec\mu_{12},\,\mat I)$, while the observations from class-$2$ are from $0.75\,N_d(\parvec\mu_{21},\,\mat I) + 0.25\,N_d(\parvec\mu_{22},\,2\mat I)$. We use $\parvec\mu_{11}=\eta(0,1,0,1,\ldots)^\top$, $\parvec\mu_{12}=\eta(0,-1,0,-1,\ldots)^\top$, $\parvec\mu_{21}=\eta(1,0,1,0,\ldots)^\top$ and $\parvec\mu_{22}=\eta(-1,0,-1,0,\ldots)^\top$, with $\eta = 0.3$.
\end{ex}

For each example, we consider two different setups: (i) fixing $d=100$, we vary $n=1000,2000$ and $4000$, and (ii) fixing $n=1000$, we vary $d=50,150$ and $250$; generating an equal number of observations from the two classes. For each scenario, we compute the misclassification rates of NN-MADD$_{\rm sc}$ with $k_1=k_2=50,100,200$ on a test sample of size $5000$ (equal number of observations from the two classes). The averages of these misclassification rates over $25$ replications are reported in Table~\ref{tab:madd_MADDsc}, along with the corresponding standard errors. The misclassification rates of NN-MADD are reported for comparison. To demonstrate the scalability of our proposed method, we also report the average computing times of the classifiers.

\begin{table}[t!]
\centering
\caption{Average misclassification rates (in $\%$) and average computing times (in seconds) of NN-MADD and NN-MADD$_{\rm sc}$ with different numbers of representative observations ($k_1=k_2$) in Examples~\ref{example1}~and~\ref{example2}. The standard errors of the misclassification rates are reported in a smaller font within parentheses.\label{tab:madd_MADDsc}}
\small
\setlength{\tabcolsep}{5pt}
\begin{adjustbox}{max width=\textwidth}
\begin{tabular}{l|rr|rr|rr|rr}
\toprule
& \multicolumn{2}{c|}{} 
& \multicolumn{6}{c}{NN-MADD$_{\rm sc}$} \\
\cmidrule(lr){4-9}
& \multicolumn{2}{c|}{NN-MADD} 
& \multicolumn{2}{c|}{$k_1=k_2=50$} 
& \multicolumn{2}{c|}{$k_1=k_2=100$} 
& \multicolumn{2}{c}{$k_1=k_2=200$} \\
\cmidrule(lr){2-9}
& \multicolumn{1}{l}{Error} & Time  
& \multicolumn{1}{l}{Error} & Time  
& \multicolumn{1}{l}{Error} & Time  
& \multicolumn{1}{l}{Error} & Time  \\
\midrule

\multicolumn{9}{c}{\textbf{Example~\ref{example1}}} \\
\midrule
\multicolumn{9}{l}{$\mathbf{d = 100}$} \\
\midrule
$n = 1000$ &
14.18 \se{0.17} & 76.43 &
14.39 \se{0.24} & 5.70 &
14.19 \se{0.16} & 11.83 &
14.24 \se{0.18} & 26.62 \\

$n = 2000$ &
14.82 \se{0.10} & 319.95 &
14.85 \se{0.14} & 14.05 &
14.64 \se{0.14} & 24.90 &
14.73 \se{0.12} & 56.87 \\

$n = 4000$ &
14.14 \se{0.13} & 1191.63 &
14.39 \se{0.11} & 43.51 &
14.22 \se{0.12} & 66.38 &
14.20 \se{0.13} & 125.46 \\

\midrule
\multicolumn{9}{l}{$\mathbf{n = 2000}$} \\
\midrule
$d = 50$ &
18.61 \se{0.11} & 331.50 &
18.79 \se{0.11} & 13.35 &
18.85 \se{0.14} & 23.73 &
18.66 \se{0.13} & 55.26 \\

$d = 150$ &
12.11 \se{0.12} & 342.81 &
12.10 \se{0.13} & 13.61 &
12.15 \se{0.13} & 24.54 &
12.37 \se{0.13} & 56.67 \\

$d = 250$ &
9.95 \se{0.09} & 353.13 &
9.98 \se{0.14} & 13.77 &
9.93 \se{0.14} & 25.31 &
10.16 \se{0.12} & 56.48 \\

\midrule
\multicolumn{9}{c}{\textbf{Example~\ref{example2}}} \\
\midrule

\multicolumn{9}{l}{$\mathbf{d = 100}$} \\
\midrule
$n = 1000$ &
20.23 \se{0.17} & 67.04 &
25.10 \se{0.24} & 7.06 &
22.60 \se{0.23} & 15.12 &
21.36 \se{0.22} & 33.35 \\

$n = 2000$ &
19.28 \se{0.13} & 339.54 &
24.65 \se{0.22} & 16.22 &
22.19 \se{0.18} & 30.11 &
20.88 \se{0.19} & 68.16 \\

$n = 4000$ &
18.44 \se{0.12} & 1251.36 &
23.15 \se{0.30} & 42.19 &
21.19 \se{0.14} & 69.61 &
19.97 \se{0.15} & 138.27 \\

\midrule
\multicolumn{9}{l}{$\mathbf{n = 2000}$} \\
\midrule
$d = 50$ &
30.92 \se{0.15} & 283.42 &
32.83 \se{0.29} & 15.91 &
32.03 \se{0.19} & 29.48 &
31.69 \se{0.21} & 67.45 \\

$d = 150$ &
12.33 \se{0.12} & 316.25 &
18.58 \se{0.23} & 16.34 &
15.65 \se{0.18} & 30.41 &
13.84 \se{0.14} & 68.22 \\

$d = 250$ &
5.25 \se{0.08} & 388.45 &
10.74 \se{0.17} & 16.38 &
8.16 \se{0.12} & 30.38 &
6.70 \se{0.12} & 68.79 \\

\bottomrule
\end{tabular}
\end{adjustbox}
\end{table}

Table~\ref{tab:madd_MADDsc} shows that the misclassification rates of NN-MADD$_{\rm sc}$ is very close to those of NN-MADD even with a small number of representative observations. On average, the misclassification rate of NN-MADD$_{\rm sc}$ remains within $1\%$ of that of NN-MADD classifier when $(k_1,k_2)$ is chosen appropriately. In Example~\ref{example1}, the misclassification rates are comparable for all choices of $(k_1,k_2)$, with a small edge at $k_1=k_2=100$; whereas in Example~\ref{example2}, NN-MADD$_{\rm sc}$ achieves the best performance with $k_1=k_2=200$. Moreover, the computing time of NN-MADD$_{\rm sc}$ is significantly smaller than that of NN-MADD, particularly for smaller values of $(k_1,k_2)$. Thus, the choice of $k_1,\ldots,k_J$ plays an important role in the performance as well as computational cost of the NN-MADD$_{\rm sc}$ classifier.

In essence, $k_1,\ldots,k_J$ are \emph{hyper-parameters} of our method, and a natural way to select them in a data-dependent way is via \emph{cross-validation} \citep{duda2001pattern,hastie2009elements}. In particular, for $V$-fold cross-validation with a set $\mathcal K$ of candidate hyper-parameters, we would randomly split the training data into $V$ equal parts; hold out one part as \emph{validation data} and use the remaining as \emph{training data} to construct NN-MADD$_{\rm sc}$ with $(k_1,\ldots,k_J) \in \mathcal K$; and compute the corresponding misclassification rate on the held out validation data. Finally, the cross-validation score would be calculated as the average of the $V$ misclassification rates, and the hyper-parameter leading to the lowest cross-validation score would be selected. However, with this usual cross-validation technique, for each fold and every choice of hyper-parameters, we would need to repeatedly select a set of representative observations using DPP and compute MADD$_{\rm sc}$ for all the observations in the training data, which would require heavy computations. To address this, we propose an \emph{incremental cross-validation} with significantly reduced computational cost. The particular structure of MADD$_{\rm sc}$ as well as DPP sampling plays an important role in the proposed strategy.

Suppose that $\mathcal K = \{\vec k_t := (k_{1t},\ldots,k_{Jt}): t=1,\ldots,T\}$ consists of $T$ candidate values for $(k_1,\ldots,k_J)$. Without loss of generality, we assume that $k_{j1} \le k_{j2} \le \cdots \le k_{jT}$ for $j=1,\ldots,J$. Thus, $\vec k_t \le \vec k_{t+1}$ for every $t=1,\ldots,T-1$, where the inequality holds coordinatewise. Now, if it was the case that the representative sets were ordered, i.e., $\sp X^{\ast}(\vec k_t) \subseteq \sp X^{\ast}(\vec k_{t+1})$ for every $t$, then
\begin{align}\label{eq:incremental_MADD_SC}
\rho_{\rm sc}^{(t+1)}(\vec u,\vec v) &= \frac{1}{|\sp X^{\ast}(\vec k_{t+1}) \setminus \{\vec u,\vec v\}|} \sum_{\vec z \in \sp X^{\ast}(\vec k_{t+1}) \setminus \{\vec u,\vec v\}} \Big|\|\vec u - \vec z\| - \|\vec v - \vec z\|\Big| \nonumber\\
&= \frac{1}{|\sp X^{\ast}(\vec k_{t+1}) \setminus \{\vec u,\vec v\}|} \left\{\sum_{\vec z \in \sp X^{\ast}(\vec k_t) \setminus \{\vec u,\vec v\}} \Big|\|\vec u - \vec z\| - \|\vec v - \vec z\|\Big| \right. \nonumber\\
&\kern40ex + \left.\sum_{\vec z \in \Delta\sp X^{\ast}(\vec k_{t+1}) \setminus \{\vec u,\vec v\}} \Big|\|\vec u - \vec z\| - \|\vec v - \vec z\|\Big|\right\} \nonumber\\
&= \frac{1}{|\sp X^{\ast}(\vec k_{t+1}) \setminus \{\vec u,\vec v\}|} \Bigg\{|\sp X^{\ast}(\vec k_t) \setminus \{\vec u,\vec v\}|\;\rho_{\rm sc}^{(t)}(\vec u,\vec v) \nonumber\\
&\kern40ex + \sum_{\vec z \in \Delta\sp X^{\ast}(\vec k_{t+1}) \setminus \{\vec u,\vec v\}} \Big|\|\vec u - \vec z\| - \|\vec v - \vec z\|\Big|\Bigg\},
\end{align}
where $\Delta\sp X^\ast(\vec k_{t+1}) = \sp X^\ast(\vec k_{t+1}) \setminus \sp X^\ast(\vec k_t)$ is the collection of additional observations in $\sp X^\ast(\vec k_{t+1})$ which are not present in $\sp X^{\ast}(\vec k_t)$. Thus, during cross-validation, we could calculate the MADD$_{\rm sc}$ values incrementally, starting with $\vec k_1$ and gradually increasing to $\vec k_T$, without the need to re-compute them repeatedly for each hyper-parameter.

Unfortunately, if we select usual DPP samples, then this ordering among the generated samples cannot be ensured. However, a careful look at the naive DPP sampling algorithm reveals two important steps: (i) selection of a potential direction set $V$ in which the eigenvector $\vec v$ is included with probability $\lambda/(\lambda+1)$ ($\lambda$ is the corresponding eigenvalue) and (ii) selection of the $i$-th observation with probability $\mathrm{pr}(i) = |V|^{-1} \sum_{\vec v \in V} (\vec v^\top \vec e_i)^2$, adjusting the set $V$ after each selection (cf. Algorithm~\ref{alg:DPP} in the Appendix). For our proposed cross-validation technique, we adjust these two steps. In particular, (i) we construct the set $V$ with the eigenvectors corresponding to the top $k$ values of $\lambda/(\lambda+1)$ (equivalently, the top $k$ values of $\lambda$) and (ii) we select the observations with the top $k$ values of $\mathrm{pr}(i) = |V|^{-1} \sum_{\vec v \in V} (\vec v^\top \vec e_i)^2$, adjusting the set $V$ after each step. This modified algorithm is given in the Appendix (see Algorithm~\ref{alg:k-DPP_modified}). As a consequence of this approach, the corresponding selected sets become naturally ordered.

Using the modified representative set selection technique described above, we propose the following incremental cross-validation technique. For each fold, we use the modified technique to select a subset of size $k_{jT}$ from $\sp X_{j,{\rm tr}}$ -- the training sample corresponding to Class-$j$, $j=1,\ldots,J$. If the selected sample is $\sp X^{\ast}_{j,{\rm tr}} = \{\vec X_{j1}^{\ast},\ldots,\vec X_{j,k_{jT}}^{\ast}\}$, then we define $\sp X^{\ast}_{j,{\rm tr}}(k_{jt}) = \{\vec X_{j1}^{\ast},\ldots,\vec X_{j,k_{jt}}^{\ast}\}$, $j=1,\ldots,J$, and $\sp X^{\ast}(\vec k_t) = \cup_{j=1}^J \sp X_{j,{\rm tr}}^{\ast}(k_{jt})$, $t=1,\ldots,T$. Consequently, $\Delta\sp X^{\ast}(\vec k_{t+1}) = \cup_{j=1}^J \big\{\vec X_{j,k_{jt}+1}^{\ast},\ldots,\vec X_{j,k_{j,t+1}}^{\ast}\big\}$. Using these, we calculate MADD$_{\rm sc}$ values sequentially according to Eq.~\eqref{eq:incremental_MADD_SC}. Unlike the usual cross-validation, here we select the representative set once and compute the MADD$_{\rm sc}$ values sequentially, saving a lot of computations.

\begin{table}[t!]
\centering
\caption{Misclassification rates (in \%) of the NN-MADD$_{\rm sc}$ classifier in Examples~\ref{example1} and \ref{example2} with the hyper-parameters chosen using the usual cross-validation and the proposed incremental cross-validation. The best observed misclassification rates over the range of hyper-parameters are also shown. The reported numbers are averages based on $25$ simulation runs. The corresponding standard errors are reported in a smaller font within parentheses.\label{tab:madd_MADDsc_CV}} 
\small
\begin{tabular}{rr|rrr|rrr}
\toprule
& & \multicolumn{3}{c|}{\textbf{Example~\ref{example1}}} 
& \multicolumn{3}{c}{\textbf{Example~\ref{example2}}} \\ [3pt]
$d$ & $n$  
& \multicolumn{1}{l}{Best} & \multicolumn{1}{l}{Usual CV} & \multicolumn{1}{l|}{Proposed CV}
& \multicolumn{1}{l}{Best} & \multicolumn{1}{l}{Usual CV} & \multicolumn{1}{l}{Proposed CV}\\
\midrule
100 & 1000 & 13.63 \se{0.16} & 14.17 \se{0.20} & 14.11 \se{0.18} & 20.79 \se{0.20} & 21.18 \se{0.25} & 21.11 \se{0.21} \\
100 & 2000 & 14.29 \se{0.11} & 14.75 \se{0.16} & 14.78 \se{0.11} & 20.30 \se{0.13} & 20.53 \se{0.16} & 20.35 \se{0.18}\\
100 & 4000 & 13.86 \se{0.10} & 14.17 \se{0.12} & 14.38 \se{0.13}  & 19.41 \se{0.11} & 19.66 \se{0.17} & 19.77 \se{0.11}\\
\hline
50  & 2000 & 18.24 \se{0.11} & 18.64 \se{0.16} & 18.79 \se{0.11} & 31.39 \se{0.16} & 31.94 \se{0.21} & 31.65 \se{0.22}\\
150 & 2000 & 11.61 \se{0.13} & 12.03 \se{0.13} & 12.46 \se{0.14} & 13.26 \se{0.14} & 13.27 \se{0.14} & 13.30 \se{0.13}\\
250 & 2000 & 9.50 \se{0.09} & 10.12 \se{0.13} & 9.99 \se{0.15} & 5.82 \se{0.07} & 5.93 \se{0.14} & 5.87 \se{0.10}\\
\bottomrule
\end{tabular}
\end{table}

To demonstrate the effectiveness of the proposed incremental cross-validation strategy, in Table~\ref{tab:madd_MADDsc_CV} we report the average misclassification rates for the NN-MADD$_{\rm sc}$ classifier (based on 25 replications) in Examples~\ref{example1} and \ref{example2} with the hyper-parameters selected using the proposed method. For Class-$j$, we use $k_{jt} = 2^{t-2} \sqrt{d} (n_j/n) \log{n_j}$ for $t=1,\ldots,5$, which takes into account both the dimension of the data as well as the training sample size. The choice is motivated by our theoretical results (cf.\ Remark~\ref{remark:representative_sample_size}). To facilitate comparison, we also report the average misclassification rates of NN-MADD$_{\rm sc}$ with the hyper-parameters selected using the usual cross-validation, and the averages of the lowest observed misclassification rates over the choice of hyper-parameters. It can be observed that the results based on the usual cross-validation and the proposed cross-validation are very close, and both of them are close to the best possible results (the differences never exceed $1\%$). To show the computational gain of the proposed method, we report the average runtimes of NN-MADD$_{\rm sc}$ with the usual and the incremental cross-validation for different training sample sizes in Table~\ref{tab:Time_MADD_MADD_sc}, along with the runtimes of NN-MADD. NN-MADD$_{\rm sc}$ with cross-validation is faster than NN-MADD for all the sample sizes, and the difference increases rapidly with $n$. The proposed cross-validation is significantly faster than the usual cross-validation, and the gain becomes more prominent with increasing $n$. At $n = 4096$, the computing time of NN-MADD$_{\rm sc}$ with the incremental cross-validation is almost half of the computing time of NN-MADD$_{\rm sc}$ with usual cross-validation, and almost one-third of that of NN-MADD.

\begin{table}[t!]
\centering
\caption{Average computing times (in seconds) for classifying $5000$ test observations using NN-MADD and NN-MADD$_{\rm sc}$ classifiers with fixed dimension $d=100$ and varying training sample sizes $n$. For NN-MADD$_{\rm sc}$, we separately report the runtimes when the hyper-parameters are chosen using the usual cross-validation and the proposed incremental cross-validation.\label{tab:Time_MADD_MADD_sc}}
\begin{tabular}{l|rrrrrr}\toprule
Training sample size ($n$) &  $512$ & $1024$ & $2048$ & $4096$ \\ \hline
NN-MADD  & $19.89$ & $81.26$ & $357.108$ & $1249.31$  \\
NN-MADD$_{\rm sc}$ with usual CV & 16.61 & 82.20 & 278.49 & 885.57 \\
NN-MADD$_{\rm sc}$ with incremental CV & 14.68 & 72.94 & 161.94 & 472.25\\\bottomrule
\end{tabular}
\end{table}

\section{Simulation Studies}
\label{sec:Other classifiers}

Here, we study the empirical performance of the proposed NN-MADD$_{\rm sc}$ classifier on seven simulated examples involving binary classification problems. Examples~\ref{example1}~and~\ref{example2} were introduced in Section~\ref{sec:cross-validation}. Examples~\ref{example3}--\ref{example7} are described below.

\begin{ex}\label{example3}
The two classes have multivariate normal distributions $N_d(\vec 0,\mat I)$ and $N_d(\mu\,\vec 1,\mat I)$, which differ only in their locations. We take $\mu = 0.4$.
\end{ex}

\begin{ex}\label{example4}
The two classes have multivariate normal distributions $N_d(\vec 0,\mat I)$ and $N_d(\vec 0,\sigma^2\,\mat I)$ differing only in their scales. We use $\sigma^2 = 1.5$.
\end{ex}

\begin{ex}\label{example5}
The two classes have multivariate normal distributions $N_d(\vec 0,\mat I)$ and $N_d(\mu\,\vec 1,\sigma^2\,\mat I)$ having different locations as well as scales. We take $\mu = 0.1$ and $\sigma^2 = 1.5$.
\end{ex}

\begin{ex}\label{example6}
Class-$1$ has distribution $0.5\,N_d(\parvec\mu_{11},\mat I) + 0.5\,N_d(\parvec\mu_{12},\mat I)$, while Class-$2$ has distribution $0.5\,N_d(\parvec\mu_{21},\mat I) + 0.5\,N_d(\parvec\mu_{22},\mat I)$, where $\parvec\mu_{11} = \vec 0$, $\parvec\mu_{12} = \eta(1,1,0,0,\ldots)^\top$, $\parvec\mu_{21} = \eta(1,0,0,0,\ldots)^\top$, and $\parvec\mu_{22} = \eta(0,1,0,0,\ldots)^\top$ with $\delta=3$. Here, only the first two coordinates of the data contain information regarding class separability, while the other coordinates are noise.
\end{ex}

\begin{ex}\label{example7}
Both the classes are equal mixtures of four multivariate normal distributions, each with dispersion matrix $\mat I$ and different mean vectors. For Class-$1$, the mean vectors are $\eta(0,1,0,1,\ldots)^\top$, $-\eta(0,1,0,1,\ldots)^\top$, $2\eta(1,1,1,\ldots)^\top$, and $-2\eta(1,1,1,\ldots)^\top$, while for Class-$2$ those are $\eta(1,0,1,0,\ldots)^\top$, $-\eta(1,0,1,0,\ldots)^\top$, $2\eta(1,-1,1,-1,\ldots)^\top$, and $-2\eta(1,-1,1,-1,\ldots)^\top$. We use $\eta=0.3$.
\end{ex}

For our study, we fix the dimension $d=100$ and vary the training sample size $n=1000, 2000, 4000$, with an equal number of observations from the two classes. The misclassification rates are calculated based on $5000$ test observations ($2500$ from each class). As before, we repeat each experiment $25$ times and report the average misclassification rates along with the corresponding standard errors. For the proposed NN-MADD$_{\rm sc}$ classifier, the hyper-parameters are chosen using the incremental cross-validation technique discussed in Section~\ref{sec:cross-validation} with $k_{jt} = 2^{t-3} \sqrt{d}(n_j/n)\log(n_j)$, $t=1,\ldots,5$, $j=1,\ldots,J$. The codes for implementation of our method are available at \url{https://github.com/Agy-code/Scalable-MADD}.

We perform two different sets of comparisons. First, we compare the performance of NN-MADD$_{\rm sc}$ with NN-MADD, the results of which are reported in Appendix~\ref{appendix: Simulation} (see Table~\ref{tab:appendix_simulation_MADD_sc}). The results show that the difference between the average misclassification rates of NN-MADD and NN-MADD$_{\rm sc}$ is never more than $1\%$ across all the simulation setups. As already noted, NN-MADD$_{\rm sc}$ is much faster compared to NN-MADD. The results assert the effectiveness of our proposed classifier as a replacement for NN-MADD when the sample size is large. 

In Table~\ref{tab:Compare}, we compare the performance of NN-MADD$_{\rm sc}$ with some state-of-the-art classifiers from the literature. In particular, we consider the Euclidean-distance based nearest neighbor classifier (NN), logistic regression with elastic-net regularization (GLMNET), classification tree (CART), random forest (RF), linear support vector machine (LSVM), and non-linear support vector machine with the radial basis function (NLSVM). All these classifiers are implemented in \texttt{R}. Specifically, we use the \texttt{class} library for NN classifier. For GLMNET, we use the \texttt{glmnet} library, with the tuning parameters $\alpha$ and $\lambda$ selected using cross-validation. CART is implemented using \texttt{rpart} library, with the complexity parameter selected as the value minimizing the internal cross-validation error before pruning. For RF, we use the \texttt{randomForest} package, with the splitting parameter \texttt{mtry} chosen by minimizing the out-of-bag error using the \texttt{tuneRF} function. For both LSVM and NLSVM, we use the \texttt{caret} package, which selects the hyper-parameters using cross-validation over a specified tuning grid.

\begin{table}[ht]
\centering
\caption{Average misclassification rates (in \%) of different classifiers in Examples~\ref{example1}--\ref{example7} with $d=100$ and varying $n$. The corresponding standard errors are reported in a smaller font within parentheses. The underlined numbers represent the best result in each simulation setup.\label{tab:Compare}}

\fontsize{9}{11}\selectfont
\setlength{\tabcolsep}{3pt}
\begin{adjustbox}{max width=\textwidth}
\begin{tabular}{l|r|rrrrrrr}
\toprule
Example & $n$
& \multicolumn{1}{l}{NN} & \multicolumn{1}{l}{GLMNET} & \multicolumn{1}{l}{CART} & \multicolumn{1}{l}{RF} & \multicolumn{1}{l}{LSVM} & \multicolumn{1}{l}{NLSVM} & \multicolumn{1}{l}{NN-MADD$_{\rm sc}$} \\
\midrule
\multirow{3}{*}{Example~\ref{example1}} 
& 1000 & 37.00 \se{0.19} & 27.97 \se{0.16} & 44.39 \se{0.24} & 28.63 \se{0.17} & 31.14 \se{0.19} & 24.62 \se{0.40} &  \best{21.27} \se{0.20} \\
& 2000 & 36.32 \se{0.15} & 26.98 \se{0.09} & 43.63 \se{0.18} & 27.47 \se{0.13} & 28.84 \se{0.14} & 24.02 \se{0.27} &  \best{20.88} \se{0.17} \\
& 4000 & 36.04 \se{0.18} & 26.79 \se{0.13} & 43.06 \se{0.16} & 26.66 \se{0.15} & 27.68 \se{0.15} & 22.11 \se{0.13} &  \best{19.55} \se{0.11} \\
\midrule

\multirow{3}{*}{Example~\ref{example2}} 
& 1000 & 50.92 \se{0.02} & 49.66 \se{0.15} & 36.31 \se{0.30} & 19.07 \se{0.10} & 47.88 \se{0.13} & 16.67 \se{0.09} &  \best{14.39} \se{0.21} \\
& 2000 & 51.10 \se{0.03} & 49.88 \se{0.18} & 35.78 \se{0.27} & 18.71 \se{0.12} & 48.09 \se{0.29} & 17.03 \se{0.10} &  \best{14.82} \se{0.14} \\
& 4000 & 51.20 \se{0.03} & 50.01 \se{0.15} & 33.34 \se{0.22} & 17.88 \se{0.11} & 47.33 \se{0.37} & 16.57 \se{0.08} &  \best{14.28} \se{0.10} \\
\midrule

\multirow{3}{*}{Example~\ref{example3}} 
& 1000 & 14.94 \se{0.24} &  \best{2.52} \se{0.06} & 35.61 \se{0.16} & 4.46 \se{0.09} & 3.51 \se{0.06} & 2.69 \se{0.05} & 5.68 \se{0.09} \\
& 2000 & 13.90 \se{0.16} &  \best{2.43} \se{0.05} & 34.41 \se{0.15} & 3.89 \se{0.09} & 3.01 \se{0.05} & 2.62 \se{0.04} & 5.39 \se{0.07} \\
& 4000 & 13.18 \se{0.14} &  \best{2.33} \se{0.04} & 33.78 \se{0.17} & 3.56 \se{0.05} & 2.71 \se{0.05} & 2.46 \se{0.05} & 4.96 \se{0.07} \\
\midrule

\multirow{3}{*}{Example~\ref{example4}} 
& 1000 & 49.20 \se{0.04} & 49.94 \se{0.14} & 41.39 \se{0.24} & 19.42 \se{0.31} & 49.47 \se{0.15} &  \best{8.73} \se{0.09} & 12.00 \se{0.19} \\
& 2000 & 49.25 \se{0.05} & 49.94 \se{0.14} & 39.71 \se{0.25} & 17.18 \se{0.27} & 49.42 \se{0.17} &  \best{8.48} \se{0.10} & 11.84 \se{0.17} \\
& 4000 & 49.27 \se{0.04} & 50.07 \se{0.14} & 38.60 \se{0.22} & 15.16 \se{0.21} & 49.46 \se{0.13} &  \best{8.23} \se{0.08} & 11.91 \se{0.14} \\
\midrule

\multirow{3}{*}{Example~\ref{example5}} 
& 1000 & 48.94 \se{0.04} & 35.97 \se{0.19} & 39.10 \se{0.24} & 16.79 \se{0.34} & 36.38 \se{0.19} &  \best{7.72} \se{0.09} & 11.39 \se{0.15} \\
& 2000 & 48.98 \se{0.05} & 34.64 \se{0.12} & 38.40 \se{0.21} & 15.62 \se{0.23} & 34.87 \se{0.15} &  \best{7.46} \se{0.07} & 11.43 \se{0.15} \\
& 4000 & 49.04 \se{0.04} & 33.65 \se{0.12} & 36.92 \se{0.20} & 14.33 \se{0.24} & 33.69 \se{0.14} &  \best{7.18} \se{0.08} & 11.19 \se{0.11} \\
\midrule

\multirow{3}{*}{Example~\ref{example6}} 
& 1000 & 39.59 \se{0.11} & 50.06 \se{0.13} & 47.30 \se{1.33} & 43.98 \se{1.06} & 50.08 \se{0.10} & 45.08 \se{0.19} &  \best{28.05} \se{0.19} \\
& 2000 & 38.47 \se{0.17} & 50.04 \se{0.12} & 45.63 \se{2.04} & 39.73 \se{1.31} & 50.09 \se{0.15} & 40.18 \se{0.14} &  \best{26.74} \se{0.18} \\
& 4000 & 37.69 \se{0.17} & 50.13 \se{0.14} & 46.48 \se{1.60} & 36.39 \se{1.30} & 50.16 \se{0.18} & 35.07 \se{0.15} &  \best{26.23} \se{0.16} \\
\midrule

\multirow{3}{*}{Example~\ref{example7}} 
& 1000 & 25.41 \se{0.16} & 50.10 \se{0.18} & 42.42 \se{0.21} & 23.54 \se{0.11} & 50.08 \se{0.11} & 19.92 \se{0.13} &  \best{12.39} \se{0.15} \\
& 2000 & 24.48 \se{0.13} & 50.19 \se{0.14} & 40.52 \se{0.19} & 22.52 \se{0.13} & 50.03 \se{0.11} & 17.51 \se{0.12} &  \best{12.36} \se{0.07} \\
& 4000 & 23.78 \se{0.16} & 50.07 \se{0.15} & 39.50 \se{0.11} & 21.84 \se{0.12} & 50.14 \se{0.18} & 15.95 \se{0.14} &  \best{11.70} \se{0.11} \\
\bottomrule
\end{tabular}
\end{adjustbox}
\end{table}

In Example~\ref{example1}, where the underlying distributions differ only in their shapes, NN-MADD$_{\rm sc}$ has the lowest misclassification rates, followed by NLSVM, RF and GLMNET. NN-MADD$_{\rm sc}$ also achieves the best performance in Examples~\ref{example2}, \ref{example6} and \ref{example7} involving mixture distributions. In Example~\ref{example2}, NLSVM and RF also perform well, but all other classifiers have more than $35\%$ misclassification rates. In Example~\ref{example6}, while the misclassification rate of NN-MADD$_{\rm sc}$ is close to $28\%$, no other classifier barring NN has misclassification rates lower than $40\%$. In Example~\ref{example7} also, while NN-MADD$_{\rm sc}$ has misclassification rates close to $12\%$, all classifiers except NLSVM have more than $20\%$ misclassification rates. Example~\ref{example3} involves a location problem. Not surprisingly, GLMNET yields the lowest misclassification rate in this problem. NN-MADD$_{\rm sc}$ achieves close to $5\%$ misclassification in this example. In fact, all the classifiers except NN and CART perform well in Example~\ref{example3}. However, in Example~\ref{example4} -- which corresponds to a scale problem -- NLSVM achieves the minimum misclassification rate followed by NN-MADD$_{\rm sc}$ and RF. We see a similar result in Example~\ref{example5}, where we have a location-scale problem. In both these examples, all the classifiers except NLSVM, NN-MADD$_{\rm sc}$ and RF misclassify more than $40\%$ of the observations. Overall, the simulation results demonstrate that the proposed NN-MADD$_{\rm sc}$ classifier consistently delivers competitive and often superior classification performance across a wide range of scenarios.

\section{Handling Massive Datasets}
\label{sec:RFF}
We have already demonstrated that our proposed method is capable of handling datasets with moderate to large sample sizes. But, when the training sample size becomes very large (e.g., $n>5000$), the proposed method can become heavily time consuming. To understand this, note that our method has two important components, viz., construction of the similarity matrices $\mat L^{(j)}$ (of size $n_j \times n_j$), $j=1,\ldots,J$ and generation of DPP samples using these similarity matrices. The second step requires eigen-decomposition of the similarity matrix, which is usually quite fast for moderate to large $n_j$ values. But, for larger $n_j$ values (e.g., $n_j > 2500$), the eigen-decomposition step requires heavy computations. Moreover, storage of the huge similarity matrices can also become challenging in such cases. One way to mitigate this issue is to use approximate methods, e.g., Nystr{\"o}m's approximation \citep{williams2000using} or randomized SVD \citep{halko2011finding}, which approximate the leading eigenvalues and eigenvectors without performing a full decomposition. However, these methods often require explicit access to (or sampling from) the similarity matrix, which can still be expensive.

To address the issue, we take an alternative approach based on the Random Fourier Features (RFF) framework introduced by \citet{rahimi2007random}. The idea is based on \cite{bochner1933monotone}'s representation of a continuous shift-invariant kernel $l(\vec x,\vec y)=l(\vec x-\vec y)$ as the Fourier transformation of a symmetric probability distribution $p$: $l(\vec x-\vec y)=\mathbb{E}_{\vec w \sim p}[\cos(\vec w^\top (\vec x-\vec y))]$. The kernel $l$ and the probability distribution $p$ are uniquely determined by each other. In particular, for $l(\vec x-\vec y) = \exp\{-\|\vec x-\vec y\|^2/(2\sigma^2)\}$, the corresponding probability distribution is $N(\vec 0,\sigma^{-2}\mat I_d)$. Due to the representation of the kernel as an expectation, it can be approximated using Monte Carlo. Consequently, the similarity matrix $\mat L^{(j)}$ can be well approximated using Monte Carlo as follows. Let $\vec w_1,\ldots,\vec w_D \overset{\iid}{\sim} N_d(\vec 0,\sigma_j^{-2}\mat I_d)$. Define the RFF map $\mathsf{r}(\vec x)= \big(\cos(\vec w_1^\top x),\sin(\vec w_1^\top x),\ldots,\cos(\vec w_D^\top x),\sin(\vec w_D^\top x)\big)^\top/\sqrt{D}$. Form the RFF matrix $\mat R_j = [\mathsf{r}(\vec X_{j1})^\top,\ldots,\mathsf{r}(\vec X_{jn_j})^\top]^\top$. The similarity matrix $\mat L^{(j)}$ can be approximated as $\mat L^{(j)} \approx \mat R_j \mat R_j^\top$ (see \citealt{rahimi2007random} for details).

It is well known that the nonzero eigenvalues of $\mat R_j \mat R_j^\top$ coincide with those of $\mat R_j^\top \mat R_j$. Moreover, the corresponding eigenvectors of $\mat R_j\mat R_j^\top$ can be obtained from the eigenvectors of $\mat R_j^\top\mat R_j$ (\citealt{golub2013matrix}). Consequently, eigen-decomposition of the $n_j \times n_j$ matrix $\mat L^{(j)} \approx \mat R_j\mat R_j^\top$ can be carried out by performing eigen-decomposition of the $2D \times 2D$ matrix $\mat R_j^\top \mat R_j$. In fact, it is enough to perform the singular value decomposition of the $n_j \times 2D$ matrix $\mat R_j$ (\citealt{golub2013matrix}). This approach also helps reduce the memory complexity since we do not need to compute or store the huge similarity matrices $\mat L^{(j)}$, $j=1,\ldots,J$. The complexity of this approach depends on the choice of $D$, which controls how good the approximations are. Good approximations are usually obtained with moderate values of $D$. In our implementations, we use $D=500$. The method is formalized in Algorithm~\ref{alg:SMNN_RFF}.

\begin{algorithm}[t!]
\caption{Nearest Neighbor Classification Using RFF-based MADD$_{\rm sc}$\label{alg:SMNN_RFF}}
	\begin{algorithmic}[1]
	\Require $\sp X_1,\ldots,\sp X_J$ -- Training sample, where $\sp X_j = \{\vec X_{j1},\ldots,\vec X_{jn_j}\}$. \newline
	$k_1,\ldots,k_J$ -- Number of representative samples to be selected. \newline
	$D$ -- Number of random Fourier features. \newline
	$\vec Z$ -- Test observation.
	
	\Ensure $\delta_{\rm sc}(\vec Z)$ -- The predicted class label for $\vec Z$.
	
	\bigskip
	
	\For {$j \in \{1,\dots,J\}$}
    \State Find $\sigma_j = \operatorname{median}\{\|\vec X_{ji} - \vec X_{ji^\prime}\|: 1 \le i < i^\prime \le n_j\}$.
    \State Generate $\vec w_{j1},\ldots,\vec w_{jD} \stackrel{\text{i.i.d.}}{\sim} N_d(\vec 0,\sigma_j^{-2}\mat I_d)$.
    \State Define the RFF map $\mathsf{r}_j(\vec x) = \big(\cos(\vec w_{j1}^\top \vec x),\sin(\vec w_{j1}^\top\vec x),\ldots,\cos(\vec w_{jD}^\top\vec x),\sin(\vec w_{jD}^\top\vec x)\big)$, $\vec x \in \mathbb R^d$.
	\State Construct the $n_j \times 2D$ RFF matrix $\mat R_j = [\mathsf{r}_j(\vec X_{j1})^\top,\ldots,\mathsf{r}_j(\vec X_{jn_j})^\top]^\top$.
	\State Select a subset $\{\pi_1,\ldots,\pi_{k_j}\}$ of $\{1,\ldots,n_j\}$ using the approximate $k_j$-DPP (Algorithm~\ref{alg:k-DPP}) with the matrix $\mat R_j$. 
    \State The representative observations from class-$j$ are $\sp X_j^\ast = \{\vec X_{j\pi_1},\ldots,\vec X_{j\pi_{k_j}}\}$.
	\EndFor
	\State Define $\sp X^\ast = \sp X_1^\ast \cup \cdots \cup \sp X_J^\ast$ to be the representative set.
	\State Compute $ \rho_{\rm sc}(\vec Z,\vec X_{ji})$, $j=1,\ldots,J$, $i=1,\ldots,n_j$ using \eqref{eq:MADD_scalable}.
	\State For $j=1,\ldots,J$, compute $\rho_{\rm sc}(\vec Z,\sp X_j) = \min_{i=1,\ldots,n_j} \rho_{\rm sc}(\vec Z,\vec X_{ji})$.
	\State Assign $\delta_{sc}(\vec Z) = \argmin_{j=1,\ldots,J} \rho_{\rm sc}(\vec Z,\sp X_j)$.
	\end{algorithmic}
\end{algorithm}

\begin{remark}
There are alternative ways of constructing RFF maps which also ensure that $\mat L^{(j)}$ can be approximated by $\mat R_j\mat R_j^\top$. The essential requirement is that $l(\vec x-\vec y) \approx \mathsf{r}(\vec x)^\top\mathsf{r}(\vec y)$. For instance, one can use $\mathsf{r}(\vec x) = \sqrt{2/D}\big(\cos(\vec w_1^\top\vec x+b_1),\ldots,\cos(\vec w_D^\top\vec x+b_D)\big)^\top$, where $\vec w_1,\ldots,\vec w_D \overset{\iid}{\sim} N_d(\vec 0,\sigma_j^{-2}\mat I_d)$ and $b_1,\ldots,b_D \overset{\iid}{\sim} U(0,2\pi)$ \citep[cf.][]{rahimi2007random}. With this choice, the size of the matrix $\mat R_j$ is further reduced to $n_j \times D$ instead of $n_j \times 2D$.  
\end{remark}

To check the usefulness of the RFF-based approach, we compute the relative errors $\|\mat L^{(j)} - \mat R_j\mat R_j^\top\|_{\rm F}/\|\mat L^{(j)}\|_{\rm F}$, $j=1,\ldots,J$ with $n=7000$ and $d=500$ for the different simulated examples considered in Section~\ref{sec:Other classifiers}. In all our simulation setups, the relative error never exceed $4\%$ for any of the classes. Also, to demonstrate the utility of the RFF-based approach, we report the computing times of NN-MADD$_{\rm sc}$ with and without the RFF approximation step for different values of $n$ and $d$ in Table~\ref{tab:Time_MADD_MADD_sc_RFF}. From the table, we observe that the reduction in computing times for the RFF approximation compared to the usual NN-MADD$_{\rm sc}$ ranges between 3X and 40X, with larger reductions observed for larger values of $n$. In particular, with $n=25000$, the RFF-based approximation takes less than $25$ minutes to complete, compared to more than $15$ hours for the usual NN-MADD$_{\rm sc}$.

\begin{table}[t!]
\centering
\caption{Average computing times (in seconds) for classifying $5000$ test observations using NN-MADD, and NN-MADD$_{\rm sc}$ with and without the RFF approximation for varying sample sizes.\label{tab:Time_MADD_MADD_sc_RFF}}
\begin{tabular}{l|rrrrr}\toprule
Training sample size ($n$) &  $5000$ & $10000$ & $15000$ & $20000$ & $25000$ \\ \hline
NN-MADD & 1521.38 & 6443.03 & 16084.28 &  78819.78 & xxx \\
NN-MADD$_{\rm sc}$ without RFF approximation   & 565.56  & 3847.29 & 12138.01 & 28262.53 & 54650.38 \\
NN-MADD$_{\rm sc}$ with RFF approximation  & 206.31 & 322.70 & 568.52 & 958.64 & 1468.68 \\\bottomrule
\end{tabular}

xxx The NN-MADD classifier failed due to insufficient memory.
\end{table}

To further demonstrate the empirical performance of the RFF-based NN-MADD$_{\rm sc}$ classifier, we revisit Examples~\ref{example1}--\ref{example7} from Section~\ref{sec:Other classifiers} with much larger sample size and dimension. We use $d=500$ and $n=7000$, with an equal number of observations from the two classes, while the test sample size is kept fixed at $2500$ observations per class. For NN-MADD$_{\rm sc}$, we choose the number of representative samples using the incremental cross-validation technique described in Section~\ref{sec:cross-validation}. For the other methods, the hyper-parameters are selected as described in Section~\ref{sec:Other classifiers}. To have a reasonable comparison among the methods, the parameter values in the examples are chosen differently than in Section~\ref{sec:Other classifiers}. In particular, we choose $\delta=0.25$ in Example~\ref{example2}, the location parameter $\mu=0.2$ in Example~\ref{example3}, the scale parameter $\sigma^2=1.2$ in Example~\ref{example4}, the location and scale parameters $\mu=0.1$, $\sigma^2=1.25$ in Example~\ref{example5}, $\delta=5$ in Example~\ref{example6}, and $\delta=0.125$ in Example~\ref{example7}. The average misclassification rates over $25$ replications, along with the corresponding standard errors, of the different classifiers are reported in Table~\ref{tab:RFF}.

\begin{table}[b!]
\centering
\caption{Average misclassification rates (in \%) of different classifiers in Examples~\ref{example1}--\ref{example7} with $n=7000$ and $d=500$. The corresponding standard errors are reported in a smaller font within parentheses. The underlined numbers represent the best result in each simulation setup. \label{tab:RFF}} 

\fontsize{9}{11}\selectfont
\begin{tabular}{c|rrrrrrrr}
\toprule
Example & \multicolumn{1}{l}{NN} & \multicolumn{1}{l}{GLMNET} & \multicolumn{1}{l}{CART} & \multicolumn{1}{l}{RF} & \multicolumn{1}{l}{LSVM} & \multicolumn{1}{l}{NLSVM} & \multicolumn{1}{l}{NN-MADD$_{\rm sc}$} \\

\midrule

\ref{example1} & 50.12 \se{0.01} & 49.68 \se{0.12} & 33.08 \se{0.19} & 16.91 \se{0.08} & 47.89 \se{0.14} & 15.45 \se{0.08} &  \best{7.54} \se{0.12} \\ 

\ref{example2} & 30.04  \se{0.16} & 24.68  \se{0.14} & 44.07 \se{0.14} & 22.21 \se{0.14} & 25.27 \se{0.18} & 19.28 \se{0.96} &  \best{4.47} \se{0.08} \\ 

\ref{example3} & 22.10 \se{0.14} &  \best{1.31} \se{0.03} & 41.28 \se{0.15} & 4.22 \se{0.07} & 2.93 \se{0.06} & 1.57 \se{0.04} & 7.65 \se{0.09} \\ 

\ref{example4} & 49.93 \se{0.01} & 50.08 \se{0.10} & 45.10 \se{0.17} & 22.10 \se{0.17} & 50.07 \se{0.14} &  \best{8.12} \se{0.06} & 11.38 \se{0.10} \\

\ref{example5} & 49.95 \se{0.01} & 15.15 \se{0.12} & 41.23 \se{0.18} & 9.25 \se{0.13} & 16.52 \se{0.14} &  \best{4.93} \se{0.05} & 5.08 \se{0.07} \\ 

\ref{example6} & 29.37 \se{0.16} & 50.09 \se{0.15} & 49.46 \se{0.49} & 46.13 \se{0.58} & 49.89 \se{0.15} & 46.55 \se{0.14} &  \best{7.96} \se{0.10} \\

\ref{example7} & 35.63 \se{0.18} & 50.15 \se{0.12} & 48.47 \se{0.18} & 34.03 \se{0.14} & 49.85 \se{0.12} & 23.35 \se{0.10} &  \best{19.97} \se{0.12} \\ 

\bottomrule
\end{tabular}
\end{table}

The results are similar to the ones obtained in Section~\ref{sec:Other classifiers}. In Example~\ref{example1}, where the underlying populations differ only in their shapes, as well as in Examples~\ref{example2}, \ref{example6} and \ref{example7}, where the populations have mixture distributions, the NN-MADD$_{\rm sc}$ classifier has the best performance. For the location problem in Example~\ref{example3}, GLMNET has the best performance, whereas in the scale problem in Example~\ref{example4} and the location-scale problem in Example~\ref{example5}, NLSVM has the best performance. In these three examples, the performance of NN-MADD$_{\rm sc}$ is also comparable, especially in Examples~\ref{example4} and \ref{example5}. Overall, the proposed NN-MADD$_{\rm sc}$ classifier has a superior performance compared to its competitors in the simulations. Also, the similarity of the results in Tables~\ref{tab:Compare} and \ref{tab:RFF} further corroborates that the RFF-based approach produces performances similar to the usual NN-MADD$_{\rm sc}$, while allowing faster implementation in massive datasets.

\begin{remark}
The first step in all our algorithms is to compute and store the matrix of pairwise distances. Although this is feasible with small sample sizes, the storage can become prohibitive when the training sample size is large. To handle such cases, our algorithms can be modified by computing the pairwise distances on the fly, without actually storing them. Such a method reduces the memory-complexity, but at the expense of being more time consuming. In a given problem, one needs to make a trade-off between being either ``time-efficient" or ``memory-efficient", based on the size of the training data and available memory.
\end{remark}

\section{Scalable Versions of Generalized MADD}
\label{sec:gmadd}
So far we have focused on the NN-MADD classifier and proposed a scalable version for it. However, it is well-known that for high-dimensional problems, the NN-MADD classifier may fail if the underlying populations differ beyond their locations and scales \citep{roy2022generalizations}. The reason behind this failure lies in the fact that the high-dimensional behavior of MADD is completely determined by the first two moments of the underlying populations \citep{sarkar2019perfect,roy2022generalizations}. To address this problem, \cite{roy2022generalizations} used a generalized version of MADD (referred to as gMADD) defined as
\begin{align}\label{eq:gmadd}
\rho^{\phi,\gamma}(\vec u,\vec v) = \frac{1}{|\sp X \setminus \{\vec u, \vec v\}|} \sum_{\vec z \in \sp X \setminus \{\vec u,\vec v\}} \Big|h^{\phi,\gamma}(\vec u,\vec z) - h^{\phi,\gamma}(\vec v,\vec z)\Big|,
\end{align}
where, for vectors $\vec x = (x_1,\ldots,x_d)^\top$ and $\vec y = (y_1,\ldots,y_d)^\top$, $h^{\phi,\gamma}(\vec x,\vec y) = \phi\big(d^{-1}\sum_{i=1}^{d} \gamma(|x_i-y_i|^2)\big)$ for non-negative, continuous and monotonically increasing functions $\phi$ and $\gamma$ satisfying $\phi(0)=\gamma(0)=0$. \cite{roy2022generalizations} showed that if $\phi$ is one-to-one and $\gamma$ has non-constant, completely monotone derivative, then NN classifier based on gMADD (henceforth referred to as NN-gMADD) can achieve \emph{perfect classification} in high-dimensions. In particular, they used $\phi(t)=t$ and $\gamma_1(t)=1-e^{-t}$, $\gamma_2(t)=\log(1+t)$, and $\gamma_3(t)=\sqrt{t}/2$, among which the first choice works even when the moments of the underlying distributions do not exist.

Similar to NN-MADD, NN-gMADD also requires heavy computations and may become infeasible to use for large datasets. To address this, we define a scalable version of gMADD (referred to as gMADD$_{\rm sc}$) as
\begin{equation}\label{eq:gMADD_sc}
\rho^{\phi,\gamma}_{\rm sc}(\vec u,\vec v) = \frac{1}{|\sp X^\ast \setminus \{\vec u, \vec v\}|} \sum_{\vec z \in \sp X^\ast \setminus \{\vec u,\vec v\}} \Big|h^{\phi,\gamma}(\vec u,\vec z) - h^{\phi,\gamma}(\vec v,\vec z)\Big|,
\end{equation}
where $\sp X^\ast$ is a representative set of training observations. The challenge here lies in choosing the representative set $\sp X^\ast$ appropriately. Ideally, our selected representative set should be such that it captures the diversity among the pairwise $h^{\phi,\gamma}$-distances. Therefore, we select a representative sample from $\sp X_j$ using $k_j$-DPP with the similarity matrix $\mat L^{(j)}_{\phi,\gamma} = \big(\big(l^{(j)}_{ii^\prime}(\phi,\gamma)\big)\big)$, where $l^{(j)}_{ii^\prime}(\phi,\gamma) = \exp\{-h^{\phi,\gamma}(\vec X_{ji},\vec X_{ji^\prime})/(2\sigma_j^2)\}$. The bandwidth parameter $\sigma_j$ is chosen using the median heuristic, i.e., 
$\sigma_j^2 = \operatorname{median}\{h^{\phi,\gamma}(\vec X_{ji},\vec X_{ji^\prime}): 1 \le i < i^\prime \le n_j\}$. Finally, the union of the representative samples from all the classes is taken to be the representative set $\sp X^\ast$. After choosing the representative set, the classification steps are similar to the ones used in NN-MADD$_{\rm sc}$. For the sake of completeness, the classification algorithm is presented in Appendix~\ref{appendix: Algorithms} (see Algorithm~\ref{alg:SGMNN}). We refer to the resulting classifier as the NN classifier based on gMADD$_{\rm sc}$ or NN-gMADD$_{\rm sc}$. For selection of the number of representative samples $k_1,\ldots,k_J$, we utilize an incremental cross-validation technique similar to the one defined in Section~\ref{sec:cross-validation}.

To demonstrate the advantage of NN-gMADD$_{\rm sc}$, we report its average runtimes in Table~\ref{tab:Time_gMADD_gMADD_sc}, along with those of NN-gMADD, with $d=100$ and varying training sample sizes. For both of these classifiers, we use $\phi(t)=t$ and $\gamma(t)=1-e^{-t}$ as in \cite{roy2022generalizations}. Clearly, the proposed scalable version has significantly reduced computing times compared to NN-gMADD. In particular, for $n=4096$, the computing time of NN-gMADD$_{\rm sc}$ is almost half of that of NN-gMADD.

We also perform some simulations to evaluate the performance of NN-gMADD$_{\rm sc}$. For this, we fix the dimension at $100$, and the training and test sample sizes at $2000$ and $5000$, respectively, with an equal number of observations from each class. We replicate the experiments $25$ times and report the average misclassification rates and the corresponding standard errors in Table~\ref{tab:gMADD}. In all these examples, the difference between the misclassification rates of NN-gMADD and NN-gMADD$_{\rm sc}$ never exceeds $1\%$.

\begin{table}[t!]
\centering
\caption{Average runtimes (in seconds) for classifying $5000$ test observations using NN-gMADD and NN-gMADD$_{\rm sc}$ classifiers with fixed dimension $d=100$ and varying training sample sizes $n$.\label{tab:Time_gMADD_gMADD_sc}}
\begin{tabular}{l|rrrr}\toprule
Training sample size ($n$) &  $512$ & $1024$ & $2048$ & $4096$ \\ \hline
NN-gMADD   & 30.13 & 107.84 & 393.95 & 1302.82 \\
NN-gMADD$_{\rm sc}$  & 20.83 & 82.08 & 260.77 & 672.86\\\bottomrule
\end{tabular}
\end{table}

\begin{ex}\label{example8}
The underlying distributions are $N_d(\vec 0_d, \Lambda_{d})$ and $N_d(\vec 0_d, \Gamma_{d})$, with $\Lambda_{d} = diag(\lambda_{1}, \dots, \lambda_{d})$ and $\Gamma_{d} = diag(\gamma_{1}, \dots, \gamma_{d})$. Here, $\lambda_{1} = \cdots = \lambda_{25} = 0.5$, $\lambda_{26} = \cdots = \lambda_{50} = 2$, and $\lambda_{51} = \cdots = \lambda_{d} = 1$, whereas $\gamma_{1} = \cdots = \gamma_{25} = 2$, $\gamma_{26} = \cdots = \gamma_{50} = 0.5$, and $\gamma_{51} = \cdots = \gamma_{d} = 1$.
\end{ex}

Here, the populations differ only in their covariance structures. Also, only the first $50$ coordinates are informative, while the rest are noise. Here, RF has the minimum misclassification rate, which is closely followed by NN-gMADD$_{\rm sc}$ and NLSVM. Since the distributions have the same locations and scales (in terms of trace of the covariance matrices), NN-MADD$_\mathrm{sc}$ has a significantly higher misclassification rate than NN-gMADD$_{\rm sc}$.

\begin{ex}\label{example9}
The two classes have i.i.d.\ marginals. For Class-$1$, the marginals are standard Cauchy random variables. For Class-$2$, the marginals are Cauchy with location $0.5$ and scale $1$.
\end{ex}

Here, the moments of the underlying distributions do not exist. RF has the best performance, followed by NN-gMADD$_\mathrm{sc}$. All other classifiers misclassify more than $30\%$ of the observations.

\begin{table}[b!]
\centering
\small
\caption{Average misclassification rates (in \%) of different classifiers in Examples~\ref{example8}--\ref{example11} with $n=2000$ and $d=100$. The corresponding standard errors are reported in a smaller font within parentheses. The underlined numbers represent the best result in each simulation setup.\label{tab:gMADD}}
\setlength{\tabcolsep}{3pt}
\begin{tabular}{c|rrrrrrrr}
\toprule
Ex & \multicolumn{1}{l}{NN} & \multicolumn{1}{l}{GLMNET} & \multicolumn{1}{l}{CART} & \multicolumn{1}{l}{RF} & \multicolumn{1}{l}{LSVM} & \multicolumn{1}{l}{NLSVM} & \multicolumn{1}{l}{NN-MADD$_{\rm sc}$}  & \multicolumn{1}{l}{NN-gMADD$_{\rm sc}$} \\
\midrule

\ref{example8} & 18.64 \se{0.22} & 49.97 \se{0.17} & 15.82 \se{0.14} & \best{0.28} \se{0.02} & 50.03 \se{0.17} & 0.95 \se{0.04} & 15.12 \se{0.09} & 1.22 \se{0.03} \\

\ref{example9} & 48.89 \se{0.13} & 44.95 \se{0.36} & 35.95 \se{0.18} & \best{8.64} \se{0.13} & 45.85 \se{0.27} & 42.82 \se{0.16} & 47.83 \se{0.16} & 14.58 \se{0.13} \\

\ref{example10} & 49.26 \se{0.05} & 49.90 \se{0.16} & 37.97 \se{0.14} & 9.04 \se{0.23} & 49.78 \se{0.16} & 33.12 \se{0.15} & 32.02 \se{0.16} & \best{5.36} \se{0.10} \\

\ref{example11} & 49.30 \se{0.12} & 50.14 \se{0.16} & 42.96 \se{0.17} & 20.49 \se{0.21} & 50.15 \se{0.17} & 49.85 \se{0.12} & 46.14 \se{0.16} & \best{19.17} \se{0.16} \\

\bottomrule
\end{tabular}
\end{table}

\begin{ex}\label{example10}
Class-$1$ has $N_d(\vec 0_d,3 \mat I_d)$ distribution. For Class-$2$, the marginal distributions are i.i.d.\ standard $t_3$. Here, the mean and the variance of the two classes are the same.
\end{ex}

Here, NN-gMADD$_{\rm sc}$ has the best result, followed by RF. All other classifiers has more than $30\%$ misclassification rates, with the linear methods having almost $50\%$ misclassifications.

\begin{ex}\label{example11}
The marginal distributions in both the classes are i.i.d. For Class-$1$, the marginals are $N(0,1)$, whereas for Class-$2$, the marginals are Laplace with location $0$ and scale $1/2$. Here also, the populations have the same mean and variance, and they differ only in their shapes.
\end{ex}

Here also, NN-gMADD$_{\rm sc}$ has the lowest misclassification rate, followed by RF. All other classifiers misclassify more than $40\%$ of the observations.

\begin{remark}
The simulation results show the utility of the proposed scalable version of NN-gMADD, particularly when the training sample size is large. However, unlike the NN-MADD$_{\rm sc}$ classifier, we cannot apply the RFF technique to use it with huge training datasets. The main difficulty lies in identifying the ambient distribution $p$ from which the random features need to be generated. Formulation of an RFF-type approach for NN-gMADD$_{\rm sc}$ can be an interesting problem for future research. 
\end{remark}

\section{Real Data Analysis}
\label{sec:Real_Data}
We analyze several benchmark datasets from the UCR Time Series Classification Archive (\url{https://www.cs.ucr.edu/~eamonn/time_series_data_2018/}) to further demonstrate the utility of the proposed classifier. For our analyses, we consider datasets that have at least $500$ samples (training and test combined), $50$ or more dimensions, and no more than $10$ classes. For each of these datasets, we randomly split the observations from each class in a $70:30$ ratio to constitute the training and test sets. The random split is replicated $25$ times, and we report the average misclassification rates of different classifiers across these $25$ replications, along with the corresponding standard errors. Among these datasets, $3$ have more than $5000$ training samples. For brevity, here we provide detailed analyses of $6$ datasets; all the $3$ with more than $5000$ training samples and $3$ with less than $5000$ training samples (see Table~\ref{tab:real}); showcasing the importance of the proposed classifier. The analyses of the remaining datasets are given in Appendix~\ref{appendix:Real}.

\subsection{Moderate-Sized Training Samples}
First, we consider $3$ datasets with less than $5000$ training samples. For these datasets, we also report the performance of NN-gMADD$_{\rm sc}$.
\medskip

The \textbf{Synthetic Control Chart} dataset consists of synthetic time-series generated using the procedures described by \cite{alcock1999time}. The data consists of $6$ different classes: Normal, Cyclic, Increasing Trend, Decreasing Trend, Upward Shift, and Downward Shift. Plots of the observations from the different classes are shown in Figure~\ref{fig:Controlchart}. For the classes Increasing Trend and Upward Shift, the location patterns are similar, but the variabilities are different. Similar traits are visible between the Downward Shift and Decreasing Trend classes. Again, the variability in the Normal class is significantly higher than those for Increasing Trend and Decreasing Trend. Due to the prominence of the scale differences among the classes, the performance of NN-MADD$_{\rm sc}$ classifier ($1.29\%$ misclassification) is much better than the NN classifier ($9.13\%$ misclassification). In this example, NN-MADD$_{\rm sc}$ has the lowest misclassification rate, followed by NLSVM and RF.

\begin{figure}[b!]
	\centering
    \footnotesize
	\begin{tabular}{c c c c c}
	~~Normal && ~~Increasing Trend && ~~Upward Shift \\
	\includegraphics[width=0.25\linewidth]{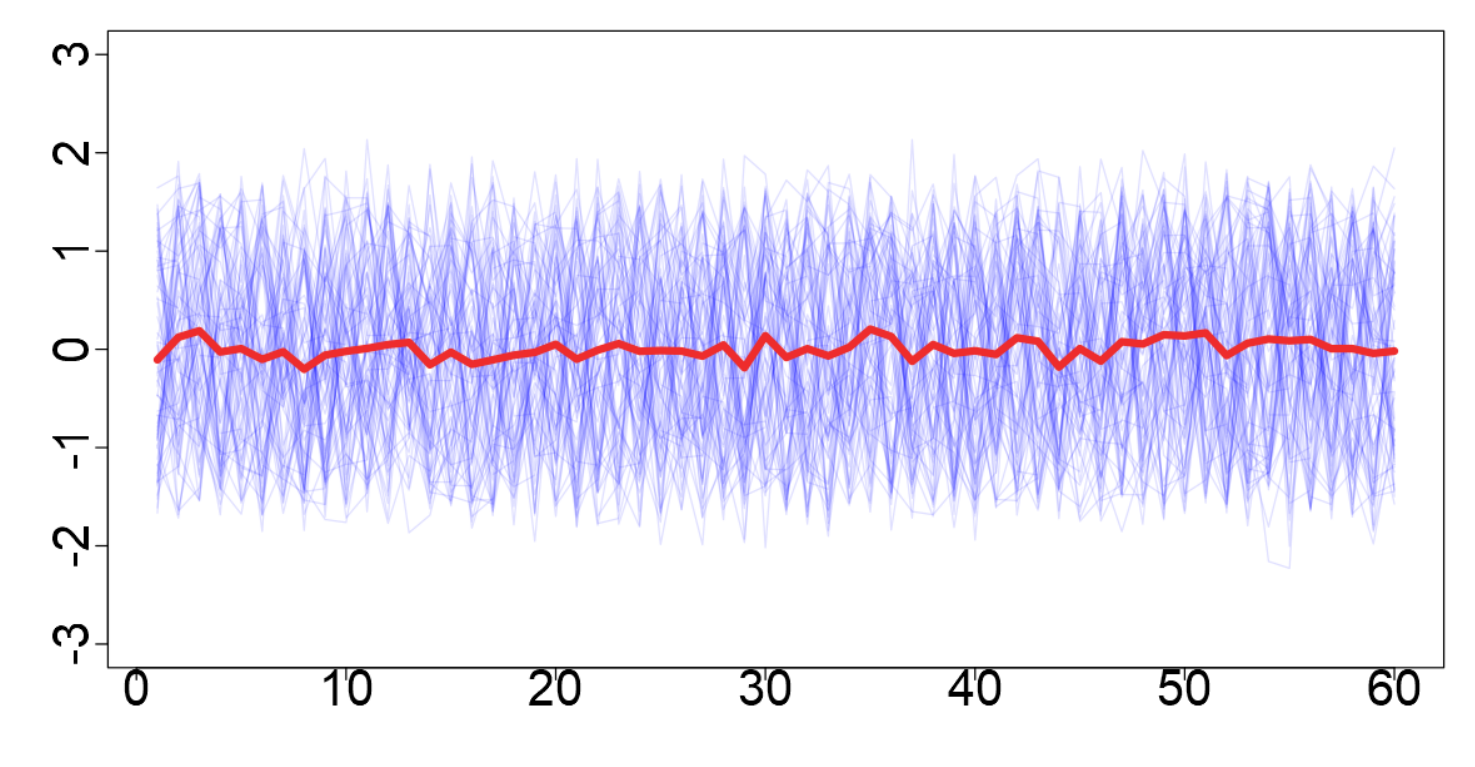} &&
	\includegraphics[width=0.25\linewidth]{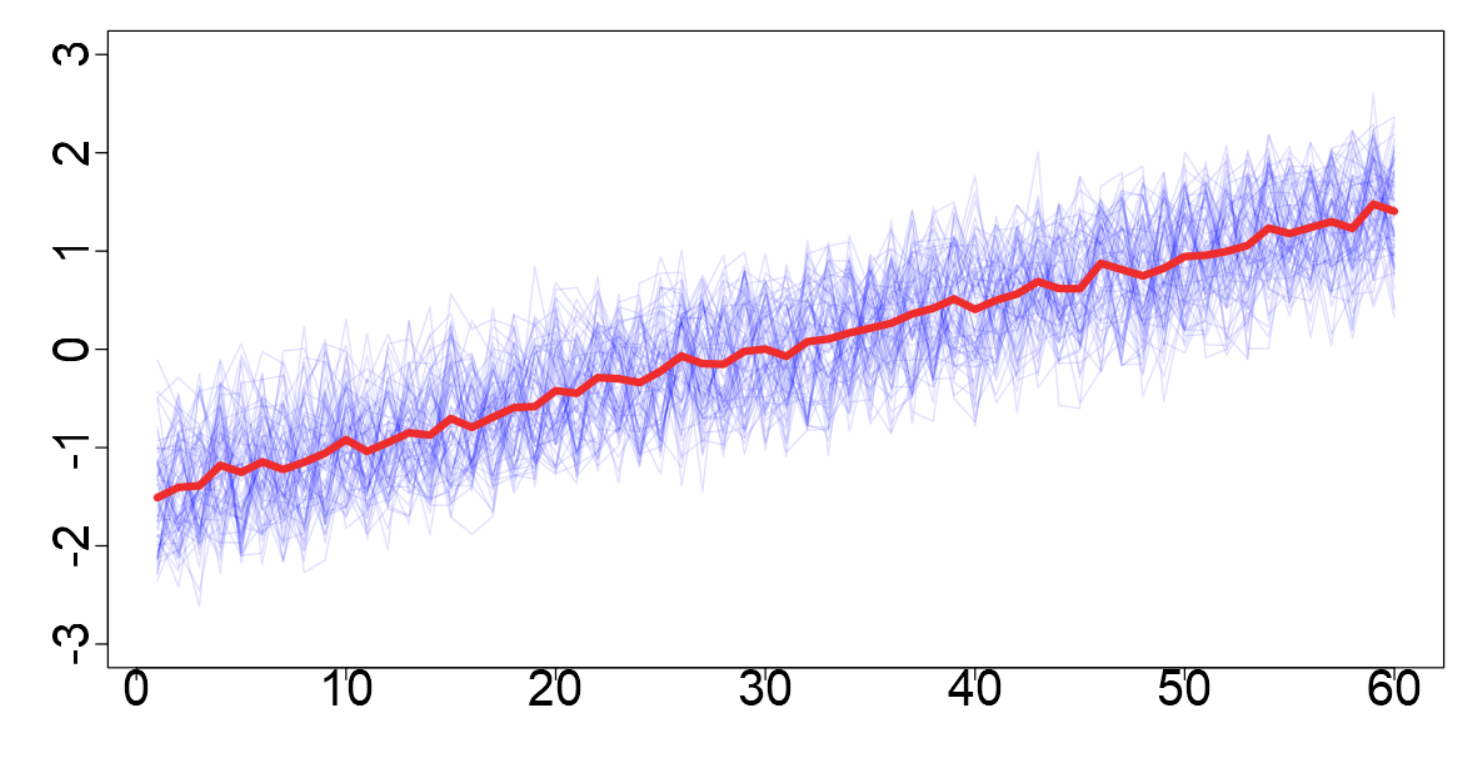} &&
	\includegraphics[width=0.25\linewidth]{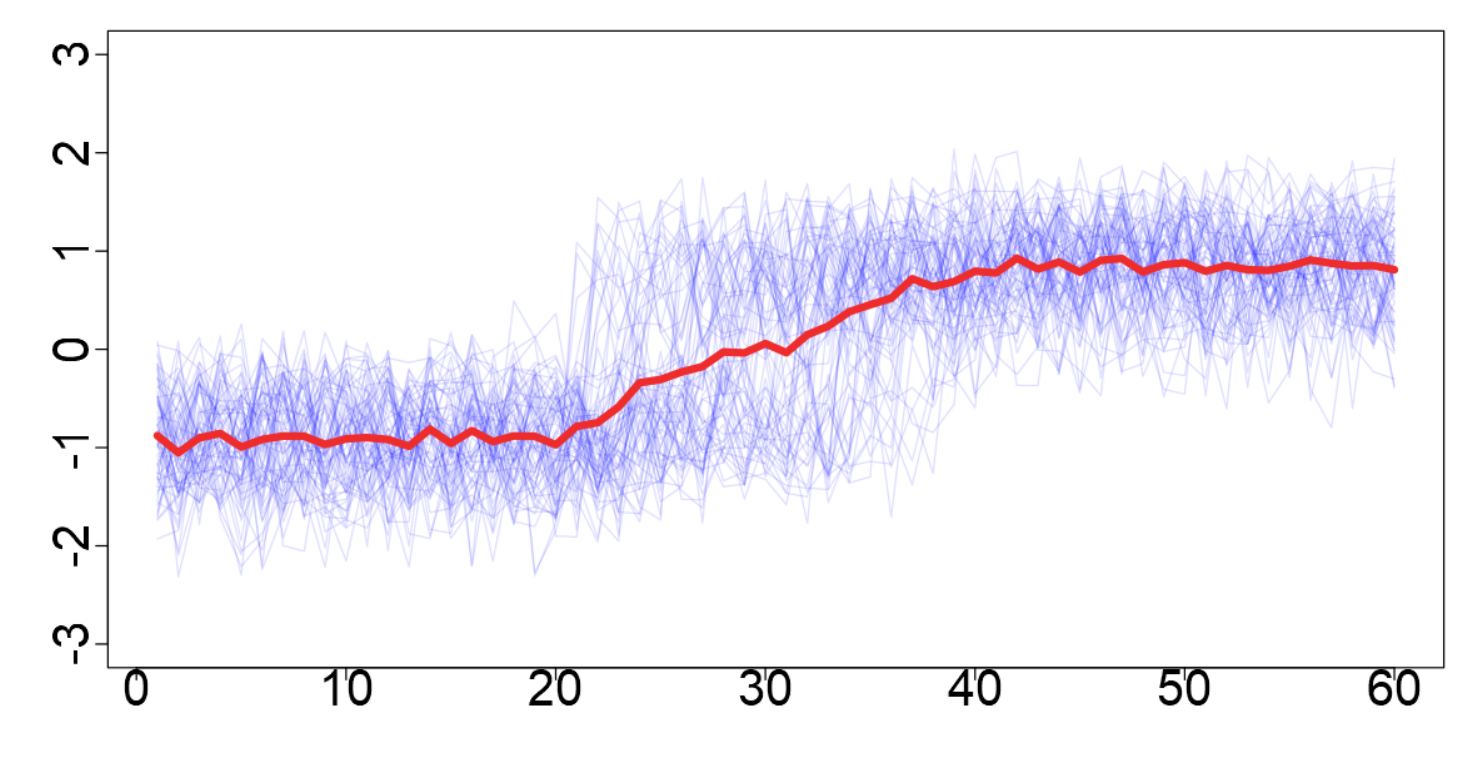} \\ [3pt]
	~~Cyclic && ~~Decreasing Trend && ~~Downward Shift \\
	\includegraphics[width=0.25\linewidth]{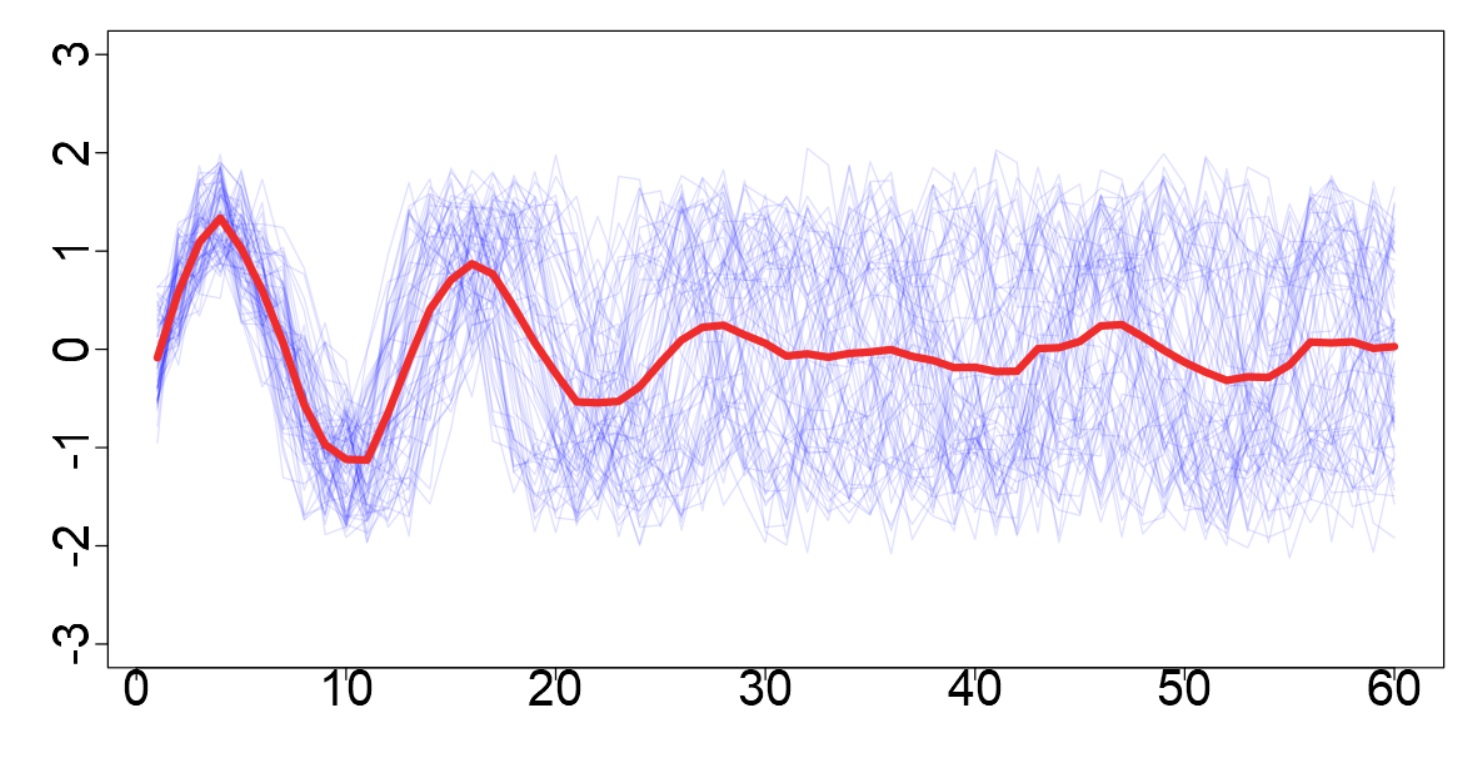} &&
	\includegraphics[width=0.25\linewidth]{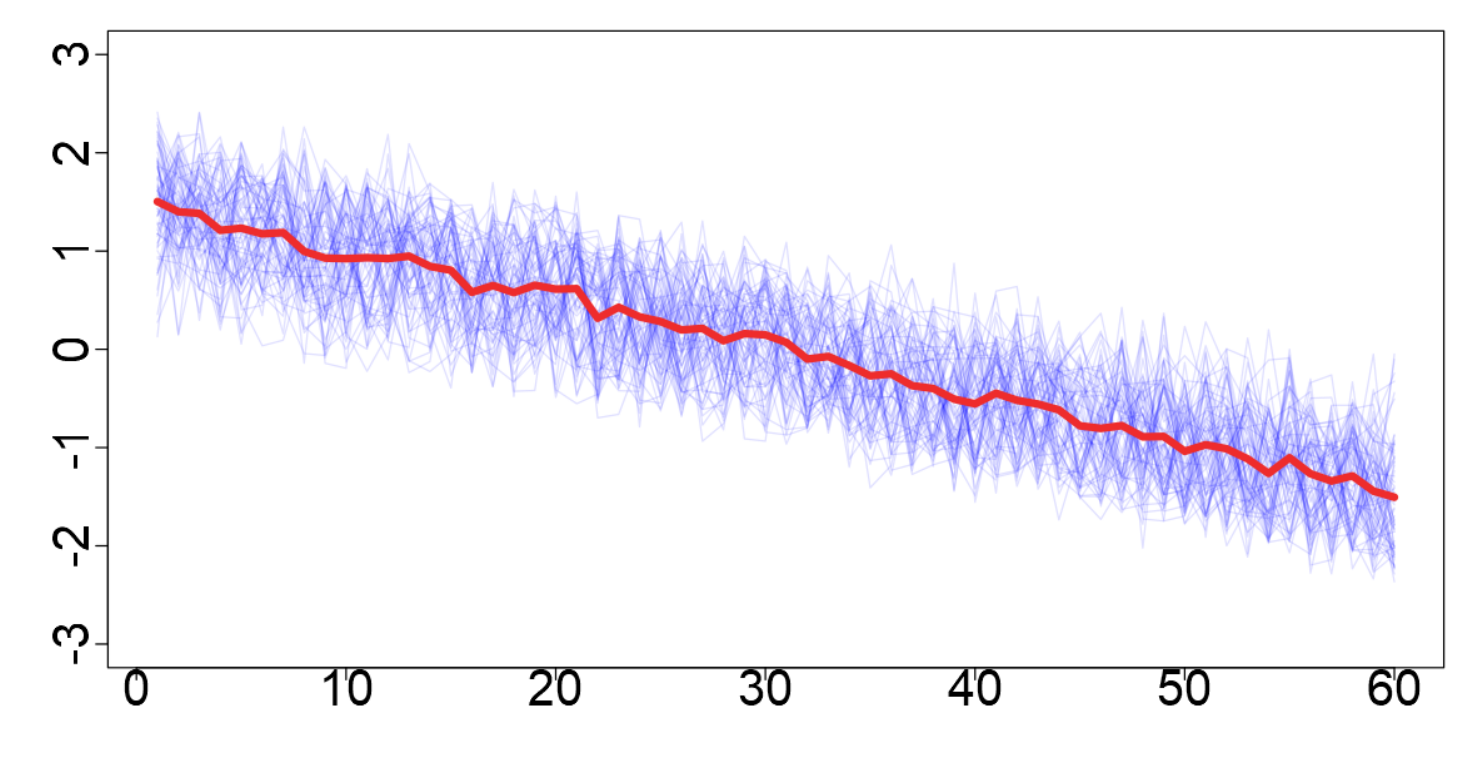} &&	
	\includegraphics[width=0.25\linewidth]{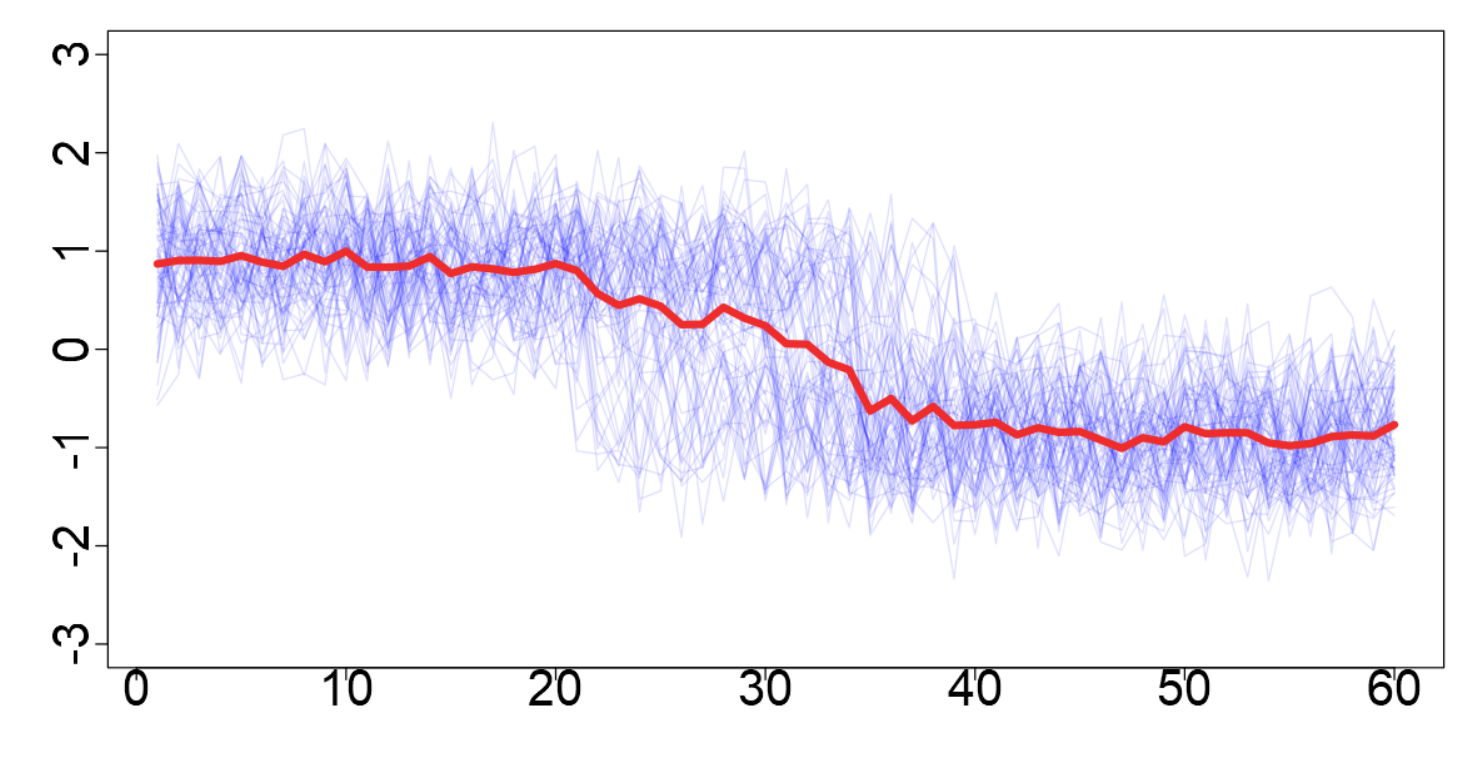}
	\end{tabular}
    \caption{Time series plot of Synthetic Control Chart data.
\label{fig:Controlchart}}
\end{figure}

The \textbf{Mote Strain} dataset is a binary classification data consisting of readings from two types of sensor: Humidity and Temperature. From the time series plot for this data in Figure~\ref{fig:motestrain}, we observe that the locations as well as scales of the two classes are different, but bear resemblance. Consequently, the performance of the NN and NN-MADD$_{\rm sc}$ classifiers are pretty close ($8.57\%$ and $6.78\%$ misclassifications, respectively). There seems to be a difference in the skewness patterns of the two classes. Interestingly, the NN-gMADD$_{\rm sc}$ classifier seems to successfully capture this difference, and produces the minimum misclassification rate ($3.31\%$) in this example.

\begin{table}[t]
    \centering
    \setlength{\tabcolsep}{2pt}
    \small
\caption{Average misclassification rates (in \%) of different classifiers on benchmark datasets. The corresponding standard errors are reported in a smaller font within parentheses. The underlined numbers represent the best result in each example. \label{tab:real}}
    \begin{tabular}{l|rrrr|rrrrrrrrr}
        \toprule
        Dataset &  $J$ & $d$ & $n$ & $n_{\rm test}$ &  \multicolumn{1}{c}{NN} &   \multicolumn{1}{c}{GLMNET} & \multicolumn{1}{c}{CART} &
        \multicolumn{1}{c}{RF} &   \multicolumn{1}{c}{LSVM}  &  \multicolumn{1}{c}{NLSVM} &
        \multicolumn{1}{c}{NN-gMADD$_{\rm sc}$} &  \multicolumn{1}{c}{NN-MADD$_{\rm sc}$} \\
        \midrule
        Synthetic & 6 & 60 & 426 & 174 & 9.13 & 6.21 & 26.16 & 2.53 & 5.72 & 1.68 & 2.62 & \best{1.29} \\
        Control Chart & & &  & &  \se{0.31} &  \se{0.26} &  \se{0.64} &  \se{0.26} &  \se{0.19} &  \se{0.17} & \se{0.21} & \se{0.12}\\
       \addlinespace[2pt]
    
        Mote Strain & 2 & 84 & 889 & 383 & 8.57 & 9.58 & 8.91 & 3.92 & 10.09 & 5.43 & \best{3.31} & 6.78 \\
          &  &  &  &  & \se{0.18} & \se{0.27} & \se{0.28} & \se{0.24} & \se{0.22} & \se{0.22} & \se{0.21} & \se{0.24} \\
      \addlinespace[2pt]
      
        Two Patterns & 4 & 128 & 3500 & 1500 & 1.83 & 12.29 & 28.22 & 1.93 & 12.01 & 2.34 & \best{0.54} & 1.68 \\
          &  &  &  &  & \se{0.07} & \se{0.13} & \se{0.33} & \se{0.08} & \se{0.15} & \se{0.07} & \se{0.04} & \se{0.06} \\
      \addlinespace[6pt]
        Wafer & 2 & 152 & 5013 & 2151 & 0.19 & 5.77  & 1.47 &  \best{0.18} & 3.63  &  0.23  & xxx & 0.30   \\
          & & & & &  \se{0.02} &  \se{0.07} &  \se{0.06} &  \se{0.02} &  \se{0.05} &  \se{0.02} & xxx &  \se{0.03} \\
      \addlinespace[2pt]
      
        Star Light & 3 & 1024 & 6464 & 2772  & 11.44 & 6.79 & 6.97 & 3.91 & 12.40 &  \best{2.83} & xxx & 6.21 \\
        Curves  & & & & &  \se{0.11} &  \se{0.07} &  \se{0.07} &  \se{0.05} &  \se{0.41} &  \se{0.06} & xxx & \se{0.08} \\
     \addlinespace[2pt]

        Electric & 7 & 96 & 11758 & 4879 & 28.31 & 46.43 & 34.40 &  \best{20.98} & 48.33 & 25.09 & xxx & 26.41 \\
        Devices & & & & &  \se{0.10} &  \se{0.10} &  \se{0.15} &  \se{0.12} &  \se{0.12} &  \se{0.15} & xxx & \se{0.11} \\
     \bottomrule
    \end{tabular}
    
    xxx The NN-gMADD$_{\rm sc}$ classifier was not used for datasets with more than $5000$ training samples.
\end{table}

\begin{figure}[h]
\centering
\footnotesize
\begin{tabular}{c c c}
~~Sensor~1 && ~~Sensor~2 \\
\includegraphics[width=0.25\linewidth]{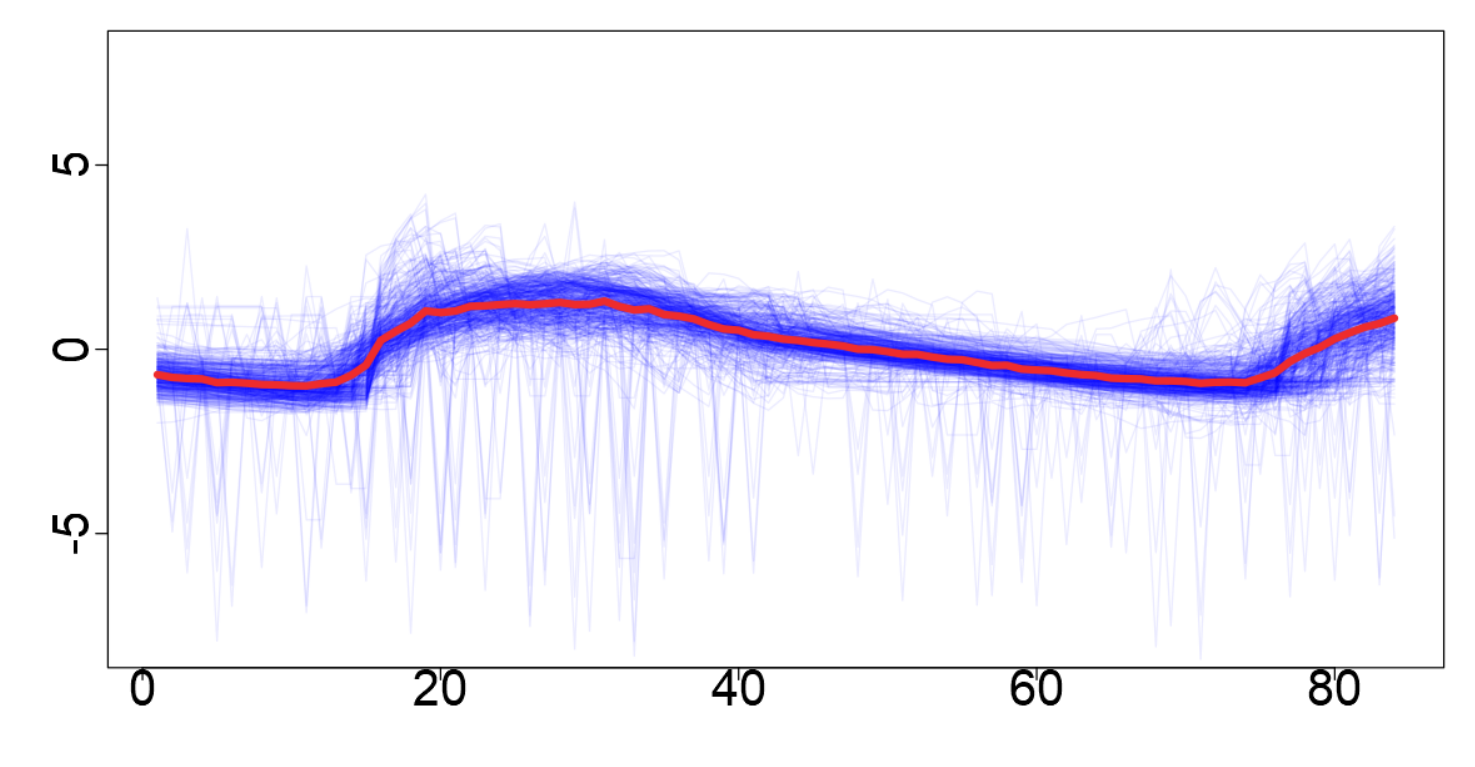} &&
\includegraphics[width=0.25\linewidth]{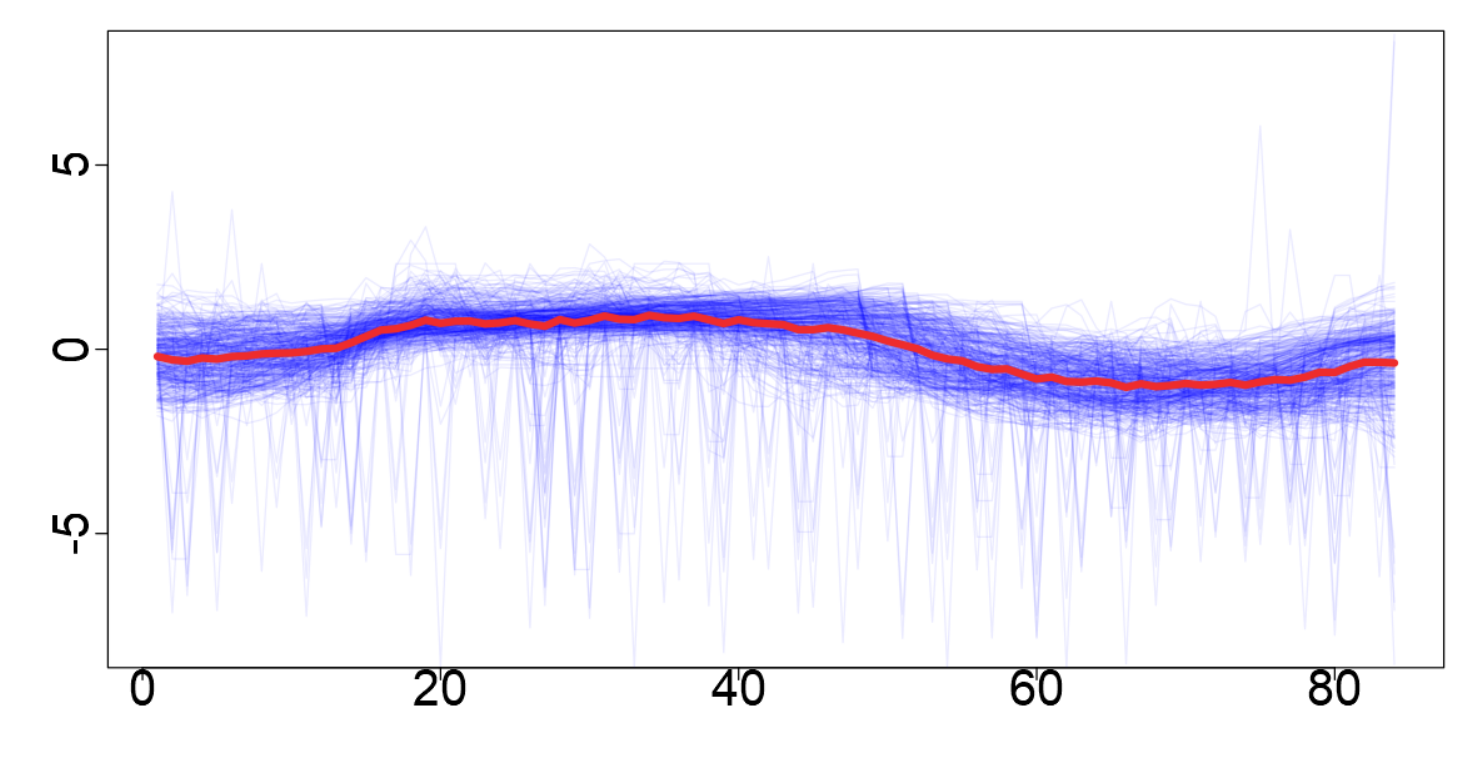} 
\end{tabular}
\caption{Time series plot of Mote Strain data.}
\label{fig:motestrain}
\end{figure}

The \textbf{Two Patterns} dataset contains artificially generated time series patterns by \cite{geurts2002contributions}. The $4$ classes are defined based on the occurrence of up and down patterns in a specific order: DD (two successive down movements), UD (up followed by down), DU (down followed by up), and UU (two successive up movements). The observations are plotted in Figure~\ref{fig:2patt}. A visual inspection shows that the location as well as scale patterns among the four classes are quite similar, although a careful examination shows that there are significant location differences between time points $80$ and $100$. Both NN and NN-MADD$_{\rm sc}$ perform well in this example, achieving misclassification rates of $1.83\%$ and $1.68\%$, respectively. However, the minimum misclassification rate is achieved by the NN-gMADD$_{\rm sc}$ classifier ($0.54\%$), which is able to extract information beyond the first two moments. The linear classifiers have significantly poor performance compared to the others.

\begin{figure}[b!]
\centering
\footnotesize

\begin{tabular}{c c c c}
DD & DU & UD & UU \\
\includegraphics[width=0.22\linewidth]{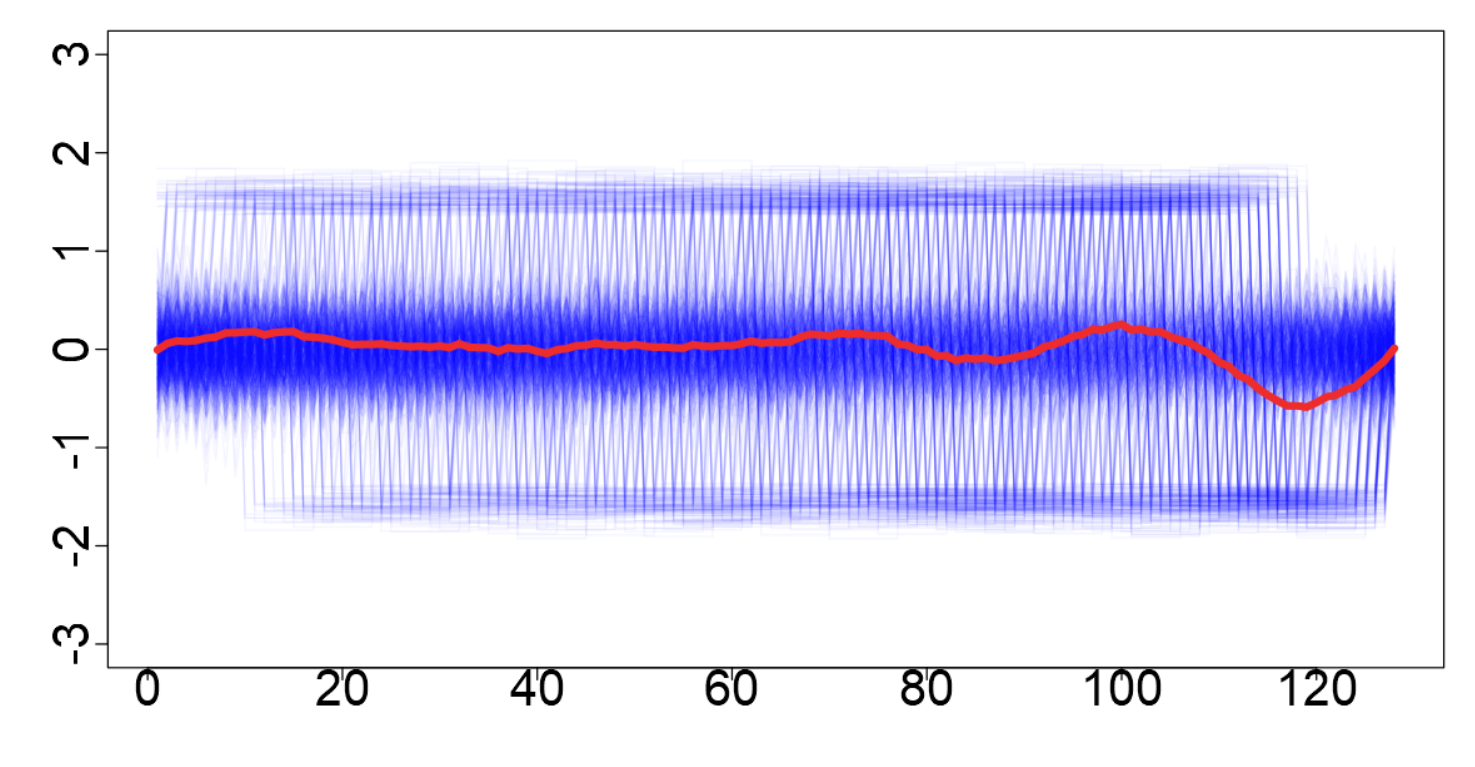} & 
\includegraphics[width=0.22\linewidth]{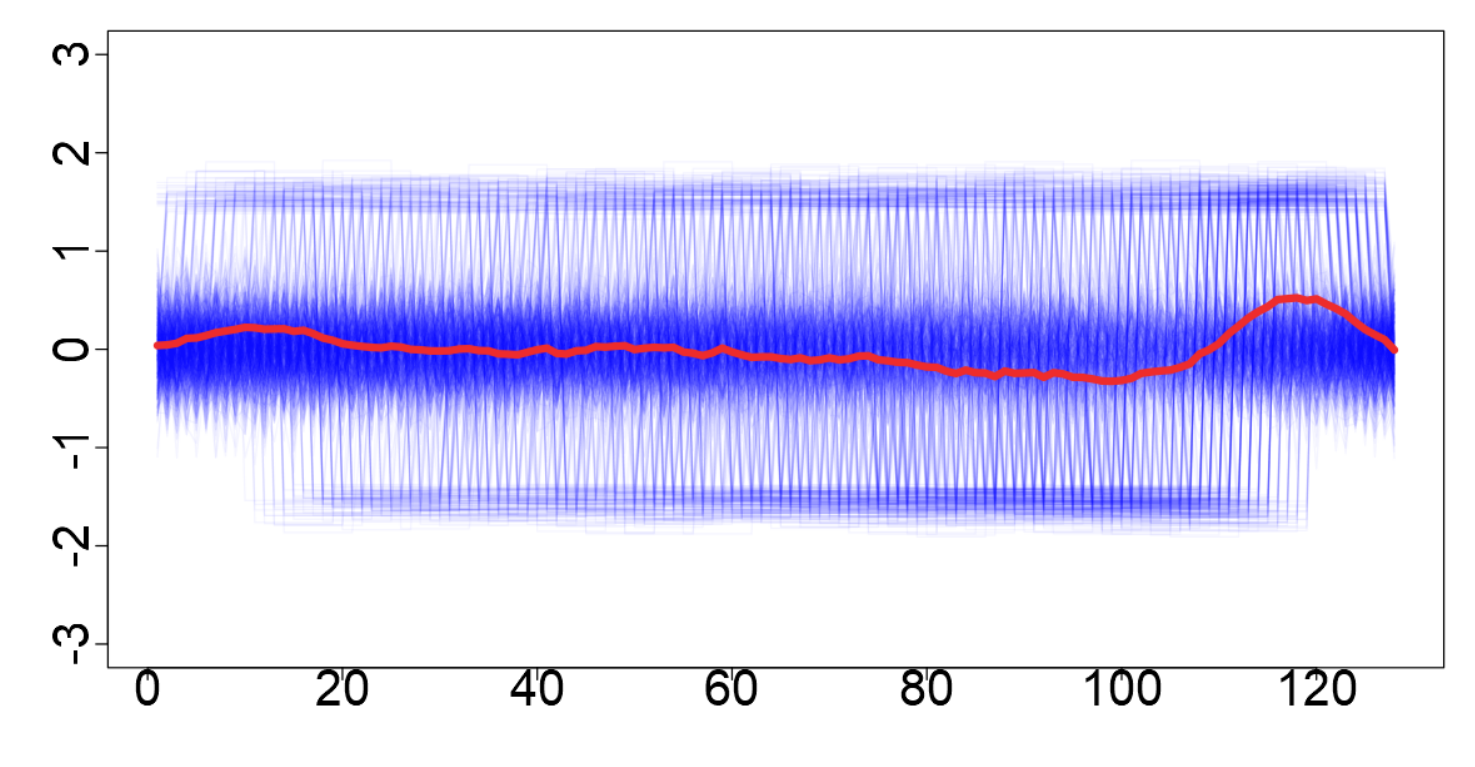} & 
\includegraphics[width=0.22\linewidth]{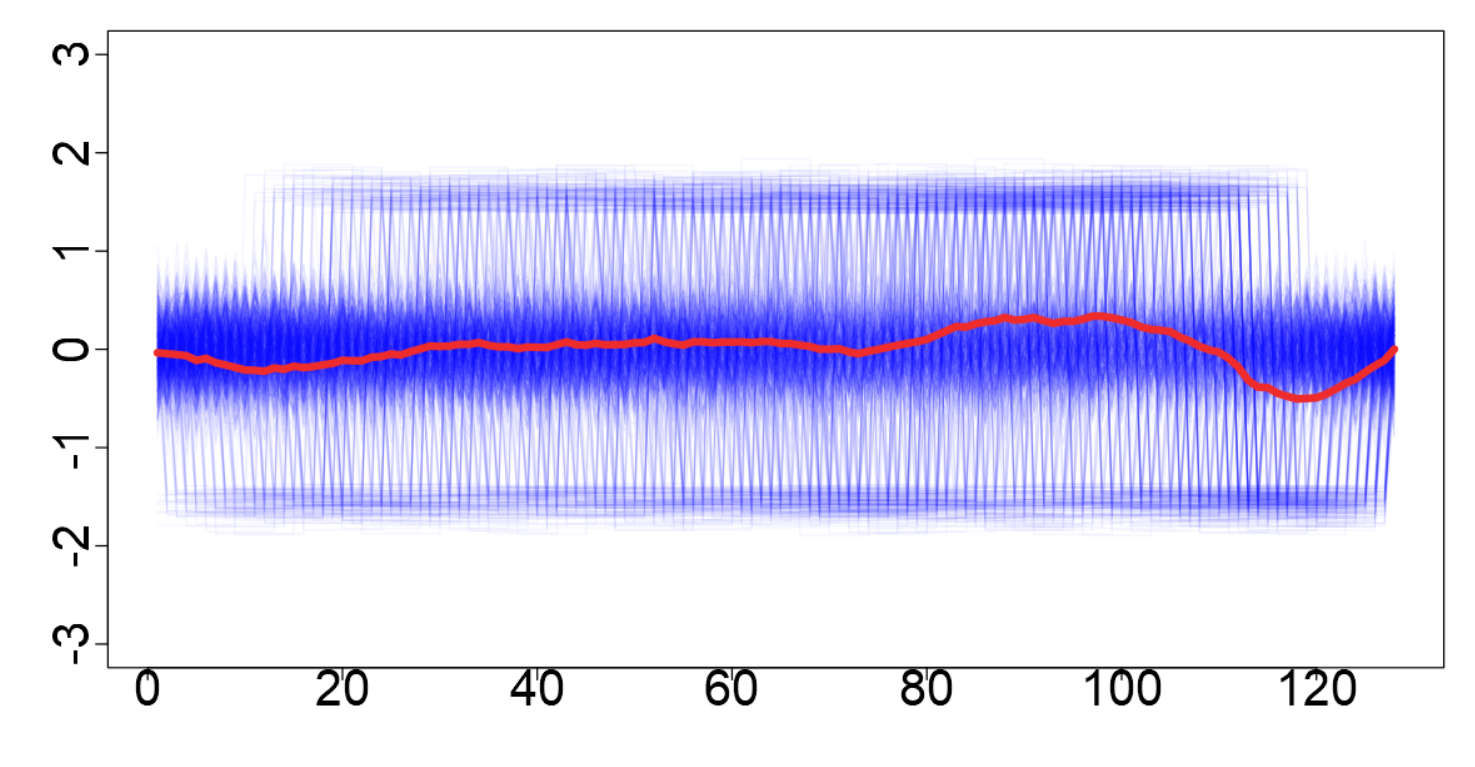} & 
\includegraphics[width=0.22\linewidth]{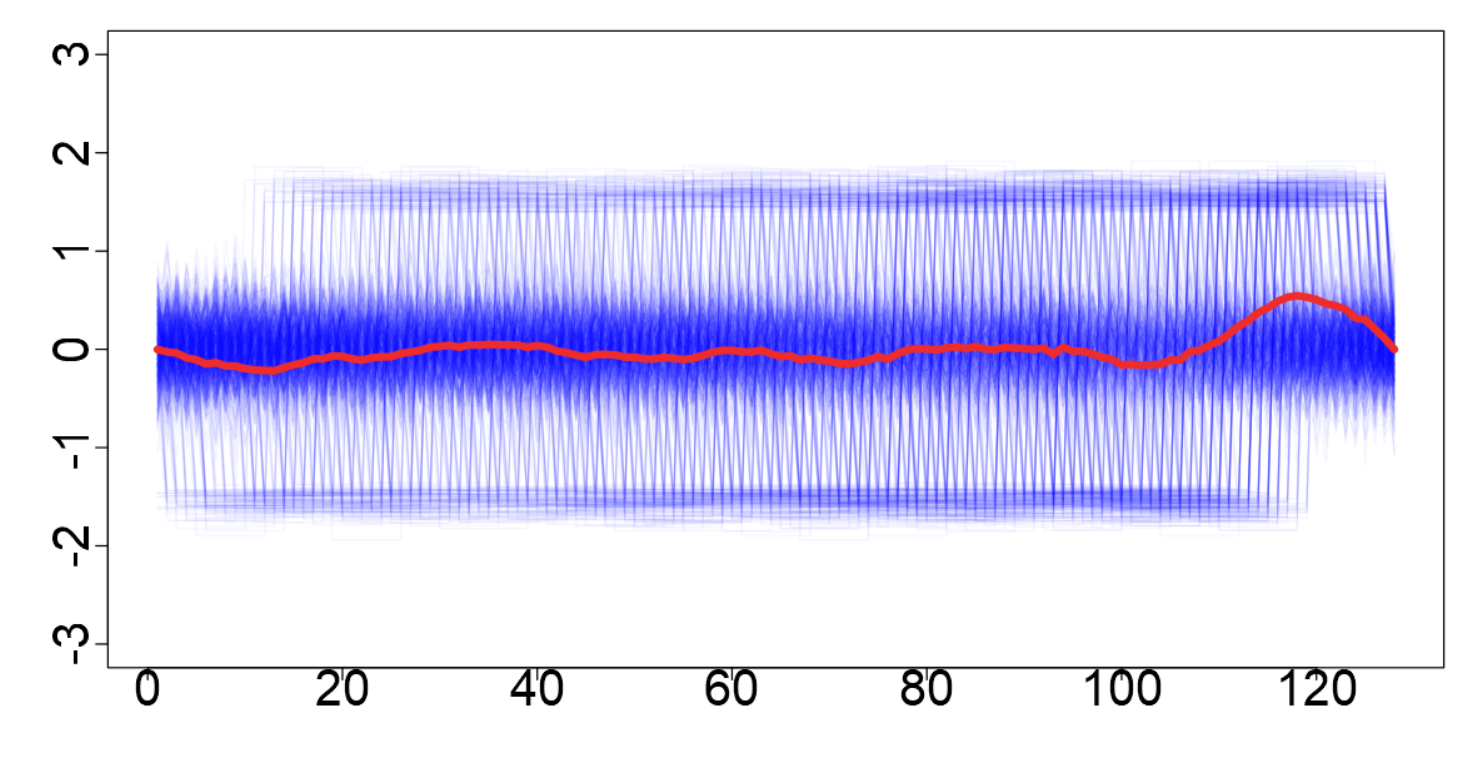}
\end{tabular}
\caption{Time series plot of Two Patterns data.} 
\label{fig:2patt} 
\end{figure}

\subsection{Large-Sized Training Samples}
Next, we consider the $3$ datasets with more than $5000$ training samples. For these examples, we do not apply the NN-gMADD$_{\rm sc}$ classifier, as it becomes computationally demanding.
\medskip

\begin{figure}[t]
\centering
\footnotesize
\begin{tabular}{c c c}
~~Normal && ~~Abormal \\
\includegraphics[width=0.25\linewidth]{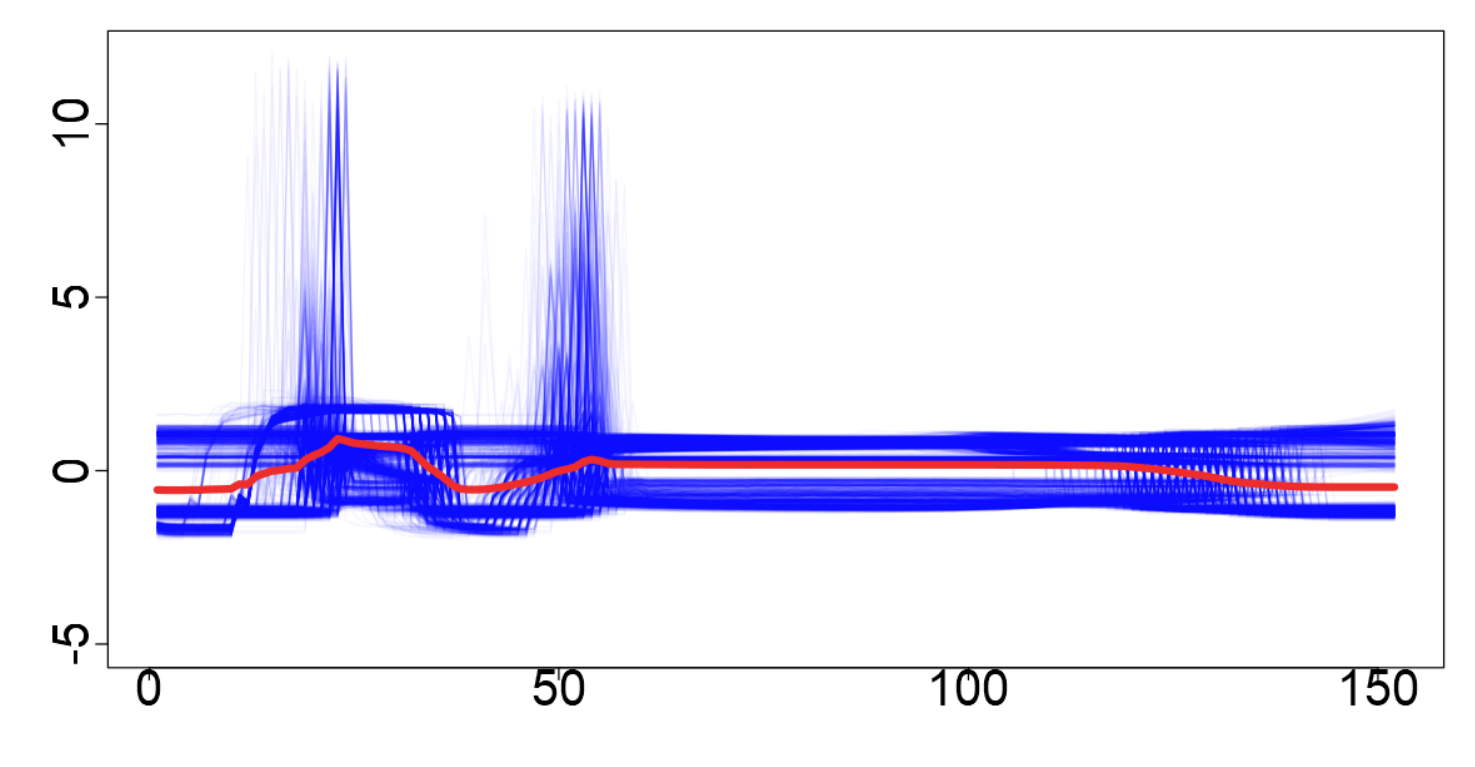} &&
\includegraphics[width=0.25\linewidth]{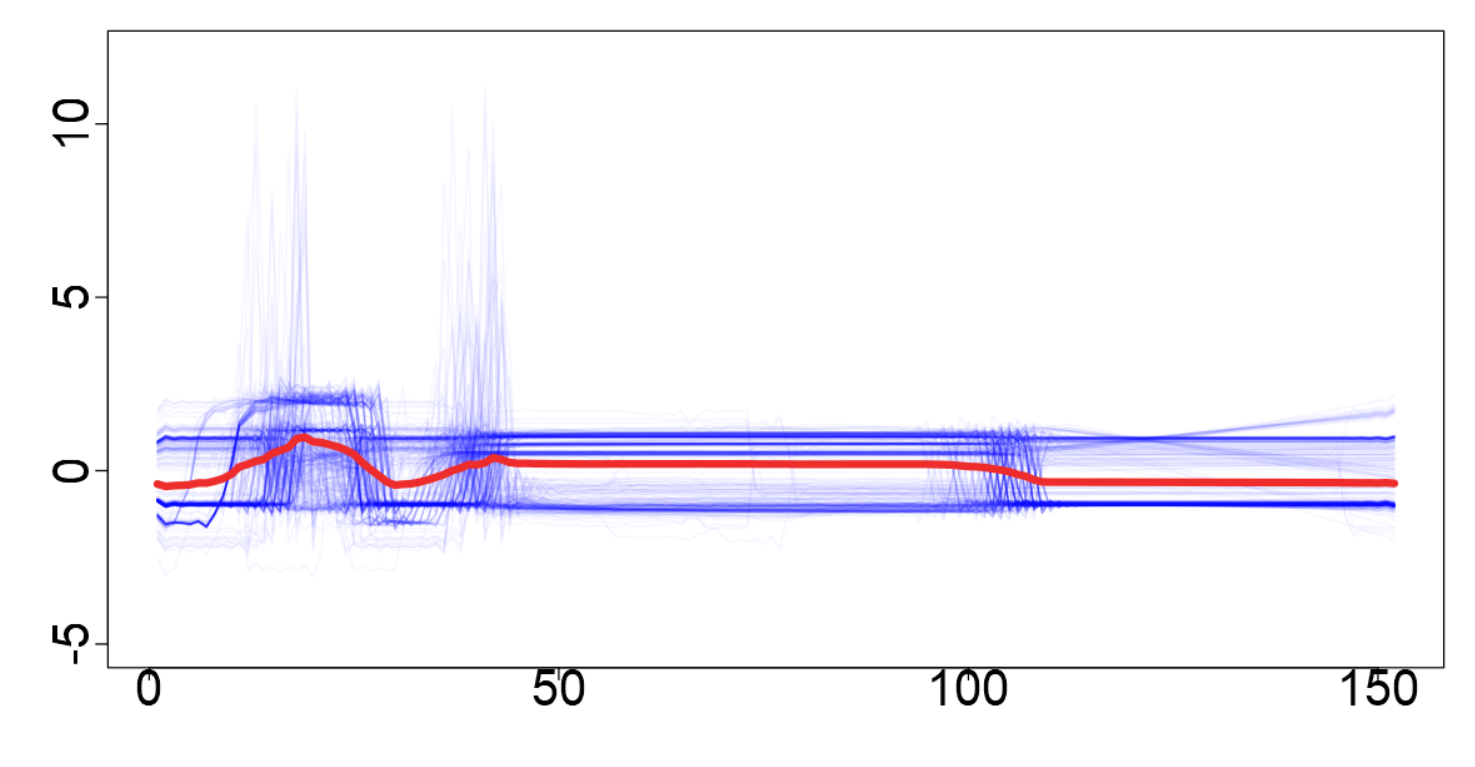} 
\end{tabular}
\caption{Time series plot of the Wafer data.}
\label{fig:Wafer}
\end{figure}

The \textbf{Wafer} dataset is a collection of measurements recorded by various sensors during the processing of silicon wafers for semiconductor fabrication \citep{olszewski2001generalized}. Each time series was recorded by a single sensor while one wafer was processed by one tool. Here, we have two underlying classes: Normal and Abnormal. The plot of the data in Figure~\ref{fig:Wafer} shows that the Abnormal class observations resemble a left-shifted version of the Normal class observations. The resulting location difference helps the NN classifier to correctly distinguish between the two classes, with a misclassification rate of only $0.19\%$. The NN-MADD$_{\rm sc}$ classifier has a slightly higher misclassification rate of $0.30\%$. Here, RF has the best performance with a misclassification rate of $0.18\%$. Overall, all the classifiers have satisfactory performance on this dataset.

The \textbf{Star Light Curves} dataset contains brightness measurements of three categories of variable stars recorded over $1024$ time points. Cepheids and RR Lyrae are both pulsating stars that change in brightness as they physically expand and contract, though they differ in their evolutionary stages and cycle lengths. In contrast, Eclipsing Binaries exhibit brightness variations because two stars orbit each other, with one periodically blocking the light of the other. The plot in Figure~\ref{fig:SLC} shows that the locations as well as scales for Cepheid and RR Lyrae are quite similar, whereas Eclipling Binaries exhibit a significantly different shape. Due to the similarity in the locations of the two categories, the NN classifier has a relatively higher misclassification rate ($11.44\%$) in this example. Comparatively, NN-MADD$\rm_{sc}$ has a much lower misclassification error ($6.21\%$), which is the third best overall.

\begin{figure}[h!]
\centering
\footnotesize
\begin{tabular}{c c c c}
Cepheid & Eclipsing Binary & RR Lyrae\\
\includegraphics[width=0.25\linewidth]{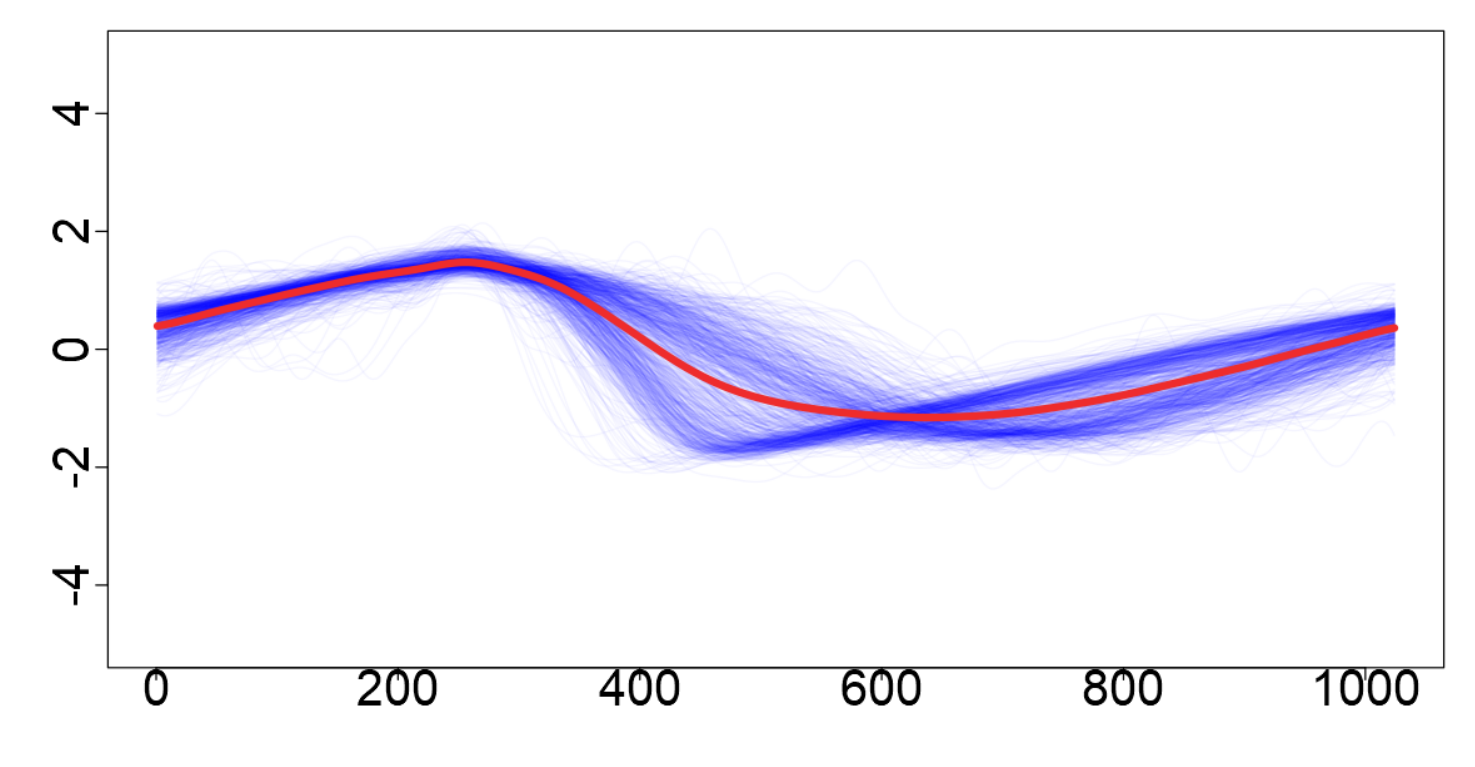} & 
\includegraphics[width=0.25\linewidth]{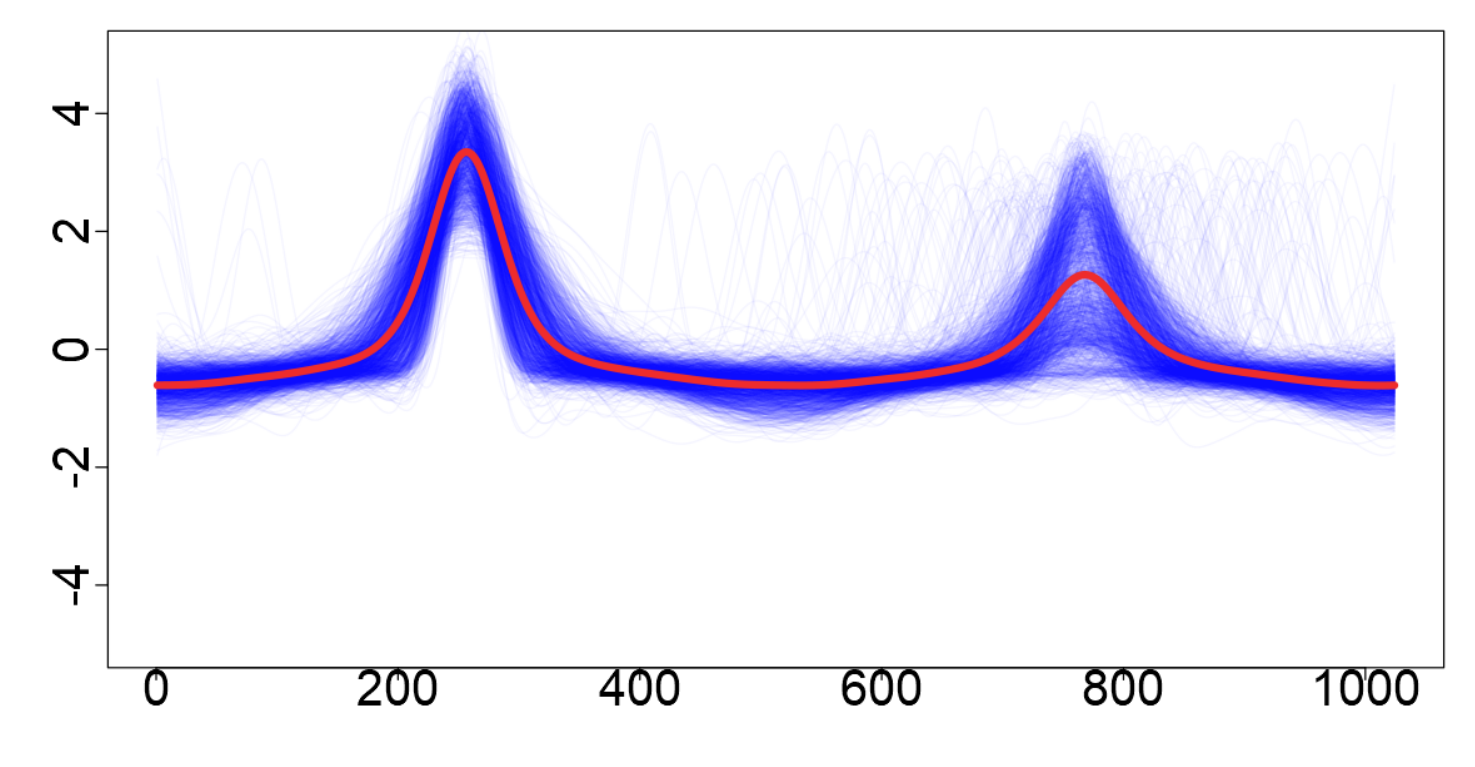} & 
\includegraphics[width=0.25\linewidth]{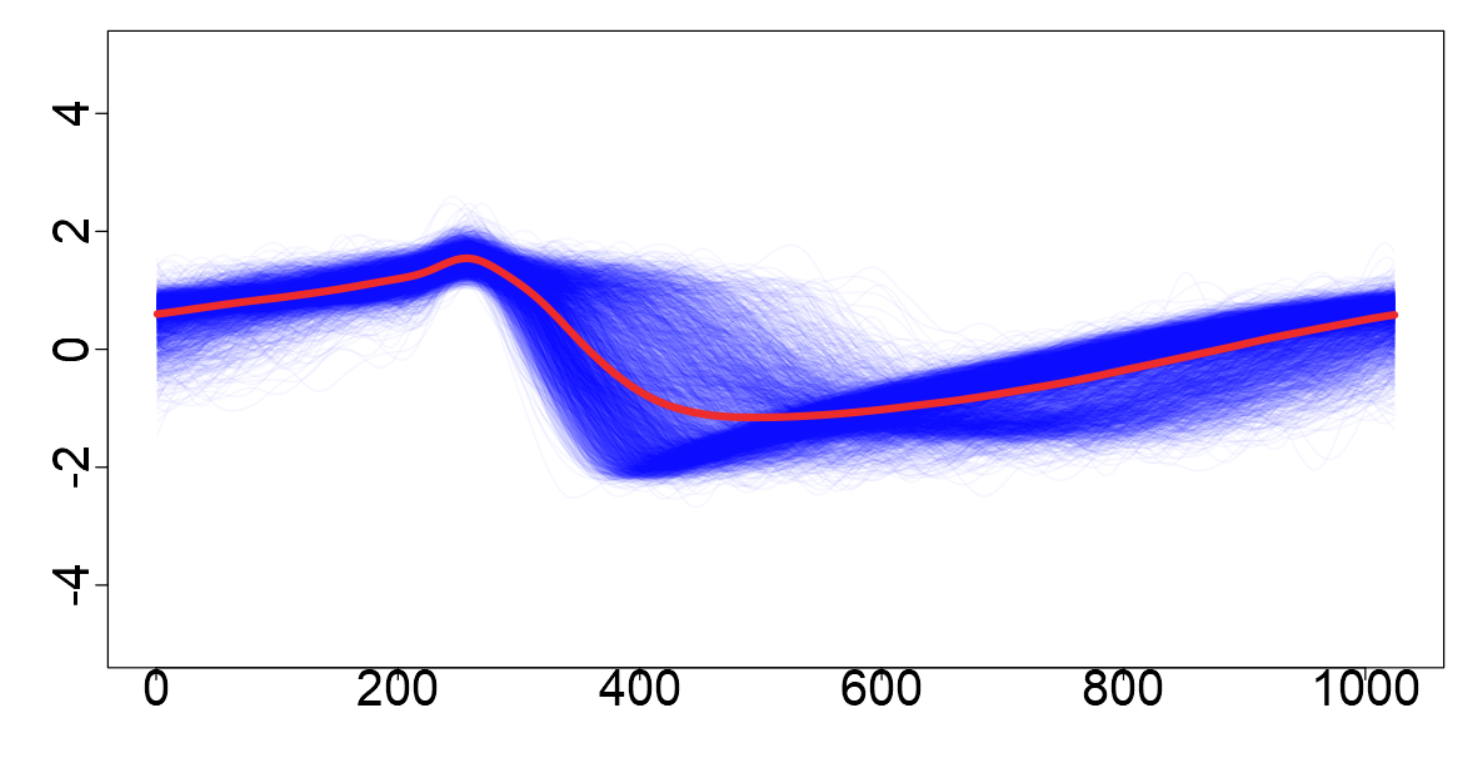} & 
\end{tabular}
\caption{Time series plot of the Star Light Curves data.} 
\label{fig:SLC} 
\end{figure}

\begin{figure}[b!]
\centering
\footnotesize
\begin{tabular}{c c c c}
Screen/TV group & Dishwasher & Cold group & Immersion heater \\
\includegraphics[width=0.23\linewidth]{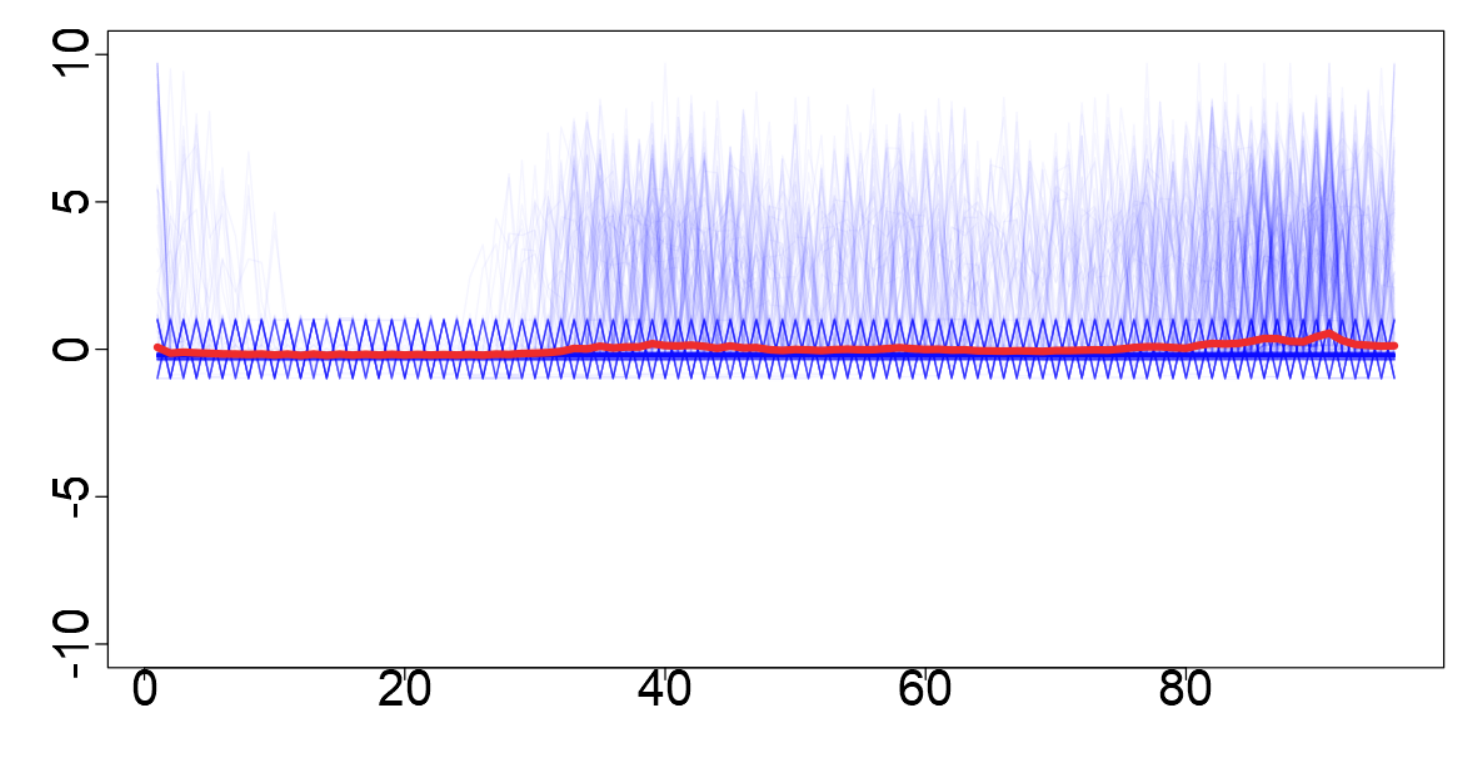} &
\includegraphics[width=0.23\linewidth]{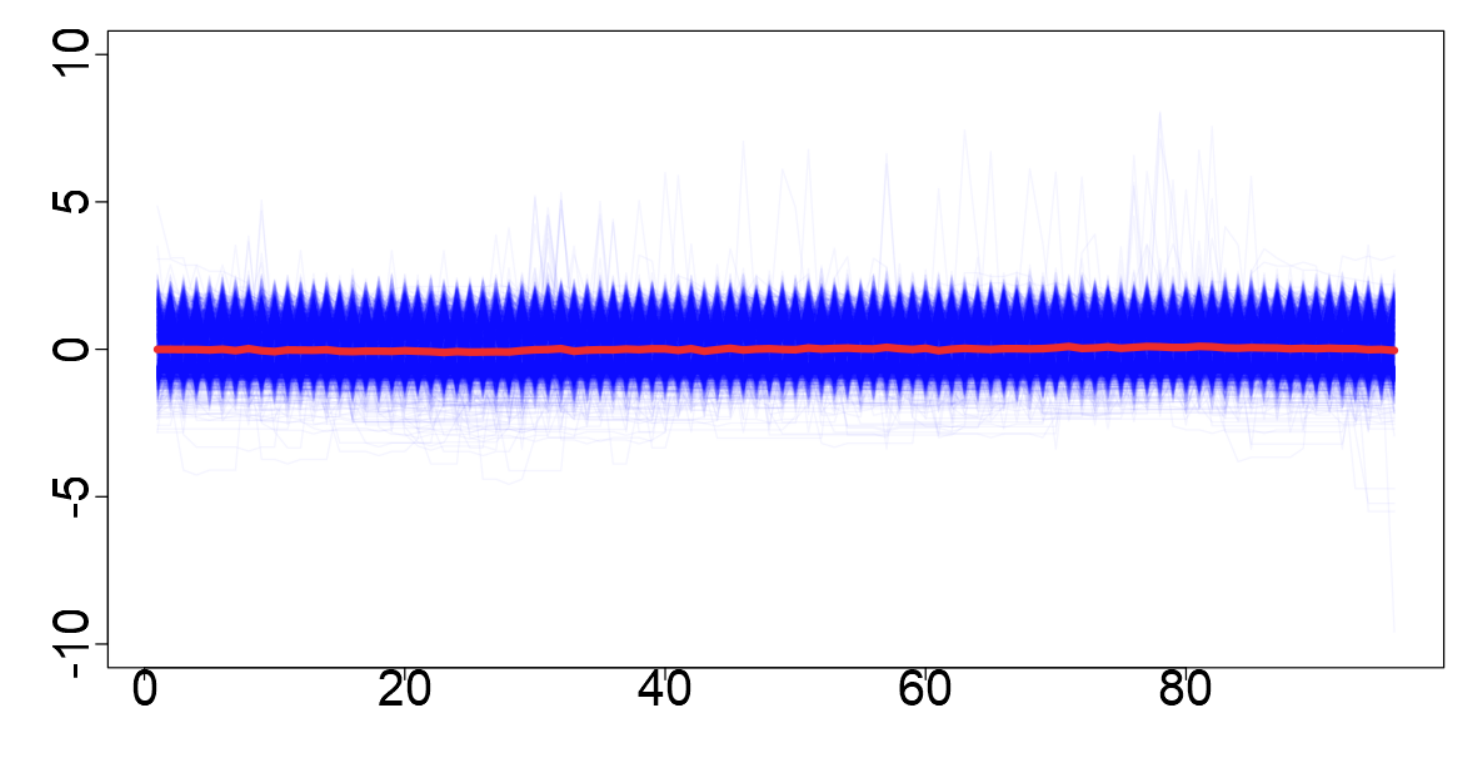} &
\includegraphics[width=0.23\linewidth]{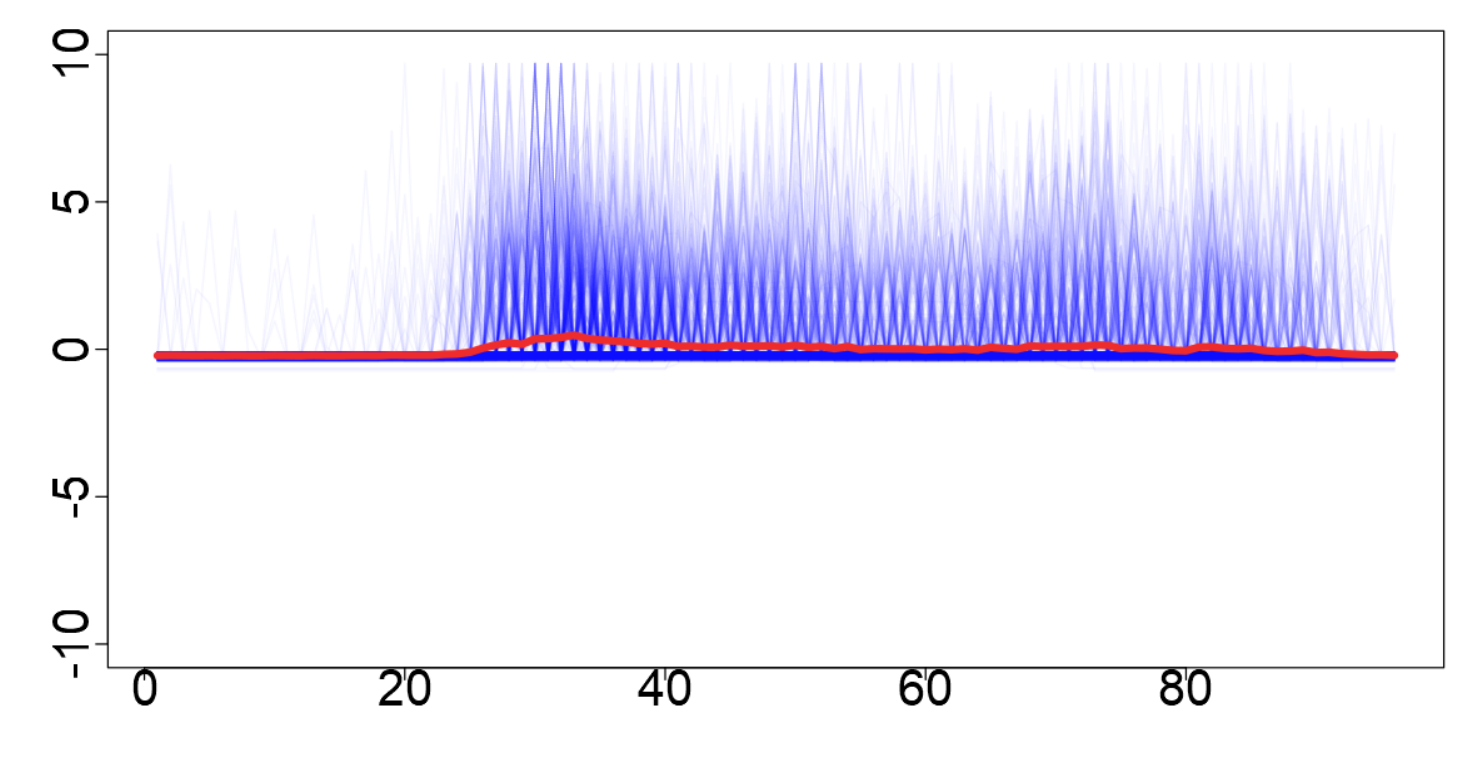} &
\includegraphics[width=0.23\linewidth]{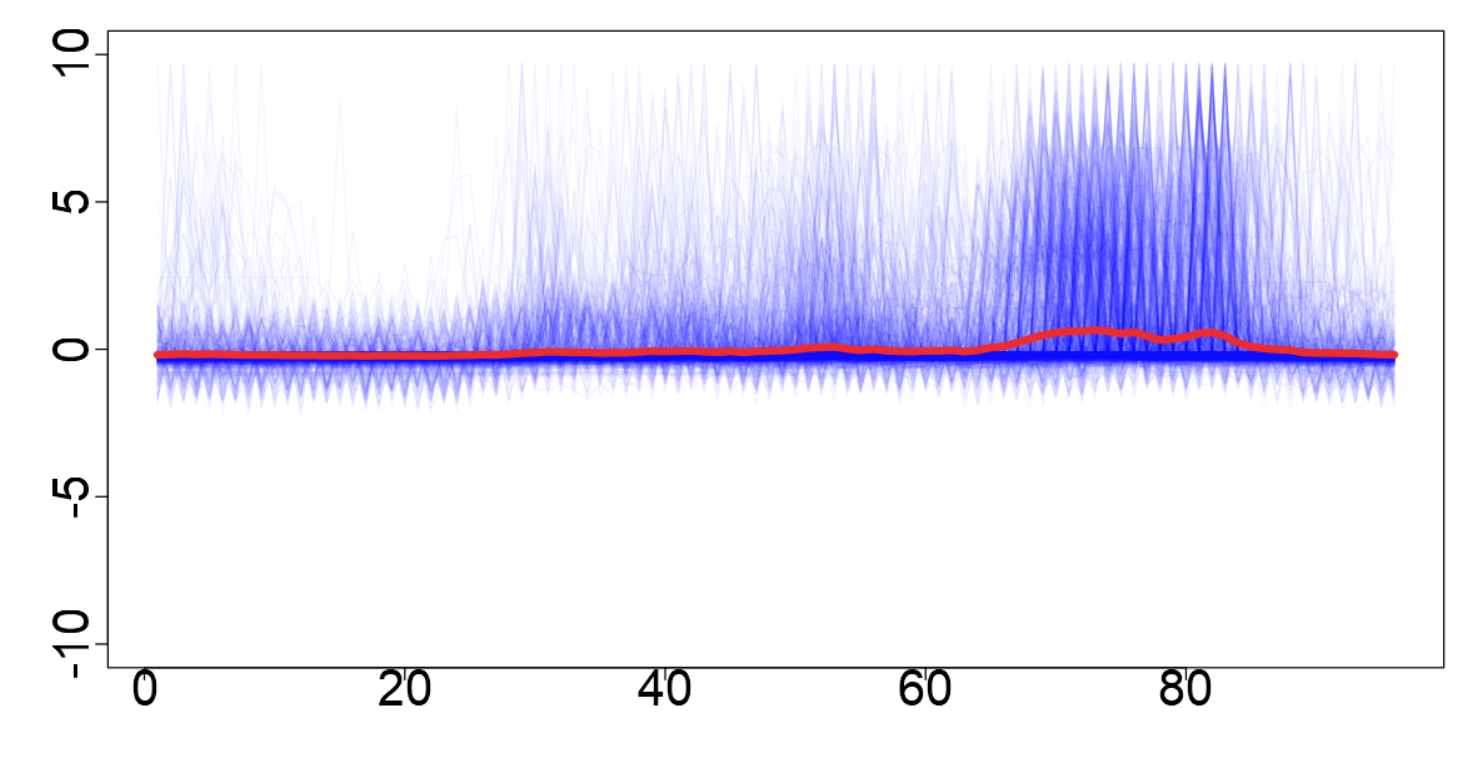}
\end{tabular}

\vspace{0.25cm}

\begin{tabular}{c c c}
Kettle & Oven/cooker & Washing machine \\
\includegraphics[width=0.23\linewidth]{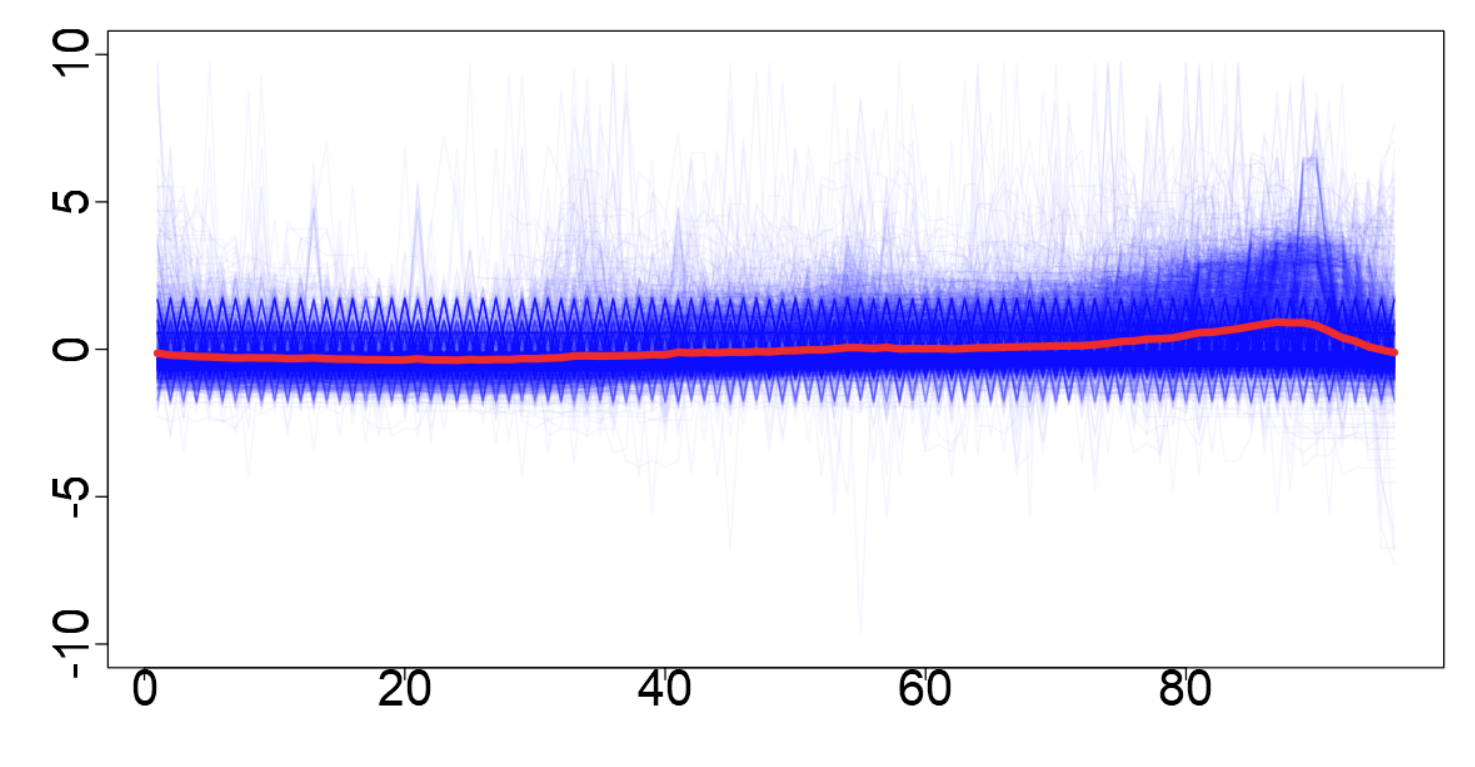} &
\includegraphics[width=0.23\linewidth]{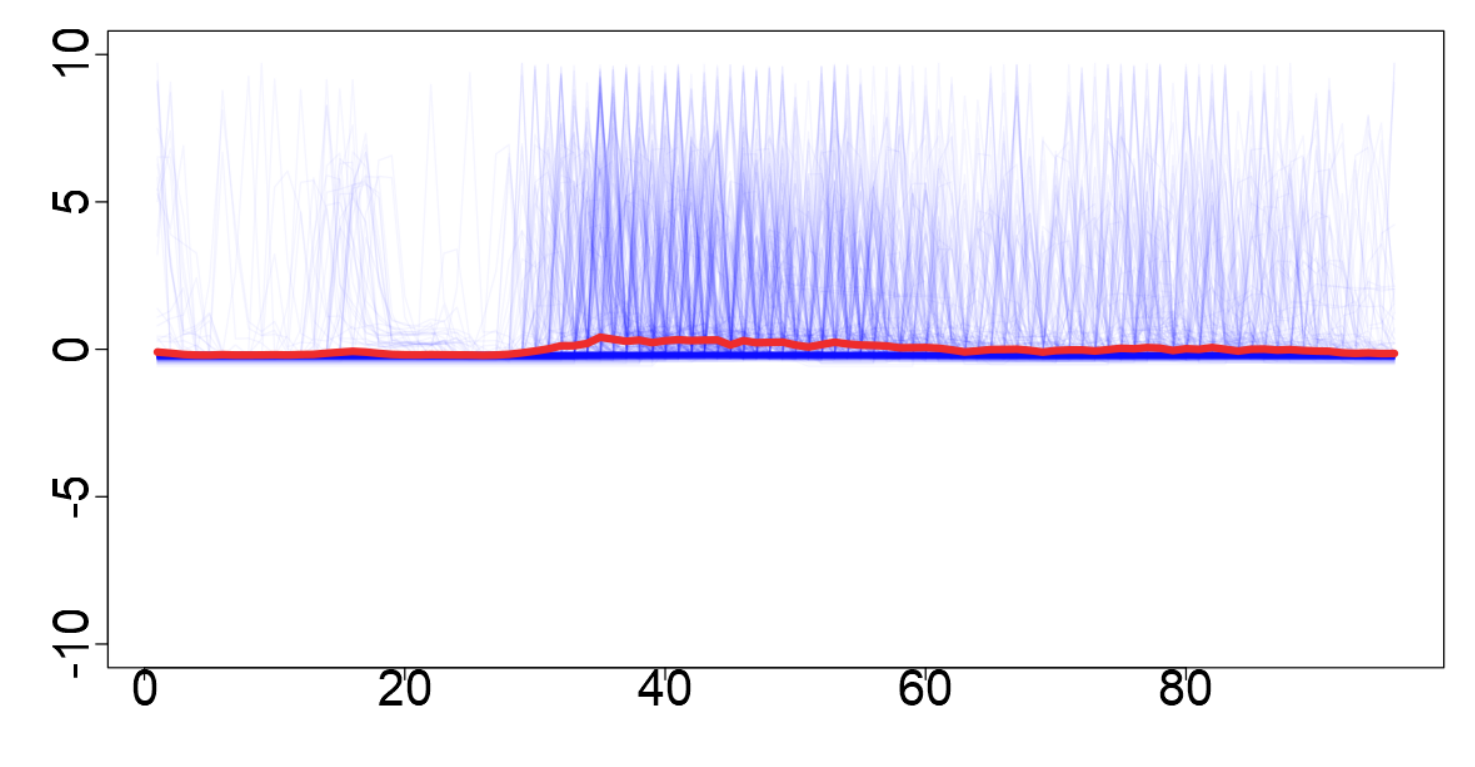} &
\includegraphics[width=0.23\linewidth]{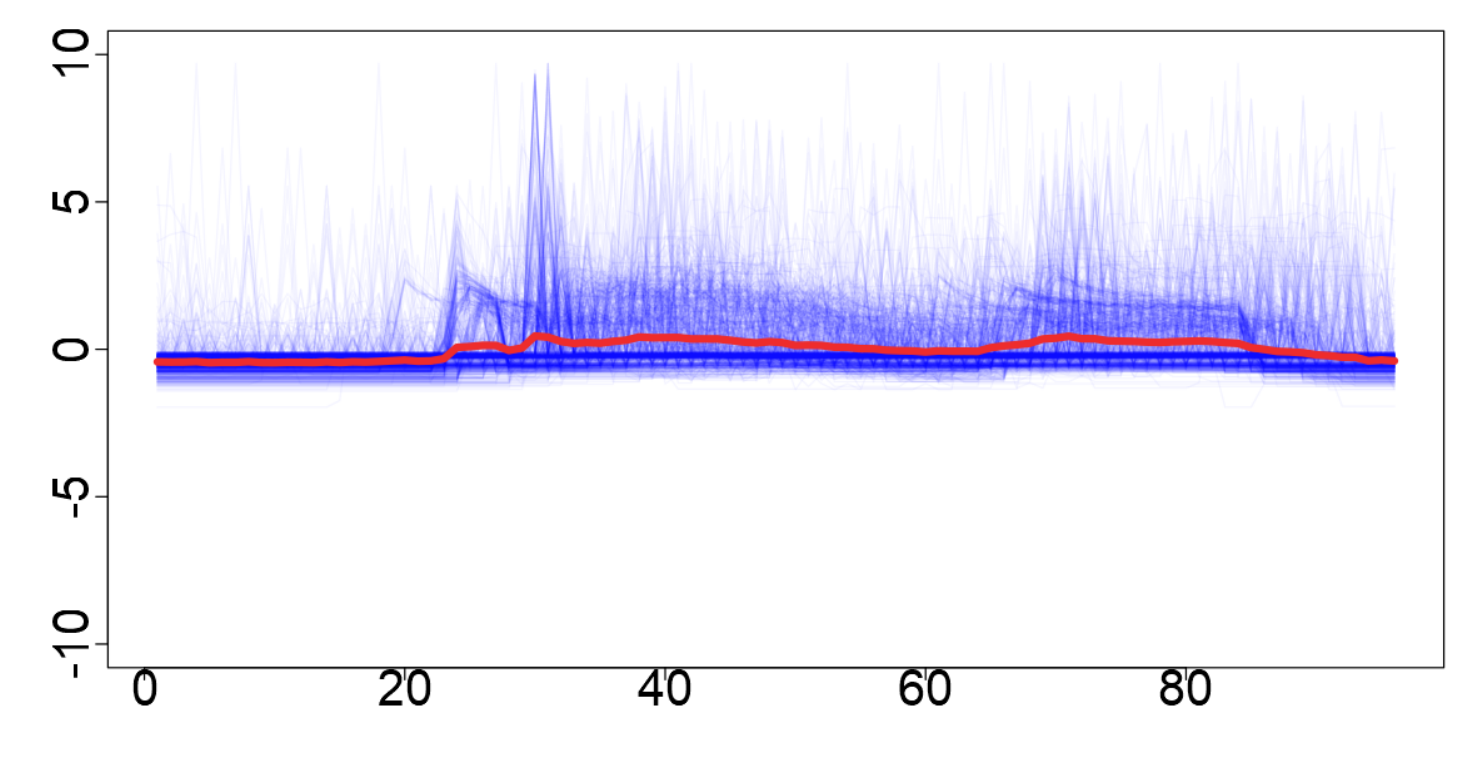}
\end{tabular}
\caption{Time series plots of the Electric Devices data.}
\label{fig:ED}
\end{figure}

The \textbf{Electric Devices} dataset is derived from one month of electricity readings from $251$ households in the UK. Each observation represents a $24$-hour usage profile of an appliance at $15$ minute intervals. Depending on the type of appliance, the observations are classified into seven categories: Screen/TV Group, Dishwasher, Cold Group, Immersion Heater, Kettle, Oven/Cooker, and Washing Machine. The plot of the data in Figure~\ref{fig:ED} shows that the patterns of the observations vary significantly across classes, but the location differences are not particularly prominent. For this dataset, RF achieves the best performance with an error rate of $20.98\%$. It is followed by NLSVM ($25.09\%$) and NN-MADD$_{\rm sc}$ ($26.41\%$), respectively. The performance of linear classifiers is noticeably worse than the others, suggesting that the class boundaries are likely to be nonlinear.

\section{Theoretical Behavior of the Proposed Classifier}
\label{sec:Theorem}

In this section, we study the theoretical behavior of the proposed NN-MADD$_{\rm sc}$ classifier. We consider the high-dimension, high-sample-size (HDHSS) regime, where both the training sample size $n$ and the dimension $d (= d_n)$ go to infinity. We also assume $n_j/n \to \alpha_j \in (0,1)$ for $j=1,\ldots,J$, which ensures that the training sample sizes from the different classes are asymptotically comparable. In the following, we show that for any test observation, the NN-MADD$_{\rm sc}$ and NN-MADD classifiers agree with probability tending to $1$. Consequently, the difference between the misclassification probabilities of the two classifiers converges to $0$.

Recall that for a test observation $\vec Z$, NN-MADD$_{\rm sc}$ classifies it to $\delta_{\rm sc}(\vec Z) = \argmin_{j=1,\ldots,J} \rho_{\rm sc}(\vec Z,\sp X_j)$, where $\rho_{\rm sc}(\vec Z,\sp X_j) = \min_{i=1,\ldots,n_j} \rho_{\rm sc}(\vec Z, \vec X_{ji})$ is the MADD$_{\rm sc}$ distance between $\vec Z$ and the training sample from class-$j$. Similarly, NN-MADD classifies $\vec Z$ into $\delta(\vec Z) = \argmin_{j=1,\ldots,J} \rho(\vec Z,\sp X_j)$, where $\rho(\vec Z,\sp X_j) = \min_{i=1,\ldots,n_j} \rho(\vec Z, \vec X_{ji})$. In the following, we work with the scaled versions of $\rho$ and $\rho_{\rm sc}$, viz., $\widetilde\rho = d_n^{-1/2}\rho$ and $\widetilde\rho_{\rm sc} = d_n^{-1/2}\rho_{\rm sc}$. This is motivated by the high-dimensional behavior of Euclidean distance, which scales as the square-root of the dimension \citep{hall2005geometric, sarkar2019perfect, roy2022generalizations}. Note that the scaling does not affect the classification rules.
\medskip

For our theoretical results, we define two quantities:
\begin{align}\label{eq:beta_n}
    \widetilde\beta_n(\vec Z) =\sup_{1\le j\le J}\ \sup_{1\le i\le n_j} \big| \widetilde\rho_{\mathrm{sc}}(\vec Z,\vec X_{ji}) - \widetilde\rho(\vec Z,\vec X_{ji})\big|,
\end{align}
which denotes the \emph{maximum difference} between $\widetilde\rho$ and $\widetilde\rho_{\rm sc}$ for a particular test observation $\vec Z$, and 
\begin{align}\label{eq:gamma_n}
    \widetilde\gamma_n(\vec Z) = \min_{j\neq j^\ast} \Big\{\widetilde \rho(\vec Z,\sp X_j)-\widetilde\rho(\vec Z,\sp X_{j^\ast})
\Big\}, \text{ where } j^\ast = \argmin_{j=1,\ldots,J} \widetilde\rho(\vec Z, \sp X_{j^\ast}),
\end{align}
which denotes the \emph{minimum margin} achieved by the full data based MADD distance. With these two quantities defined, we get the following result connecting the NN-MADD and NN-MADD$_{\rm sc}$ classifiers.

\begin{proposition}
Let $\delta,\delta_{\rm sc}: \mathbb R^{d_n} \to \{1,\ldots,J\}$ denote the NN-MADD and NN-MADD$_{\rm sc}$ classifiers, respectively. For a test observation $\vec Z$, let $\widetilde\beta_n(\vec Z)$ and $\widetilde\gamma_n(\vec Z)$ be defined as in \eqref{eq:beta_n} and \eqref{eq:gamma_n}, respectively. If $\widetilde\beta_n(\vec Z) \overset{P}{\to} 0$ and $\widetilde\gamma_n(\vec Z) \overset{P}{\to} \gamma>0$ as $n \to \infty$, then $\pr\big(\delta_{\rm sc}(\vec Z) \neq \delta(\vec Z)\big) \to 0$ as $n \to \infty$. If the two conditions hold for all test observations from all the classes, and $R_n(\delta)$, $R_n(\delta_{\rm sc})$ denote the misclassification probabilities of NN-MADD and NN-MADD$_{\rm sc}$, respectively, then $|R_n(\delta_{\rm sc}) - R_n(\delta)| \overset{P}{\to} 0$ as $n \to \infty$.
\label{theorem}
\end{proposition}

\begin{proof}
Let $j^\ast = \argmin_{j=1,\ldots,J} \widetilde\rho(\vec Z, \sp X_j)$ denote the class to which NN-MADD classifies $\vec Z$. Let $j\neq j^\ast$ denote any other class. Then, 
\begin{align*}
\widetilde \rho_{\mathrm{sc}}(\vec Z,\sp X_j)-\widetilde \rho_{\mathrm{sc}}(\vec Z,\sp X_{j^\ast}) = \big(\widetilde \rho_{\mathrm{sc}}(\vec Z,\sp X_j)-\widetilde \rho(\vec Z,\sp X_j)\big) + \big(\widetilde \rho(\vec Z,\sp X_j)-\widetilde \rho(\vec Z,\sp X_{j^\ast})\big) + \big(\widetilde \rho(\vec Z,\sp X_{j^\ast})-\widetilde \rho_{\mathrm{sc}}(\vec Z,\sp X_{j^\ast})\big).
\end{align*}
From the definition of $\widetilde \beta_n(\vec Z)$, we have $\bigl|\widetilde \rho(\vec Z,\sp X_j)-\widetilde \rho_{\mathrm{sc}}(\vec Z,\sp X_j)\bigr| \le \widetilde \beta_n(\vec Z)$ for every $j=1,\ldots,J$. Hence, 
\[
\widetilde \rho_{\mathrm{sc}}(\vec Z,\sp X_j)-\widetilde \rho_{\mathrm{sc}}(\vec Z,\sp X_{j^\ast}) 
\ge -2\widetilde \beta_n(\vec Z)
+\bigl(\widetilde \rho(\vec Z,\sp X_j)-\widetilde \rho(\vec Z,\sp X_{j^\ast})\bigr).
\]
Taking minimum over all $j\neq j^\ast$, we obtain
\[
\min_{j\neq j^\ast} \Bigl\{\widetilde \rho_{\mathrm{sc}}(\vec Z,\sp X_j)-\widetilde \rho_{\mathrm{sc}}(\vec Z,\sp X_{j^\ast}) \Bigr\} \ge - 2\widetilde \beta_n(\vec Z) + \min_{j\neq j^\ast} \Bigl\{\widetilde \rho(\vec Z,\sp X_j)-\widetilde \rho(\vec Z,\sp X_{j^\ast})\Bigr\} = -2\widetilde\beta_n(\vec Z) + \widetilde\gamma_n(\vec Z).
\]
Now, NN-MADD$_{\rm sc}$ would classify $\vec Z$ to class $j^\ast$ if the left side is positive. Thus, $\pr\big(\delta_{\rm sc}(\vec Z) \ne j^\ast\big) \le \pr\big(\widetilde\gamma_n(\vec Z) \le 2\widetilde\beta_n(\vec Z)\big)$. Since $\widetilde\gamma_n(\vec Z) \overset{P}{\to} \gamma >0$ and $\widetilde\beta_n(\vec Z) \overset{P}{\to} 0$ as $n \to \infty$, it follows that $\pr\big(\widetilde\gamma_n(\vec Z) \le 2\widetilde\beta_n(\vec Z)\big) \to 0$ as $n \to \infty$. Since there are finitely many classes, this completes the proof of the first part of the result. The second part of the result is a direct consequence of the first part.
\end{proof}

While $\widetilde\beta_n$ quantifies how well the scalable version $\widetilde\rho_{\rm sc}$ approximates $\widetilde\rho$, $\widetilde\gamma_n$ quantifies the ability of $\widetilde\rho$ to distinguish between classes. In this sense, the two terms can be viewed as the \emph{noise} of the scalable version and \emph{signal} of the full version, respectively. Our result shows that when the signal-to-noise ratio is large enough, then the behavior of the scalable version matches the full version. This type of result is not surprising. Rather, the utility of the result lies in identifying the signal and noise terms themselves, giving a precise understanding of the method.

For our result, we have assumed that the signal $\widetilde\gamma_n(\vec Z)$ converges to a positive quantity. This type of convergence is well known for $\widetilde\rho$ in the HDLSS regime, which is a ramification of the class separation property of MADD \citep{pal2016high, sarkar2019perfect}. However, a closer look into the proof reveals that the result continues to hold even when $\widetilde\gamma_n(\vec Z)$ converges to $0$, provided that the rate is slower than that of $\widetilde\beta_n(\vec Z)$. In fact, our result continues to hold as long as $\pr(\widetilde\gamma_n(\vec Z) > 2 \widetilde\beta_n(\vec Z)) \to 1$ as $n \to \infty$. This precisely quantifies the allowable signal-to-noise ratio for our proposed method.

The assumption $\widetilde\beta_n(\vec Z) \overset{P}{\to} 0$ as $n \to \infty$, requiring the maximum difference between $\widetilde\rho$ and $\widetilde\rho_{\rm sc}$ to be negligible, is admittedly strong. Note that, for a fixed set of observations, $\widetilde\rho_{\rm sc}$ can be viewed as the \emph{sample analogue} of $\widetilde\rho$, where \emph{sampling} refers to the selection of the representative observations. In fact, $\widetilde\rho_{\rm sc}$ can be viewed as the sample mean, whereas $\widetilde\rho$ serves as the population mean. Therefore, techniques from survey sampling can be used to tackle $\widetilde\beta_n(\vec Z)$. Indeed, in Appendix~\ref{appendix:maths} we show that
\begin{align}\label{eq:beta_n_bound}
\widetilde\beta_n(\vec Z) \le \sup_{1\le j \le J} \sup_{1 \le i \le n_j} \big|\widetilde{\mu}_{ji}^p(\vec Z) - \widetilde{\rho}(\vec Z,\vec X_{ji})\big| + \bigO_{P}\left(\sqrt{\frac{\log n}{k^\ast}}\right),
\end{align}
where $\widetilde\mu^p_{ji}(\vec Z) = {k^\ast}^{-1} \sum_{j'=1}^J\sum_{i'=1}^{n_{j'}} p_{j'i'} \big|\|\vec Z-\vec X_{j'i'}\| - \|\vec X_{ji} - \vec X_{j'i'}\|\big|$ with $p_{j'i'} = \pr(\vec X_{j'i'} \in \sp X^\ast | \sp X)$ denoting the inclusion probability of $\vec X_{j'i'}$ in the representative set $\sp X^\ast$. Here, $k^\ast = |\sp X^\ast| = \sum_{j=1}^J k_j$ is the size of the representative set $\sp X^\ast$. We also derive conditions under which the first part converges to $0$ in probability. But, those conditions are admittedly strong, so we do not report them here. Instead, we numerically probe the validity of the condition by calculating $\widetilde\beta_n(\vec Z)$ in all our simulation examples under varying sample sizes and dimensions. For each setting, we generate the training observations, and compute $\widetilde\beta_n(\vec Z)$ for $5000$ randomly generated test observations $\vec Z$. The entire process is replicated $5$ times, giving us $25000$ values of $\widetilde\beta_n(\vec Z)$. In Figure~\ref{fig:assumption}, we show boxplots of these $25000$ $\widetilde\beta_n(\vec Z)$ values for Example~\ref{example2} (where the underlying populations are mixtures of Gaussian distributions) and Example~\ref{example5} (where the underlying populations differ in their locations and scales). The results for the other examples are given in Appendix~\ref{appendix: Simulation}. From the plots, it can be seen that the $\widetilde\beta_n(\vec Z)$ values decrease as $n$ and $d$ increase simultaneously. The same pattern is observed across all our simulation examples. The plots also hint at the allowable rates of growth of $n$ and $d$. However, the precise mathematical relationship is, admittedly, not clear to us at the moment.

\begin{figure}[t]
    \small
    \centering
    \begin{tabular}{c c}
    (a) Example~\ref{example2} (mixture of Gaussians) & (b) Example~\ref{example5} (location-scale problem)\\
    \includegraphics[width=0.48\textwidth]{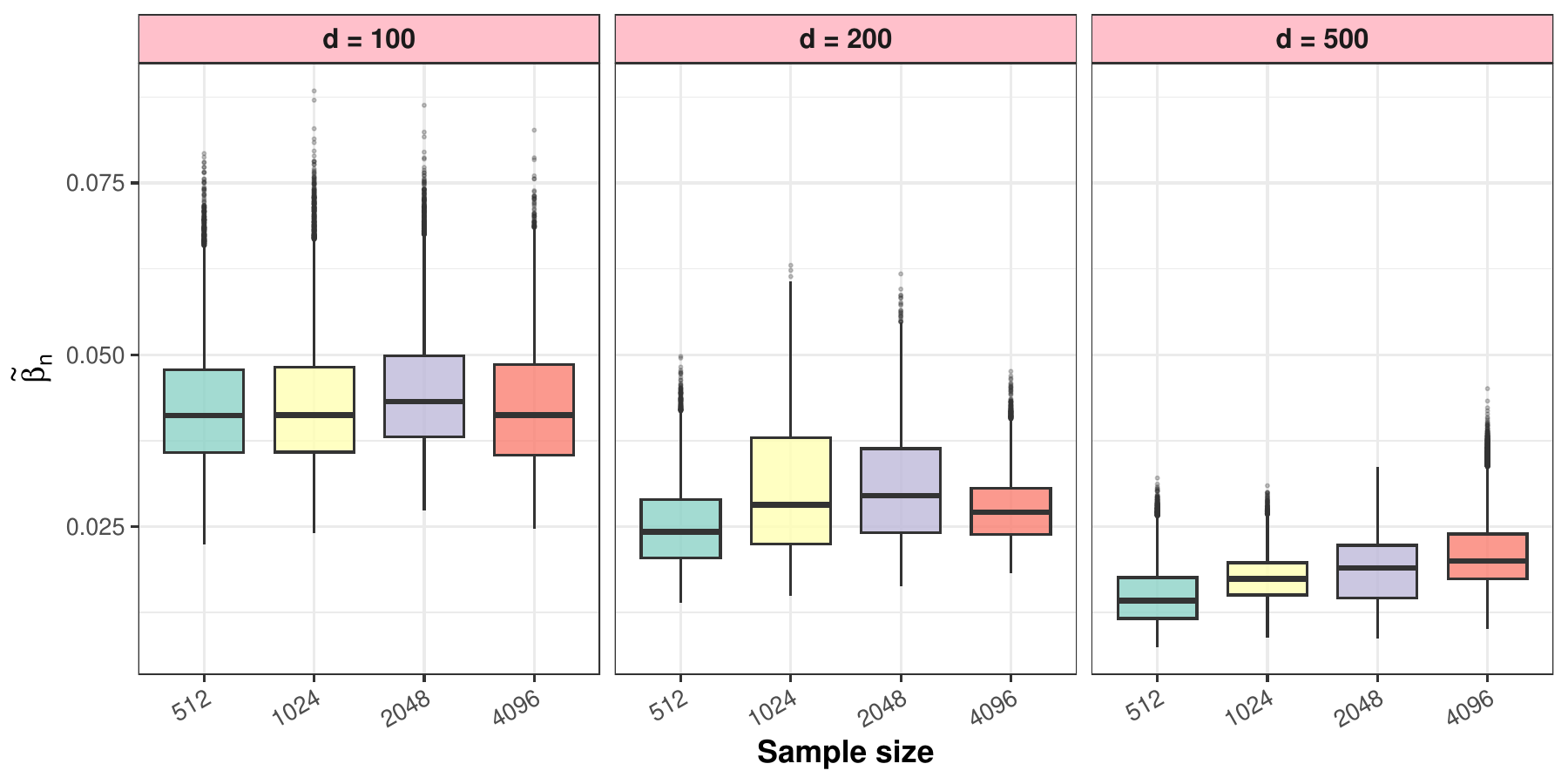} &
    \includegraphics[width=0.48\textwidth]{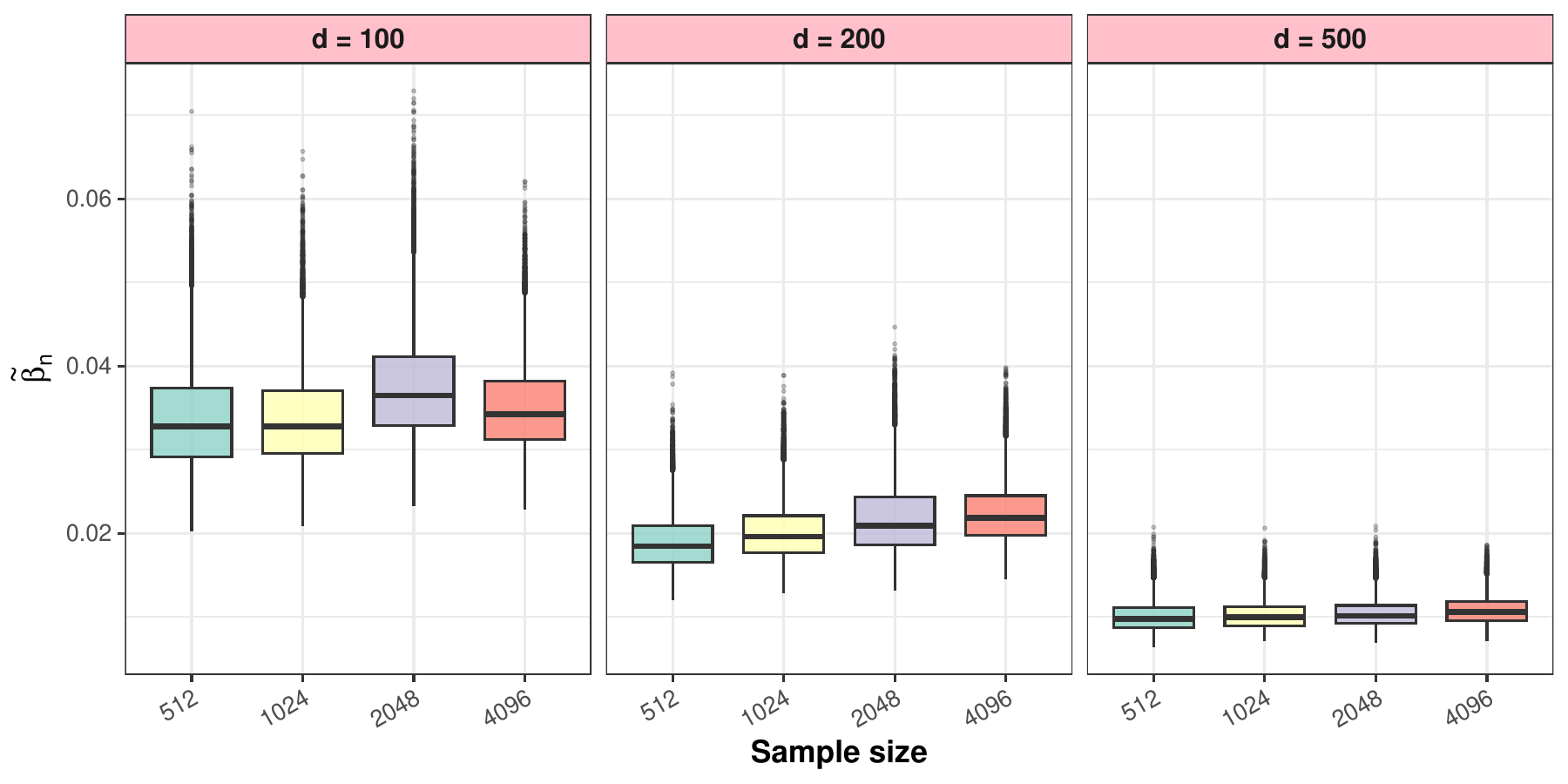}
    \end{tabular}
    \caption{Boxplots of $\widetilde\beta(\vec Z)$ values in Examples~\ref{example2} and \ref{example5} for varying values of $n$ and $d$.\label{fig:assumption}}
\end{figure}

\begin{remark}\label{remark:representative_sample_size}
    Eq.~\eqref{eq:beta_n_bound} gives a bound on the representative sample size $k^\ast$, viz., $k^\ast \gtrsim \log n$. Also, in Appendix~\ref{appendix:maths}, we show that the first term in Eq.~\ref{eq:beta_n_bound} converges to $0$ if $k_j/k^\ast - n_j/n \to 0$ as $n \to \infty$. Our empirical choice of $k_j \propto \sqrt{d} (n_j/n) \log n_j$, along with $n_j/n \to \alpha_j \in (0,1)$, ensures that $k^\ast \asymp \sqrt{d}\log n$ as $n \to \infty$. This choice also ensures that $k_j/k^\ast \to \alpha_j$, and thus $k_j/k^\ast - n_j/n \to 0$ as $n \to \infty$. The factor $\sqrt{d}$ is used as we particularly focus on situations where $n$ is much larger than $d$. In situations where $d$ is also very large, other factors like $\log(d)$ can be used as well. In our simulations, the different factors of $d$ did not make significant differences in the results.
\end{remark}

Here, we only study the theoretical behavior of NN-MADD$_{\rm sc}$. Similar results can be established for NN-gMADD$_{\rm sc}$ as well. However, for the sake of brevity, we do not pursue those results here.

\section{Concluding Remarks}
\label{sec:Conclusion}
In this paper, we propose a scalable version of the semi-metric mean absolute difference of distances (MADD) and its generlized version gMADD. Both MADD and gMADD are useful in high-dimensional settings where the usual Euclidean distance often suffers from distance concentration. However, their heavy computational costs for large sample sizes restrict their applicability to relatively small-sized training datasets. To address this issue, we use determinantal point processes (DPP) to select a diverse set of representative points from the training data, which are then used to construct scalable versions of MADD and gMADD. A limitation of the DPP approach is that it requires the eigen-decomposition of the similarity matrix, which can be computationally expensive for very large datasets. We employ random Fourier features (RFF) to obtain an efficient approximation of the similarity matrix, which significantly reduces the computational burden of the scalable version of MADD. However, due to technical reasons, the RFF-based approach could not be generalized to reduce the computational complexity of gMADD. Reducing the computational complexity of the scalable version of gMADD for large training datasets, therefore, remains an open direction for future research.

In this work, we study the utility of the proposed scalable version of MADD for classification problems only. However, MADD is shown to perform well in other high-dimensional problems as well, like clustering, two-sample testing. A systematic study of scalable versions of MADD for these problems may also be a possible direction of future research. The utility of the proposed method is demonstrated via extensive simulation studies. We also provide some theoretical justification behind the behavior of the proposed method. But, our theoretical analysis is admittedly incomplete. A more rigorous theoretical study of the proposed method may be considered in future research.

\section*{Acknowledgment}
The research of Soham Sarkar was partially supported by the INSPIRE Faculty Fellowship from the Department of Science and Technology, Government of India.

\appendix
\section*{Appendices}
The appendices contain algorithms referenced in the main text, additional mathematical details, and additional empirical results. The algorithms are presented in Appendix~\ref{appendix: Algorithms} and the mathematical details are in Appendix~\ref{appendix:maths}. The additional numerical results are shown in two parts. In Appendix~\ref {appendix: Simulation}, we  report results from additional simulation studies. In Appendix~\ref{appendix:Real}, we show the results from the analyses of the benchmarks datasets omitted in the main text.

\section {Algorithms}
\label{appendix: Algorithms}

Algorithm~\ref{alg:DPP} describes how to draw a sample from a standard DPP, which we refer to as the naive DPP in the main paper. This algorithm is taken from \cite{kulesza2012determinantal}, which we state here for completeness.

\noindent\rule{\linewidth}{0.4pt}
\refstepcounter{algorithm}
\noindent\textbf{Algorithm~\thealgorithm} Naive DPP algorithm\par\vspace{-0.25cm}
\noindent\rule{\linewidth}{0.4pt}\vspace{-0.5cm}
\label{alg:DPP}
\begin{algorithmic}[1]
\small

\Require $\mat L$ -- Similarity matrix of size $N \times N$.
\smallskip
\Ensure $\boldsymbol{\pi}$ -- Selected subset from $\{1,\ldots,N\}$.
\medskip

\Procedure{Eigen-decomposition of $\mat L$}{}
\State Find eigen-pairs $\{(\lambda_1,\vec v_1),\ldots,(\lambda_N,\vec v_N)\}$, where $\lambda_1,\ldots,\lambda_N$ are the eigenvalues and $\vec v_1,\ldots,\vec v_N$ are the corresponding orthonormal eigenvectors. 
\EndProcedure
\medskip

\Procedure{Sample eigenvectors}{}
\State \textbf{Initialize:} $J \gets \emptyset$
\For{$n = 1,2,\ldots,N$}
    \State $J \gets J \cup \{n\}$ with probability $\dfrac{\lambda_n}{\lambda_n + 1}$
\EndFor
\State $V \gets \{\vec v_n\}_{n \in J}$.
\EndProcedure
\medskip

\Procedure{Sample items}{}
\State $\boldsymbol{\pi} \gets \emptyset$.
\While{$|V| > 0$}
    \State Select $i \in \{1,\dots,N\}$ with probability $p(i) = \dfrac{1}{|V|} \sum_{\vec v \in V} (\vec v^\top \vec e_i)^2$, where $\vec e_i$ is the $i$-th canonical basis vector in $\mathbb R^N$ with $1$ in the $i$-th position, and $0$ elsewhere.
    \State $\boldsymbol{\pi} \gets \boldsymbol{\pi} \cup \{i\}$
    \State Replace $V$ with an orthonormal basis for the subspace of $V$ orthogonal to $\vec e_i$.
\EndWhile
\EndProcedure
\medskip

\State \Return $\boldsymbol{\pi}$
\end{algorithmic}
\noindent\rule{\linewidth}{0.4pt}

Algorithm~\ref{alg:k-DPP} describes the procedure for sampling from a $k$-DPP. The algorithm is similar to the one given in \cite{kulesza2012determinantal}. However, we also use Random Fourier Features (RFFs) to further reduce the computational complexity for larger datasets. Therefore, we modify the algorithm to take either the full similarity matrix $\mat L$ or an RFF matrix $\mat R$ as input. The rest of the algorithm is modified accordingly.

\noindent\rule{\linewidth}{0.4pt}
\refstepcounter{algorithm}
\noindent\textbf{Algorithm~\thealgorithm} Sampling from a $k$-DPP \par\vspace{-0.25cm}
\noindent\rule{\linewidth}{0.4pt}\vspace{-0.5cm}
\label{alg:k-DPP}
\begin{algorithmic}[1]
\small

\Require $\mat L$ -- A similarity matrix of size $N\times N$, or\newline
        $\mat R$ -- An RFF matrix of size $N\times M$ such that $\mat L \approx \mat R\mat R^\top$.\newline
         $k$ -- Number of samples to be selected.
\smallskip
\Ensure $\boldsymbol{\pi}=(\pi_1,\ldots,\pi_k)$ -- Selected subset of size $k$ from $\{1,\ldots,N\}$.
\smallskip

\If{the input is the similarity matrix $\mat L$}
    \State Perform eigen-decomposition of $\mat L$. Find eigen-pairs $\{(\lambda_1,\vec v_1),\ldots,(\lambda_N,\vec v_N)\}$, where $\lambda_1,\ldots,\lambda_N$ are the eigenvalues and $\vec v_1,\ldots,\vec v_N$ are the corresponding orthonormal eigenvectors.
    \State Set $M \gets N$.
\ElsIf{the input is the RFF matrix $\mat R$}
    \State Compute the singular value decomposition of $\mat R$. Let $\sigma_1,\ldots,\sigma_M$ be the singular values and $\vec v_1,\ldots,\vec v_M$ be the corresponding orthonormal left singular vectors. Obtain the eigen-pairs $\{(\lambda_1,\vec v_1),\ldots,(\lambda_M,\vec v_M)\}$, where $\lambda_i = \sigma_i^2$.
\EndIf
\medskip

\Procedure{Compute elementary symmetric polynomials}{}
    \State \textbf{Initialize:} $E$ -- an $(M+1)\times(k+1)$ matrix of zeros.
    \State $E[n,1] \gets 1$, for $n=1,2,\ldots,M+1$.
    \For{$\ell = 2,3,\ldots,k+1$}
        \For{$n = 2,3,\ldots,M+1$}
            \State $E[n,\ell] \gets E[n-1,\ell] + \lambda_{n-1} E[n-1,\ell-1]$.
        \EndFor
    \EndFor
\EndProcedure
\medskip

\Procedure{Sample $k$ eigenvectors}{}
    \State \textbf{Initialize:} $J \gets \emptyset$, $\ell \gets k$.
    \For{$n = M,M-1,\ldots,1$}
        \If{$\ell = 0$}
            \State break
        \EndIf
        \State Draw $u \sim \mathrm{Uniform}(0,1)$.
        \If{$u < \dfrac{\lambda_n E[n,\ell]}{E[n+1,\ell+1]}$}
            \State $J \gets J \cup \{n\}$.
            \State $\ell \gets \ell - 1$.
        \EndIf
    \EndFor
    \State $V \gets \{\vec v_n\}_{n\in J}$.
\EndProcedure
\medskip

\Procedure{Sample $k$ items}{}
    \State $\boldsymbol{\pi} \gets \emptyset$.
    \While{$|V| > 0$}
        \State Select $i \in \{1,\ldots,N\}$ with probability $p(i)=\dfrac{1}{|V|}\sum_{\vec v\in V}(\vec v^\top \vec e_i)^2$, where $\vec e_i$ is the $i$-th canonical basis vector in $\mathbb R^N$.
        \State $\boldsymbol{\pi} \gets \boldsymbol{\pi}\cup\{i\}$.
        \State Replace $V$ by an orthonormal basis for the subspace of $V$ orthogonal to $\vec e_i$.
    \EndWhile
\EndProcedure
\smallskip

\State \Return $\boldsymbol{\pi}$.

\end{algorithmic}
\noindent\rule{\linewidth}{0.4pt}

\bigskip

Algorithm~\ref{alg:k-DPP_modified} presents a deterministic version of the $k$-DPP algorithm, which is used in our proposed incremental cross-validation technique (see Section~\ref{sec:cross-validation} in the main paper).

\noindent\rule{\linewidth}{0.4pt}
\refstepcounter{algorithm}
\noindent\textbf{Algorithm~\thealgorithm} Sampling from a Deterministic $k$-DPP \par\vspace{-0.25cm}
\noindent\rule{\linewidth}{0.4pt}\vspace{-0.5cm}
\label{alg:k-DPP_modified}
\begin{algorithmic}[1]
\small

\Require $\mat L$ -- A similarity matrix of size $N\times N$, or\newline
        $\mat R$ -- An RFF matrix of size $N\times M$ such that $\mat L \approx \mat R\mat R^\top$.\newline
         $k$ -- Number of samples to be selected.
\medskip
\Ensure $\boldsymbol{\pi} = (\pi_1,\ldots,\pi_k)$ -- Selected subset of size $k$ from $\{1,\ldots,N\}$.
\medskip

\If{the input is the similarity matrix $\mat L$}
    \State Perform eigen-decomposition of $\mat L$. Find eigen-pairs $\{(\lambda_1,\vec v_1),\ldots,(\lambda_N,\vec v_N)\}$, where $\lambda_1,\ldots,\lambda_N$ are the eigenvalues and $\vec v_1,\ldots,\vec v_N$ are the corresponding orthonormal eigenvectors.
    \State Set $M \gets N$.
\ElsIf{the input is the RFF matrix $\mat R$}
    \State Compute the singular value decomposition of $\mat R$. Let $\sigma_1,\ldots,\sigma_M$ be the singular values and $\vec v_1,\ldots,\vec v_M$ be the corresponding orthonormal right singular vectors. Obtain the eigen-pairs $\{(\lambda_1,\vec v_1),\ldots,(\lambda_M,\vec v_M)\}$, where $\lambda_i = \sigma_i^2$.
\EndIf
\medskip

\Procedure{Select $k$ eigenvectors}{}
\State Let $\lambda_{(1)} \ge \lambda_{(2)} \ge \cdots \ge \lambda_{(M)}$ be the eigenvalues sorted in decreasing order
\State $J \gets \{ i : \lambda_i \in \{\lambda_{(1)}, \dots, \lambda_{(k)}\} \}$
\State $V \gets \{\vec v_n\}_{n \in J}$.
\EndProcedure
\medskip

\Procedure{Sample $k$ items}{}
\State $\boldsymbol{\pi} \gets \emptyset$
\While{$|V| > 0$}
    \State Compute $p(i) = \dfrac{1}{|V|} \sum_{\vec v \in V} (\vec v^\top \vec e_i)^2$, for $i=1,2,\dots,N$,  where $\vec e_i$ is a vector with $1$ in the $i$-th position, and $0$ elsewhere.   
    \State Select $i^\ast = \arg\max_{1 \le i \le N} p(i)$
    \State $\boldsymbol{\pi} \gets \boldsymbol{\pi} \cup \{i^\ast\}$
    \State Replace $V$ with an orthonormal basis for the subspace of $V$ orthogonal to $\vec e_{i^\ast}$
\EndWhile
\EndProcedure
\medskip

\State \Return $\boldsymbol{\pi}$
\end{algorithmic}
\noindent\rule{\linewidth}{0.4pt}

\bigskip
In the next algorithm (Algorithm~\ref{alg:SGMNN}), we present nearest neighbor classification technique based on the scalable version of gMADD. The details are in Section~\ref{sec:gmadd} of the main text.

\noindent\rule{\linewidth}{0.4pt}
\refstepcounter{algorithm}
\noindent\textbf{Algorithm~\thealgorithm} Nearest Neighbor Classification Using Scalable gMADD \par\vspace{-0.25cm}
\noindent\rule{\linewidth}{0.4pt}\vspace{-0.5cm}
\label{alg:SGMNN}
\begin{algorithmic}[1]
\small

	\Require $\sp X_1,\ldots,\sp X_J$ -- Training sample, where $\sp X_j = \{\vec X_{j1},\ldots,\vec X_{jn_j}\}$.\newline
	$k_1,\ldots,k_J$ -- Number of representative samples to be selected.\newline
	$\phi,\gamma$ -- Non-negative, continuous and monotonically increasing functions with $\phi(0)=\gamma(0)=0$.\newline
	$\vec Z$ -- Test observation.
	\medskip
	\Ensure $\delta^{h,\psi}_{\rm sc}(\vec Z)$ -- The predicted class label for $\vec Z$.
	\medskip

	\For {$j \in \{1,\dots, J\}$}
	\State Compute $h^{\phi,\gamma}(\vec X_{ji},\vec X_{ji^\prime})$, $i,i^\prime=1,\ldots,n_j$.
	\State Find $\sigma_j^2 = \operatorname{median}\{h^{\phi,\gamma}(\vec X_{ji},\vec X_{ji^\prime}): 1 \le i < i^\prime \le n_j\}$.
	\State Compute $l^{(j)}_{ii^\prime}(\phi,\gamma) =	\exp\{-h^{\phi,\gamma}(\vec X_{ji},\vec X_{ji^\prime})/(2\sigma_j^2)\},$ $i,i^\prime = 1,\ldots,n_j.$ 
    Construct $\mat L^{(j)}_{\phi,\gamma} = ((l^{(j)}_{ii^\prime}(\phi,\gamma)))$.
	\State Select a subset $\{\pi_1,\ldots,\pi_{k_j}\}$ of $\{1,\ldots,n_j\}$ using $k_j$-DPP with the matrix $\mat L^{(j)}_{\phi,\gamma}$.
	\State The representative observations from class-$j$ are $\sp X_j^\ast = \{\vec X_{j\pi_1},\ldots,\vec X_{j\pi_{k_j}}\}$.
	\EndFor

	\State Define $\sp X^\ast = \sp X_1^\ast \cup \cdots \cup \sp X_J^\ast$ to be the representative set.
	\State Compute $\rho^{\phi,\gamma}_{\rm sc}(\vec Z,\vec X_{ji})$, $j=1,\ldots,J$, $i=1,\ldots,n_j$, using \eqref{eq:gMADD_sc}.
	\State For $j=1,\ldots,J$, compute $\rho^{\phi,\gamma}_{\rm sc}(\vec Z,\sp X_j) = \min_{i=1,\ldots,n_j}
	\rho^{\phi,\gamma}_{\rm sc}(\vec Z,\vec X_{ji}).$
	\smallskip
	\State Return $\delta^{\phi,\gamma}_{\rm sc}(\vec Z) = \argmin_{j=1,\ldots,J} \rho^{\phi,\gamma}_{\rm sc}(\vec Z,\sp X_j).$
	\end{algorithmic}

\noindent\rule{\linewidth}{0.4pt}

\FloatBarrier

\section{Additional Mathematical Details}
\label{appendix:maths}

In the theoretical analysis of our proposed method (Section~\ref{sec:Theorem} of the main paper), we assumed that for any test observation $\vec Z$, $\widetilde\beta_n(\vec Z) := \sup_{1\le j\le J}\ \sup_{1\le i\le n_j} \big| \widetilde\rho_{\mathrm{sc}}(\vec Z,\vec X_{ji}) - \widetilde\rho(\vec Z,\vec X_{ji})\big|$ converges to $0$ in probability. Here, we provide some justification for this assumption. Throughout our discussion, we assume that there exists a constant $C>0$ such that for any two independent training observations $\vec X_{ji}$ and $\vec X_{j^\prime i^\prime}$, and any test observation $\vec Z$, $\|\vec X_{ji} - \vec X_{j^\prime i^\prime}\| \le C\sqrt{d_n}$ and $\|\vec Z - \vec X_{ji}\| \le C\sqrt{d_n}$. This holds if, for example, all the components of the underlying populations are uniformly bounded. As a consequence, we get that $\|\vec X_{ji}-\vec X_{j'i'}\| = \bigO_{P}(\sqrt{d_n})$ and $\|\vec Z-\vec X_{ji}\| = \bigO_{P}(\sqrt{d_n})$. This is a common phenomena in high-dimensional regimes, where the Euclidean distance scales as the square-root of the dimension \citep{hall2005geometric, sarkar2019perfect, roy2022generalizations}. 
\medskip
For ease of exposition, we define
\[
A_{j'i'}(\vec Z, ji) =
\begin{cases}
\big|\|\vec Z - \vec X_{j'i'}\| -  \|\vec X_{ji} - \vec X_{j' i'}\|\big|  & \text{if } (j',i') \neq (j,i), \\
0   & \text{if } (j',i') = (j,i).
\end{cases}
\]
Note that $0\le A_{j'i'}(\vec Z, ji)\le \|\vec Z - \vec X_{ji}\|$, thus $A_{j'i'}(\vec Z, ji) \le C\sqrt{d_n}$. So, the scaled quantity $\widetilde{A}_{j' i'}(\vec Z, ji) = d_n^{-1/2} A_{j' i'}(\vec Z, ji)$ is bounded by $C$, uniformly over $j,j'$, $i,i'$ and $\vec Z$. We also define the indicator random variables $\xi_{ji} = \mathbf{1}\{\vec X_{ji} \in \sp X^\ast\}$, indicating whether $\vec X_{ji}$ is included in the representative set $\sp X^\ast$. Notice that
\[
\widetilde\rho_{\rm sc}(\vec Z,\vec X_{ji}) = \frac{1}{k^\ast} \sum_{j'=1}^J \sum_{i'=1}^{n_{j'}} \widetilde{A}_{j' i'}(\vec Z,ji)\,\xi_{j' i'} \quad\text{and}\quad \widetilde\rho(\vec Z,\vec X_{ji}) = \frac{1}{n-1} \sum_{j'=1}^{J}\sum_{i'=1}^{n_{j'}} \widetilde A_{j'i'}(\vec Z,ji).
\]
Let $p_{ji}=\mathbb{E}[\xi_{ji}\mid \sp X] = \pr(\vec X_{ji} \in \sp X^\ast \mid \sp X)$ denote the \emph{inclusion probability} of $\vec X_{ji}$ in $\sp X^\ast$. Define
\begin{equation*}
\widetilde{\mu}_{ji}^{p}(\vec Z) = \dfrac{1}{k^\ast}\sum_{j'=1}^J \sum_{i'=1}^{n_{j'}} p_{j' i'}\widetilde{A}_{j' i'}(\vec Z, ji).
\end{equation*}
With the above defined quantities, we get the following upper bound on $\widetilde\beta_n(\vec Z)$.

\begin{lemma}
\label{lemma:1}
Consider the above setup, and let $\widetilde\beta_n(\vec Z)$ and $\widetilde\mu_{ji}^p(\vec Z)$ be as defined above. Then,
\[
\widetilde{\beta}_n(\vec Z) \le \sup_{1\le j\le J}\ \sup_{1\le i\le n_j} \bigl|\widetilde{\mu}_{ji}^p(\vec Z)-\widetilde{\rho}(\vec Z,\vec X_{ji})\bigr| + \bigO_{P}\left(\sqrt{\frac{\log n}{k^\ast}}\right).
\]
\end{lemma}

\begin{proof}
By the triangle inequality,
\begin{align*}
\widetilde{\beta}_n(\vec Z) &= \sup_{1\le j\le J}\ \sup_{1\le i\le n_j} \big| \widetilde\rho_{\mathrm{sc}}(\vec Z,\vec X_{ji}) - \widetilde\rho(\vec Z,\vec X_{ji})\big|\\
&\le \sup_{1\le j\le J}\ \sup_{1\le i\le n_j} \bigl|\widetilde{\rho}_{\mathrm{sc}}(\vec Z,\vec X_{ji})-\widetilde{\mu}_{ji}^p(\vec Z)\bigr| + \sup_{1\le j\le J}\ \sup_{1\le i\le n_j} \bigl|\widetilde{\mu}_{ji}^p(\vec Z)-\widetilde{\rho}(\vec Z,\vec X_{ji})\bigr|.
\end{align*}
So, the result follows if we show $\sup_{1\le j\le J}\ \sup_{1\le i\le n_j} \bigl|\widetilde{\rho}_{\mathrm{sc}}(\vec Z,\vec X_{ji})-\widetilde{\mu}_{ji}^p(\vec Z)\bigr| = \bigO_{P}\big(\sqrt{\log n/k^\ast}\big)$. By definition of $\widetilde\rho_{\rm sc}$, since $k^\ast$ is non-random, it follows that $\E[\widetilde\rho_{\rm sc}(\vec Z,\vec X_{ji}) \mid \sp X] = \widetilde\mu_{ji}^p(\vec Z)$ almost surely.

For $j'=1,\ldots,J$, define $\parvec\xi_{j'} = (\xi_{j'i'})_{i'=1,\ldots,n_{j'}}$ to be the indicator vector associated with the $j'$-th training sample. Recall that the reference observations from class-$j'$ are selected using $k_{j'}$-DPP. Hence, the associated indicator vector $\parvec\xi_{j'}$ is strongly Rayleigh \citep[see][]{BBL2009}. Also, since $\sum_{i'=1}^{n_{j'}} \xi_{j'i'} = k_{j'}$ almost surely, the random vector $\parvec\xi_{j'}$ is $k_{j'}$-homogeneous.
\bigskip

Let $w_{j'i'} = \widetilde{A}_{j' i'}(\vec Z, ji)/C \in [0,1]$, and define
\[
g(\parvec\xi_{j'}) := \frac{1}{2C} \sum_{i'=1}^{n_{j'}} \widetilde{A}_{j'i'}(\vec Z,ji)\,\xi_{j'i'} = \frac{1}{2} \sum_{i'=1}^{n_{j'}} w_{j'i'}\,\xi_{j'i'}. 
\]
Now, consider two vectors $\parvec\xi_{j'} = (\xi_{j'i'})_{i'=1,\ldots,n_{j'}}$ and $\parvec\xi_{j'}' = (\xi'_{j'i'})_{i'=1,\ldots,n_{j'}}$. Then, 
\begin{align*}
|g(\parvec\xi_{j'})-g(\parvec\xi_{j'}')| \le \frac{1}{2} \sum_{i'=1}^{n_{j'}} w_{j' i'} |\xi_{j' i'}-\xi'_{j'i'}| \le \frac{1}{2} \sum_{i'=1}^{n_{j'}} |\xi_{j'i'}-\xi'_{j'i'}| = \|\parvec\xi_{j'}-\parvec\xi'_{j'}\|_{\mathrm{TV}},
\end{align*}
where $\|\cdot\|_{\rm TV}$ denotes the total variation distance with respect to counting measure \citep[cf.][]{PemantlePeres2014}. Thus, $g$ is $1$-Lipschitz with respect to $\|\cdot\|_{\rm TV}$. Hence, using concentration results from \cite{PemantlePeres2014}, we get for any $a>0$,
\[
\pr\left(\left|g(\parvec\xi_{j'})-\E\left[g(\parvec\xi_{j'})\mid \sp X\right]\right| \ge a \mid \sp X \right) \le 2 \exp\left(-\frac{a^2}{8k_{j'}}\right) \text{ almost surely}.
\]
Since $g(\parvec\xi_{j'}) = (2C)^{-1} \sum_{i'=1}^{n_{j'}} \widetilde A_{j'i'}(\vec Z,ji)\,\xi_{j'i'}$, it follows that
\[
\pr\left(\left|\sum_{i'=1}^{n_{j'}} \widetilde A_{j'i'}(\vec Z,ji)\,\xi_{j'i'} - \sum_{i'=1}^{n_{j'}} \widetilde A_{j'i'}(\vec Z,ji) p_{j'i'}\right| > a \mid \sp X \right) \le 2 \exp\left(-\frac{a^2}{32C^2 k_{j'}}\right) \text{ almost surely}.
\]
Consequently,
\begin{align*}
    &\pr\left(\left|\sum_{j'=1}^J\sum_{i'=1}^{n_{j'}} \widetilde A_{j'i'}(\vec Z,ji)\,\xi_{j'i'} - \sum_{j'=1}^J\sum_{i'=1}^{n_{j'}} \widetilde A_{j'i'}(\vec Z,ji) p_{j'i'}\right| > a \mid \sp X \right) \\
    &\le \sum_{j'=1}^J \pr\left(\left|\sum_{i'=1}^{n_{j'}} \widetilde A_{j'i'}(\vec Z,ji)\,\xi_{j'i'} - \sum_{i'=1}^{n_{j'}} \widetilde A_{j'i'}(\vec Z,ji) p_{j'i'}\right| > \frac{a}{J} \mid \sp X \right) \\
    &\le 2\sum_{j'=1}^J \exp\left(-\frac{a^2}{32C^2J^2k_{j'}}\right).
\end{align*}
Since the right side is free of $\sp X$, it follows that the unconditional probability also has the same upper bounded. Now,
\begin{align*}
    \pr\left(\left|\widetilde\rho_{\rm sc}(\vec Z,\vec X_{ji}) - \widetilde\mu_{ji}^p(\vec Z)\right| > a\right) &= \pr\left(\left|\sum_{j'=1}^J\sum_{i'=1}^{n_{j'}} \widetilde A_{j'i'}(\vec Z,ji)\,\xi_{j'i'} - \sum_{j'=1}^J\sum_{i'=1}^{n_{j'}} \widetilde A_{j'i'}(\vec Z,ji) p_{j'i'}\right| > k^\ast a\right) \\
    &\le 2\sum_{j'=1}^J \exp\left(-\frac{{k^\ast}^2 a^2}{32C^2J^2k_{j'}}\right) \le 2J \exp\left(-\frac{k^\ast a^2}{32C^2J^2}\right),
\end{align*}
where the last inequality follows since $k_{j'} < k^\ast$. Using the union bound over $j=1,\ldots,J$ and $i=1,\ldots,n_j$, we get
\[
\pr\left(\sup_{1\le j \le J} \sup_{1 \le i \le n_j} \left|\widetilde\rho_{\rm sc}(\vec Z,\vec X_{ji}) - \widetilde\mu_{ji}^p(\vec Z)\right| > a\right) \le 2nJ \exp\left(-\frac{k^\ast a^2}{32C^2J^2}\right).
\]
Taking $a = 8CJ\sqrt{\log n/k^\ast}$, we get
\[
\pr\left(\sup_{1\le j \le J} \sup_{1 \le i \le n_j} \left|\widetilde\rho_{\rm sc}(\vec Z,\vec X_{ji}) - \widetilde\mu_{ji}^p(\vec Z)\right| > 8CJ \sqrt{\frac{\log n}{k^\ast}}\right) \le \frac{2J}{n} \to 0 \text{ as } n \to \infty.
\]
This shows that $\sup_{1\le j \le J} \sup_{1 \le i \le n_j} \left|\widetilde\rho_{\rm sc}(\vec Z,\vec X_{ji}) - \widetilde\mu_{ji}^p(\vec Z)\right| = \bigO_P\big(\sqrt{\log n/k^\ast}\big)$, as intended.
\end{proof}

\begin{remark}
We have proved the above result assuming that the pairwise distances are bounded by $C\sqrt{d_n}$ for some $C>0$. More generally, the result can be obtained with $C_n$ in place of $C$, satisfying $\sup_{j=1,\ldots,J}\sup_{i=1,\ldots,n_j} \|\vec X_{ji}-\vec Z\| = \bigO_P(C_n\sqrt{d_n})$. Under some regularity conditions, e.g., uniformly bounded fourth moments, we have $C_n = \bigO(1)$. However, under sparsity in the underlying populations, we may also have $C_n = \smallO(1)$, leading to finer bounds than in Lemma~\ref{lemma:1}.
\end{remark}

Next, we study $\sup_{1\le j\le J}\ \sup_{1\le i\le n_j} \bigl|\widetilde{\mu}_{ji}^p(\vec Z)-\widetilde{\rho}(\vec Z,\vec X_{ji})\bigr|$. Instead of giving a rigorous derivation, we provide some heuristic justification for this term. In the high-dimensional literature, it has been established that the pairwise scaled distances $d_n^{-1/2}\|\vec X_{ji} - \vec X_{j'i'}\|$ converge in probability to constants $a_{jj'}$ depending only on the locations and scales of the underlying classes \citep{hall2005geometric,sarkar2019perfect,roy2022generalizations}. Consequently, for large $d_n$, the scaled terms $\widetilde A_{j'i'}(\vec Z,ji)$ are close to $a_{jj'}^{\vec Z}$, where the constant $a_{jj'}^{\vec Z}$ depends on $j,j'$, as well as the location and scale of $\vec Z$. Therefore, for large $n$, we get the following approximation
\begin{align*}
    \bigl|{\widetilde\mu}_{ji}^p(\vec Z)-\widetilde{\rho}(\vec Z,\vec X_{ji})\bigr| 
    &= \left|\dfrac{1}{k^\ast}\sum_{j'=1}^J \sum_{i'=1}^{n_{j'}} p_{j' i'}\widetilde{A}_{j' i'}(\vec Z, ji) - \dfrac{1}{n-1}\sum_{j'=1}^J \sum_{i'=1}^{n_{j'}} \widetilde{A}_{j' i'}(\vec Z, ji)\right| \\
    &= \left|\sum_{j'=1}^J \sum_{i'=1}^{n_{j'}} \widetilde{A}_{j' i'}(\vec Z, ji)\left(\dfrac{p_{j' i'}}{k^\ast}- \dfrac{1}{n-1}\right)\right|\approx \left|\sum_{j'=1}^J a_{j j'}^{\vec Z} \sum_{i'=1}^{n_{j'}} \left(\dfrac{p_{j' i'}}{k^\ast}- \dfrac{1}{n-1}\right)\right|\\
    &= \left|\sum_{j'=1}^J a_{jj'}^{\vec Z} \left(\dfrac{k_{j'}}{k^\ast}- \dfrac{n_{j'}}{n-1}\right)\right| \leq \sqrt{\sum_{j'=1}^J (a_{jj'}^{\vec Z})^2} \sqrt{\sum_{j'=1}^J\left(\dfrac{k_{j'}}{k^\ast}- \dfrac{n_{j'}}{n-1}\right)^2}.
\end{align*}
Now, note that, the second term is free from $i,j$. Therefore,
\[
\sup_{1\le j \le J} \sup_{1 \le i \le n_j} \bigl|{\widetilde\mu}_{ji}^p(\vec Z)-\widetilde{\rho}(\vec Z,\vec X_{ji})\bigr| \leq \sqrt{\sum_{j=1}^J\left(\dfrac{k_{j}}{k^\ast}- \dfrac{n_{j}}{n-1}\right)^2} \sup_{1 \le j \le J} \sqrt{\sum_{j'=1}^J (a_{jj'}^{\vec Z})^2}.
\]
Since the constants $a_{jj'}^{\vec Z}$ are finite and $J$ is fixed, it follows that $\sup_{1 \le j \le J} \sqrt{\sum_{j'=1}^J (a_{jj'}^{\vec Z})^2}$ is a finite quantity. Also, by our assumption, $n_j/n \to \alpha_j \in (0,1)$. Therefore, if $k_j/k^\ast \to \alpha_j$ for $j=1,\ldots,J$, then $\sup_{1\le j \le J} \sup_{1 \le i \le n_j} \bigl|{\widetilde\mu}_{ji}^p(\vec Z)-\widetilde{\rho}(\vec Z,\vec X_{ji})\bigr| \overset{P}{\to} 0$ as $n \to \infty$. In our practical implementations, we use $k_j \propto \sqrt{d} (n_j/n)\log n_j$. This choice, along with $n_j/n \to \alpha_j \in (0,1)$ ensures that $k_j/k^\ast \to \alpha_j$ as $n \to \infty$.

\begin{table}[b!]
\centering
\caption{Average misclassification rates (in $\%$) of NN-MADD and NN-MADD$_{\rm sc}$ in Examples~\ref{example1}--\ref{example7} with varying training sample sizes. The standard errors of the misclassification rates are reported in a smaller font within parentheses.\label{tab:appendix_simulation_MADD_sc}}
\small
\begin{tabular}{c|rr|rr|rr}
\toprule

& \multicolumn{2}{c|}{$n=1000$} 
& \multicolumn{2}{c|}{$n=2000$} 
& \multicolumn{2}{c}{$n=4000$} \\
\cline{2-7}
& NN-MADD & NN-MADD$_{\rm sc}$  
& NN-MADD & NN-MADD$_{\rm sc}$  
& NN-MADD & NN-MADD$_{\rm sc}$  \\
\midrule

Example~\ref{example1}
& 20.23 \se{0.17} & 21.27 \se{0.20} 
& 19.28 \se{0.13} & 20.88 \se{0.17} 
& 18.44 \se{0.12} & 19.55 \se{0.11} \\

Example~\ref{example2}
& 14.18 \se{0.17} & 14.39 \se{0.21} 
& 14.82 \se{0.10} & 14.82 \se{0.14} 
& 14.14 \se{0.13} & 14.28 \se{0.10} \\

Example~\ref{example3}
& 5.34 \se{0.08} & 5.68 \se{0.09} 
& 5.06 \se{0.07} & 5.39 \se{0.07} 
& 4.60 \se{0.06} & 4.96 \se{0.07} \\

Example~\ref{example4}
& 11.88 \se{0.14} & 12.00 \se{0.19} 
& 11.78 \se{0.13} & 11.84 \se{0.17} 
& 11.57 \se{0.10} & 11.91 \se{0.14} \\

Example~\ref{example5}
& 11.21 \se{0.13} & 11.39 \se{0.15} 
& 11.05 \se{0.11} & 11.43 \se{0.15} 
& 10.83 \se{0.11} & 11.19 \se{0.11} \\

Example~\ref{example6}
& 27.33 \se{0.21} & 28.05 \se{0.19} 
& 25.88 \se{0.14} & 26.74 \se{0.18} 
& 25.21 \se{0.12} & 26.23 \se{0.16} \\

Example~\ref{example7}
& 12.12 \se{0.14} & 12.39 \se{0.15} 
& 11.91 \se{0.09} & 12.36 \se{0.17} 
& 11.50 \se{0.10} & 11.70 \se{0.11} \\

\bottomrule
\end{tabular} 
\end{table}

\section{Additional Simulation Studies}
\label{appendix: Simulation}

In this section, we report some additional simulation results which are omitted from the main text. We start with a comparative study of the performances of the NN-MADD and NN-MADD$_\mathrm{sc}$ classifiers in Examples~\ref{example1}--\ref{example7} for varying training sample sizes.
The misclassification rates are reported in Table~\ref{tab:appendix_simulation_MADD_sc} along with the corresponding standard errors. We observe that, in all the examples, the differences between the misclassification rates of NN-MADD and NN-MADD$_{\rm sc}$ are very small. Except for a few instances, the difference never exceeds $1\%$. 

Next, in Table~\ref {tab:Compare_d}, we compare the performance of the proposed classifier with some state-of-the-art classifiers in Examples~\ref{example1}--\ref{example7} for fixed training sample size $n=1000$ and varying dimensions $d$. A similar comparison with fixed dimension $d=100$ and varying training sample sizes $n$ is made in the main text in Table ~\ref{tab:Compare}. The results are similar to what we observe in Table~\ref{tab:Compare}. NN-MADD$_{\rm sc}$ achieves performance close to that of NN-MADD, and often outperforms the existing classifiers. The superiority is especially noticeable in large dimensions.

\begin{table}[t!]
\centering
\caption{Average misclassification rates (in \%) of different classifiers in Examples~\ref{example1}--\ref{example7} with $n=1000$ and varying $d$. The corresponding standard errors are reported in a smaller font within parentheses.\label{tab:Compare_d}}

\small
\setlength{\tabcolsep}{3pt}
\begin{adjustbox}{max width=\textwidth}
\begin{tabular}{l|r|rrrrrrrr}
\toprule
Example & \multicolumn{1}{r|}{$d$}
& \multicolumn{1}{l}{NN} & \multicolumn{1}{l}{GLMNET} & \multicolumn{1}{l}{CART} & \multicolumn{1}{l}{RF} & \multicolumn{1}{l}{LSVM} & \multicolumn{1}{l}{NLSVM} & 
\multicolumn{1}{l}{NN-MADD}&
\multicolumn{1}{l}{NN-MADD$_{\rm sc}$} \\
\midrule

\multirow{3}{*}{Example~\ref{example1}} 
& 50 & 51.56 \se{0.07} & 49.99 \se{0.14} & 34.83 \se{0.20} & 19.68 \se{0.12} & 48.49 \se{0.16} & 18.43 \se{0.10} & 18.61 \se{0.11} & 18.79 \se{0.11} \\
& 150 & 50.58 \se{0.02} & 49.86 \se{0.13} & 35.66 \se{0.23} & 18.10 \se{0.11} & 47.74 \se{0.19} & 15.86 \se{0.08} & 12.11 \se{0.12} & 12.46 \se{0.14}\\
& 250 & 50.34 \se{0.01} & 49.35 \se{0.12} & 35.64 \se{0.29} & 18.05 \se{0.12} & 47.85 \se{0.15} & 15.51 \se{0.09} & 9.95 \se{0.09} & 9.99 \se{0.15}\\
\midrule

\multirow{3}{*}{Example~\ref{example2}} 
& 50 & 39.87 \se{0.18} & 31.13 \se{0.17} & 44.15 \se{0.13} & 31.59 \se{0.15} & 31.76 \se{0.17} & 29.61 \se{0.19} & 30.92 \se{0.15} & 31.65 \se{0.22} \\ 
& 150 & 34.11 \se{0.17} & 25.67 \se{0.14} & 43.78 \se{0.15} & 25.64 \se{0.12} & 28.15 \se{0.17} & 21.55 \se{0.34} & 12.33 \se{0.12} & 13.30 \se{0.13} \\
& 250 & 30.79 \se{0.16} & 24.47 \se{0.13} & 43.58 \se{0.14} & 23.51 \se{0.12} & 28.20 \se{0.18} & 19.85 \se{0.60} & 5.25 \se{0.08} & 5.87 \se{0.10}\\
\midrule

\multirow{3}{*}{Example~\ref{example3}} 
& 50 & 13.72 \se{0.14} & 4.06 \se{0.06} & 31.04 \se{0.14} & 5.35 \se{0.07} & 4.41 \se{0.06} & 4.22 \se{0.06} & 7.23 \se{0.09} & 7.44 \se{0.10} \\ 

& 150 & 27.50 \se{0.19} & 6.64 \se{0.07} & 40.24 \se{0.14} & 10.26 \se{0.13} & 7.98 \se{0.07} & 6.99 \se{0.08} & 15.60 \se{0.13} & 16.51 \se{0.16} \\ 

& 250 & 22.18 \se{0.19} & 2.64 \se{0.04} & 40.24 \se{0.16} & 5.68 \se{0.09} & 4.38 \se{0.08} & 2.87 \se{0.06} & 7.11 \se{0.08} & 7.74 \se{0.12} \\ 
\midrule

\multirow{3}{*}{Example~\ref{example4}} 
& 50 & 47.37 \se{0.08} & 49.91 \se{0.14} & 39.72 \se{0.21} & 23.04 \se{0.19} & 49.40 \se{0.16} & 16.73 \se{0.12} & 22.90 \se{0.15} & 23.32 \se{0.23} \\
& 150 & 49.76 \se{0.02} & 49.80 \se{0.16} & 39.65 \se{0.21} & 12.94 \se{0.22} & 49.37 \se{0.13} & 4.59 \se{0.06} & 6.16 \se{0.08} & 6.44 \se{0.09} \\
& 250 & 49.98 \se{0.01} & 49.91 \se{0.16} & 39.59 \se{0.24} & 9.26 \se{0.25} & 49.52 \se{0.13} & 1.42 \se{0.04} & 1.84 \se{0.06} & 1.94 \se{0.07} \\ 

\midrule

\multirow{3}{*}{Example~\ref{example5}} 
& 50 & 46.74 \se{0.09} & 38.93 \se{0.15} & 38.65 \se{0.20} & 21.80 \se{0.16} & 38.85 \se{0.12} & 15.60 \se{0.12} & 22.07 \se{0.15} & 22.48 \se{0.17} \\ 
& 150 & 49.69 \se{0.02} & 31.28 \se{0.19} & 37.82 \se{0.24} & 11.64 \se{0.19} & 31.75 \se{0.16} & 3.73 \se{0.04} & 5.78 \se{0.08} & 6.04 \se{0.11} \\
& 250 & 49.97 \se{0.01} & 26.01 \se{0.14} & 37.79 \se{0.19} & 6.99 \se{0.23} & 27.77 \se{0.16} & 1.10 \se{0.03} & 1.69 \se{0.05} & 1.81 \se{0.07} \\ 

\midrule

\multirow{3}{*}{Example~\ref{example6}} 
& 50 & 33.04 \se{0.16} & 50.12 \se{0.12} & 43.38 \se{1.99} & 29.93 \se{1.40} & 50.27 \se{0.15} & 30.33 \se{0.20} & 22.89 \se{0.10} & 23.30 \se{0.13} \\
& 150 & 41.30 \se{0.13} & 50.16 \se{0.16} & 48.95 \se{0.59} & 45.26 \se{0.43} & 50.07 \se{0.14} & 44.65 \se{0.21} & 28.27 \se{0.15} & 29.62 \se{0.18} \\ 
& 250 & 44.30 \se{0.15} & 50.14 \se{0.11} & 49.18 \se{0.64} & 47.83 \se{0.30} & 49.87 \se{0.17} & 48.83 \se{0.19} & 32.09 \se{0.16} & 33.76 \se{0.25} \\
\midrule

\multirow{3}{*}{Example~\ref{example7}} 

& 50 & 30.54 \se{0.23} & 49.97 \se{0.16} & 41.09 \se{0.19} & 24.41 \se{0.15} & 50.12 \se{0.15} & 23.12 \se{0.17} & 22.50 \se{0.13} & 22.69 \se{0.17} \\

& 150 & 20.51 \se{0.15} & 50.13 \se{0.15} & 40.46 \se{0.12} & 21.72 \se{0.13} & 50.14 \se{0.14} & 14.47 \se{0.16} & 6.11 \se{0.09} & 6.45 \se{0.13} \\

& 250 & 15.37 \se{0.14} & 50.17 \se{0.15} & 40.24 \se{0.20} & 21.16 \se{0.11} & 49.82 \se{0.14} & 11.77 \se{0.12} & 1.66 \se{0.04} & 1.96 \se{0.06} \\  

\bottomrule
\end{tabular}
\end{adjustbox}
\end{table}

Finally, we numerically probe the validity of the assumption made in Proposition~\ref{theorem} in our simulation examples. We already numerically studied the behavior of $\widetilde{\beta}_n$ in Examples~\ref{example2} and \ref{example5} in the main text (see Figure~\ref{fig:assumption}) by showing boxplots of $\widetilde\beta(\vec Z)$ values with varying $n$ and $d$. Here, we repeat the same exercise for the remaining examples. The boxplots are shown in Figure~\ref{fig:assumption_appendix}. In all the examples, the values of $\widetilde\beta(\vec Z)$ concentrate, and decrease as a whole, as $n$ and $d$ increase at an appropriate rate. These results give empirical justification to our assumption in Proposition~\ref{theorem}.

\begin{figure}[t]
    \small
    \centering
    \begin{tabular}{c c}
    (a) Example~\ref{example1} (Gaussian vs multivariate $t$) & (b) Example~\ref{example3} (location problem)\\
    \includegraphics[width=0.48\textwidth]{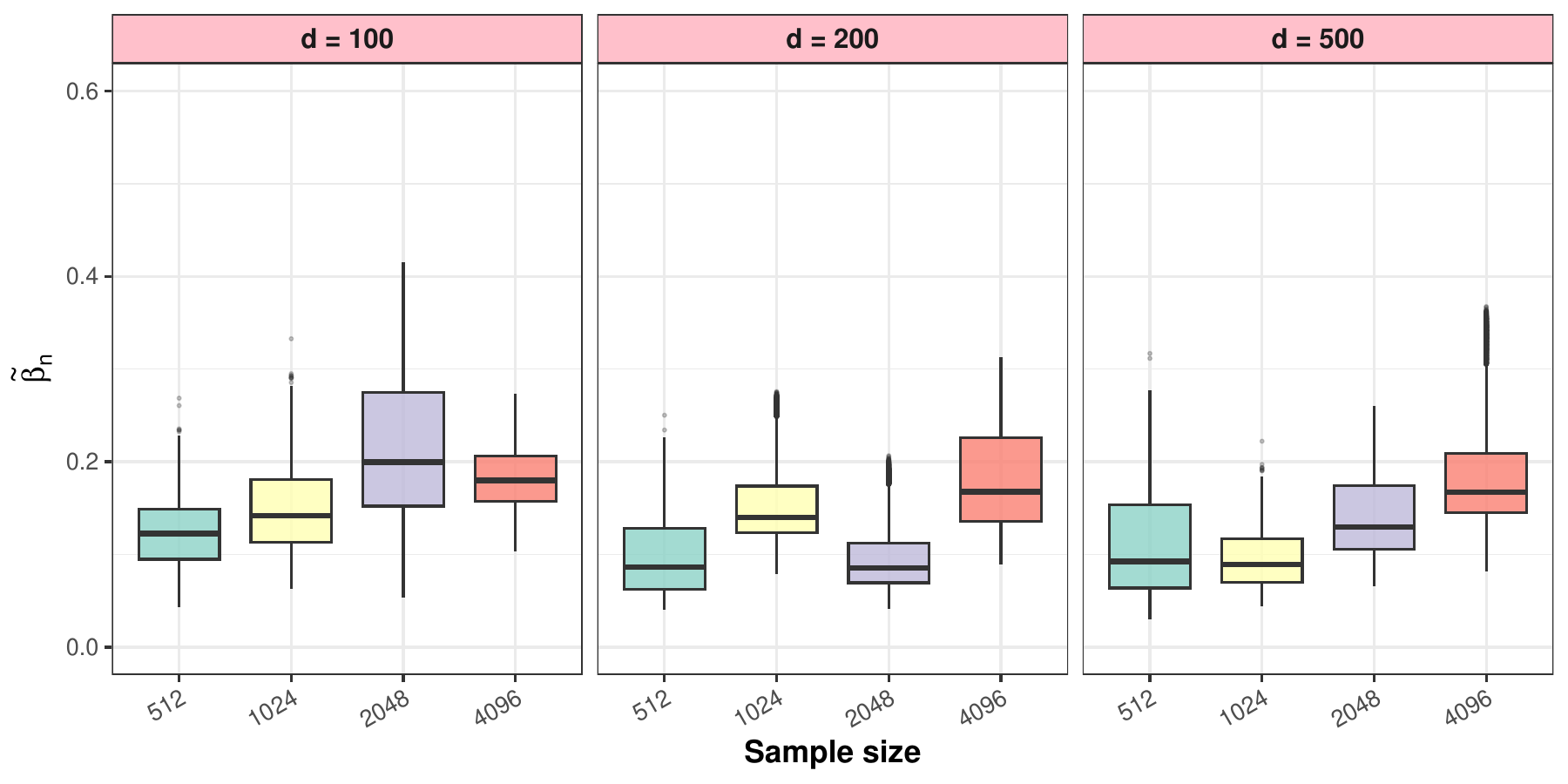} &
    \includegraphics[width=0.48\textwidth]{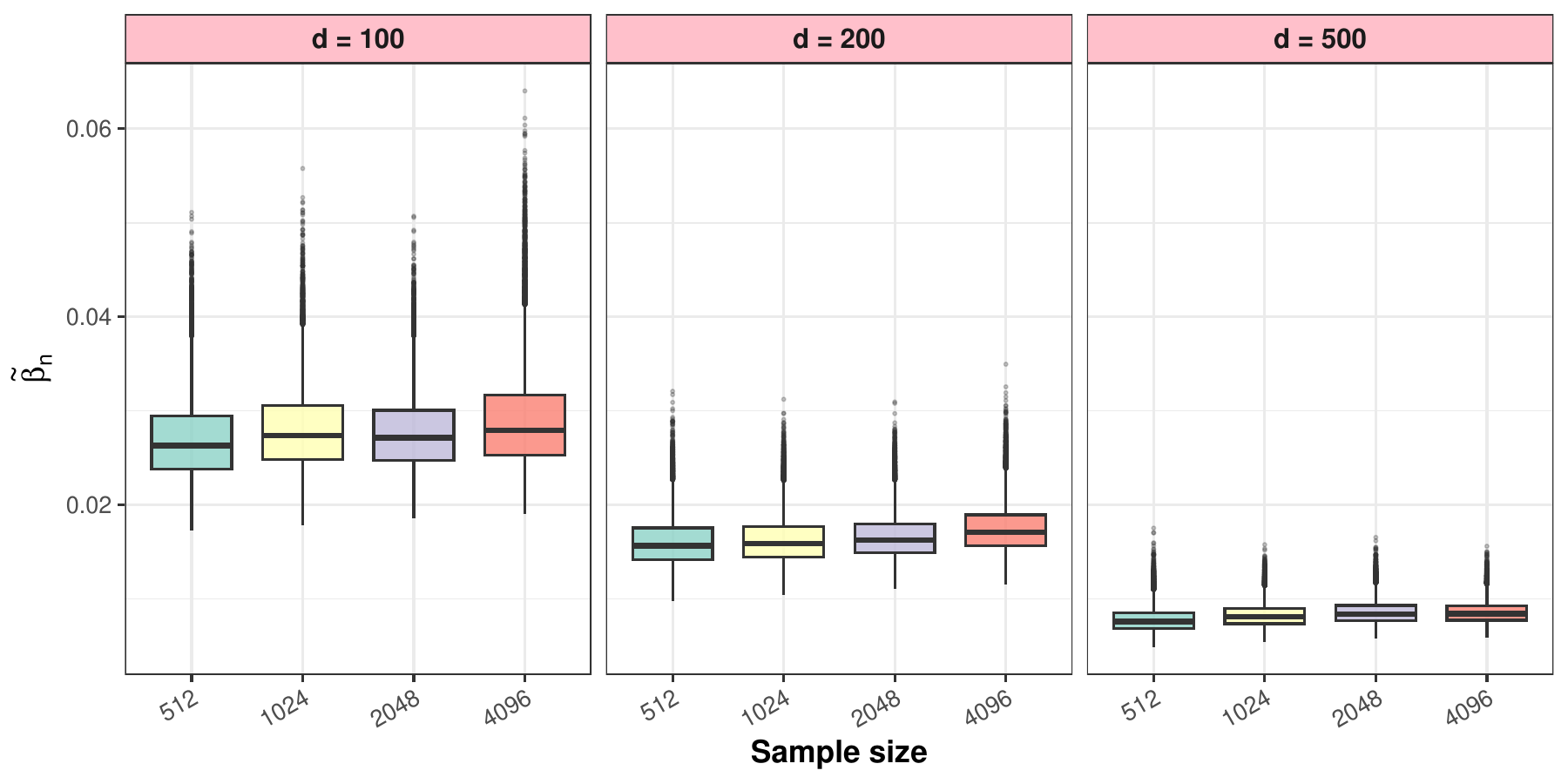} \\ [5pt]
    (c) Example~\ref{example4} (scale problem) & (d) Example~\ref{example6} (mixture of Gaussians)\\
    \includegraphics[width=0.48\textwidth]{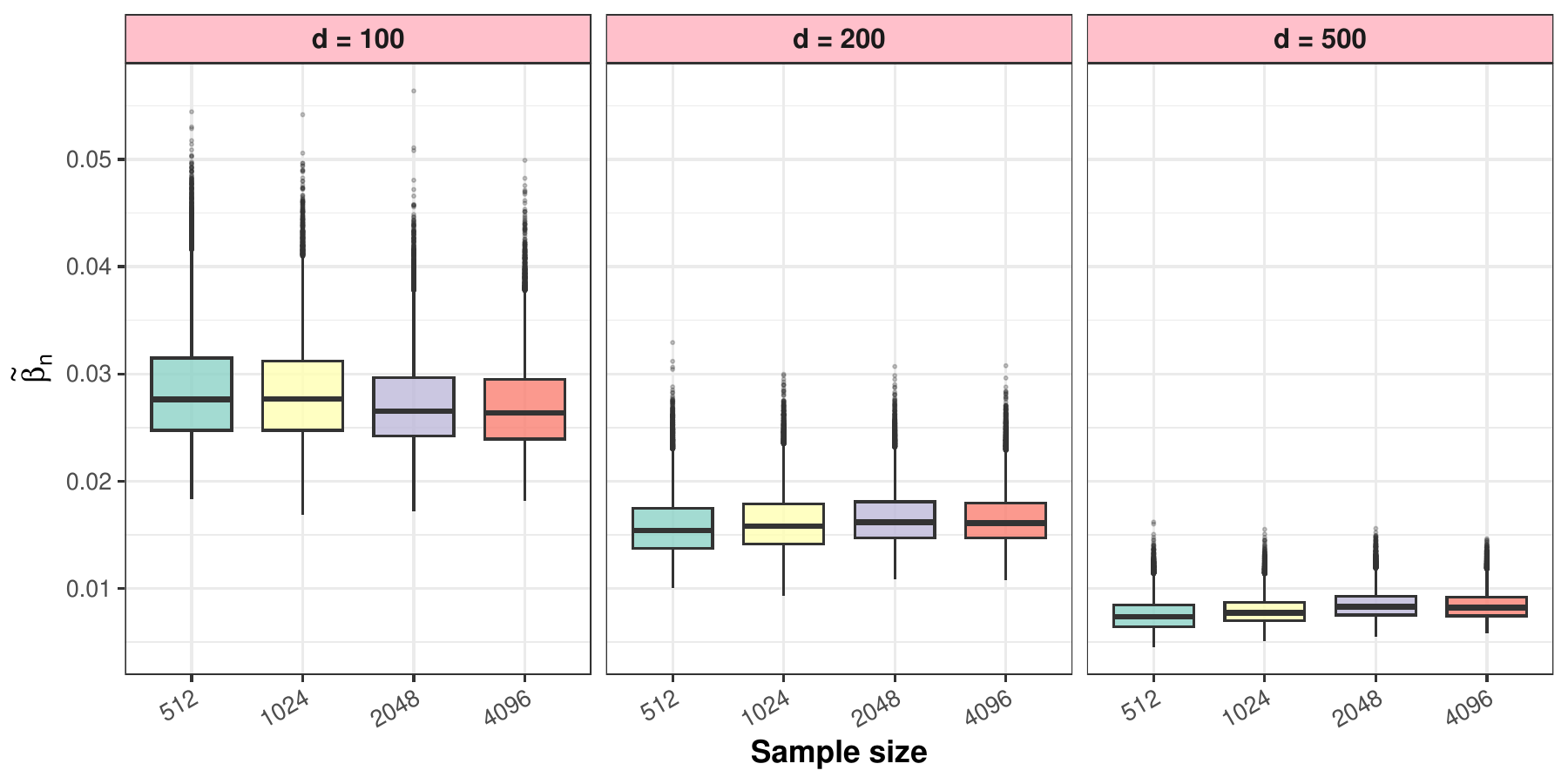} &
    \includegraphics[width=0.48\textwidth]{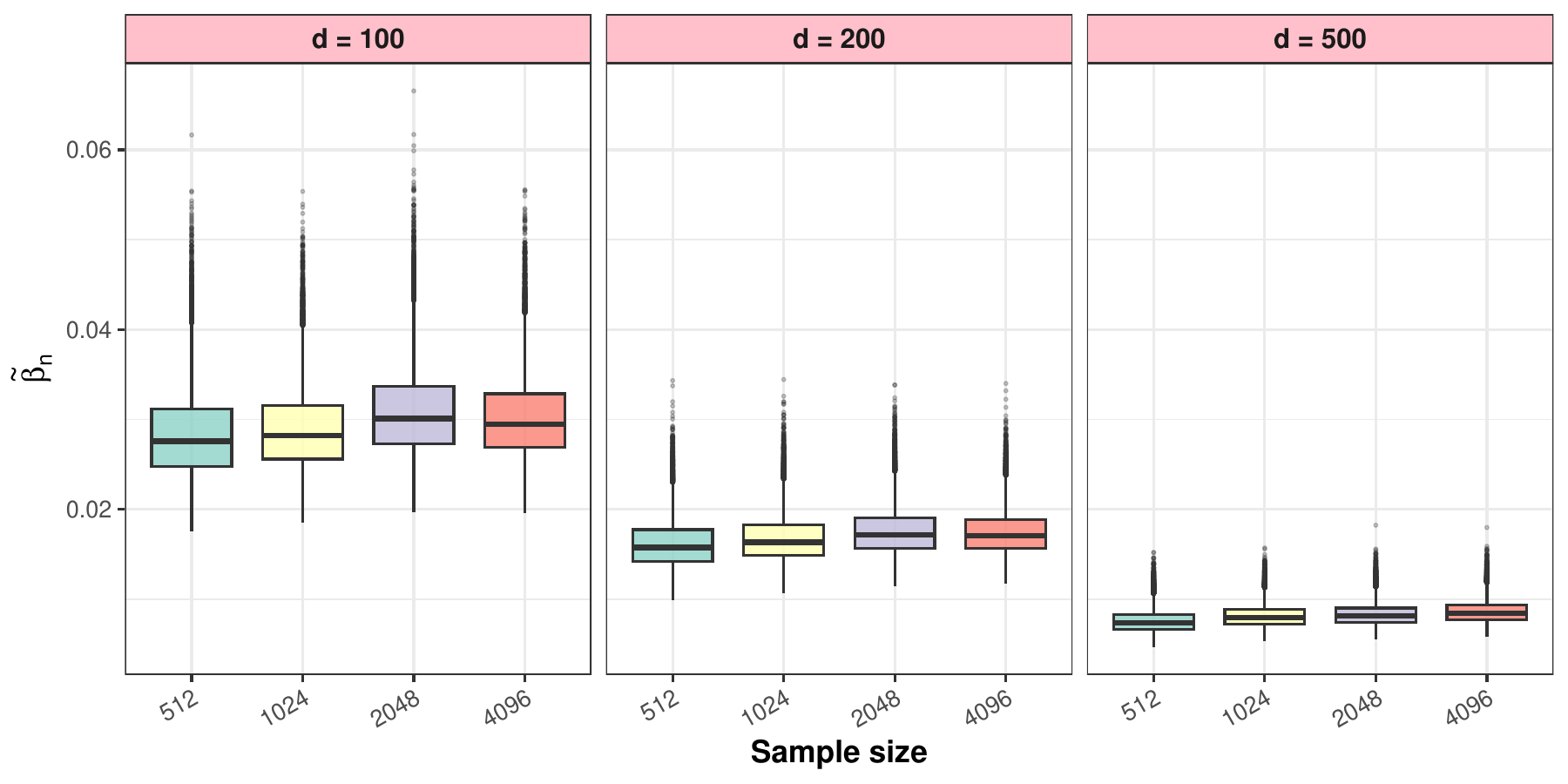} \\ [5pt]
    \multicolumn{2}{c}{(e) Example~\ref{example7} (mixture of Gaussians)} \\
    \multicolumn{2}{c}{\includegraphics[width=0.48\textwidth]{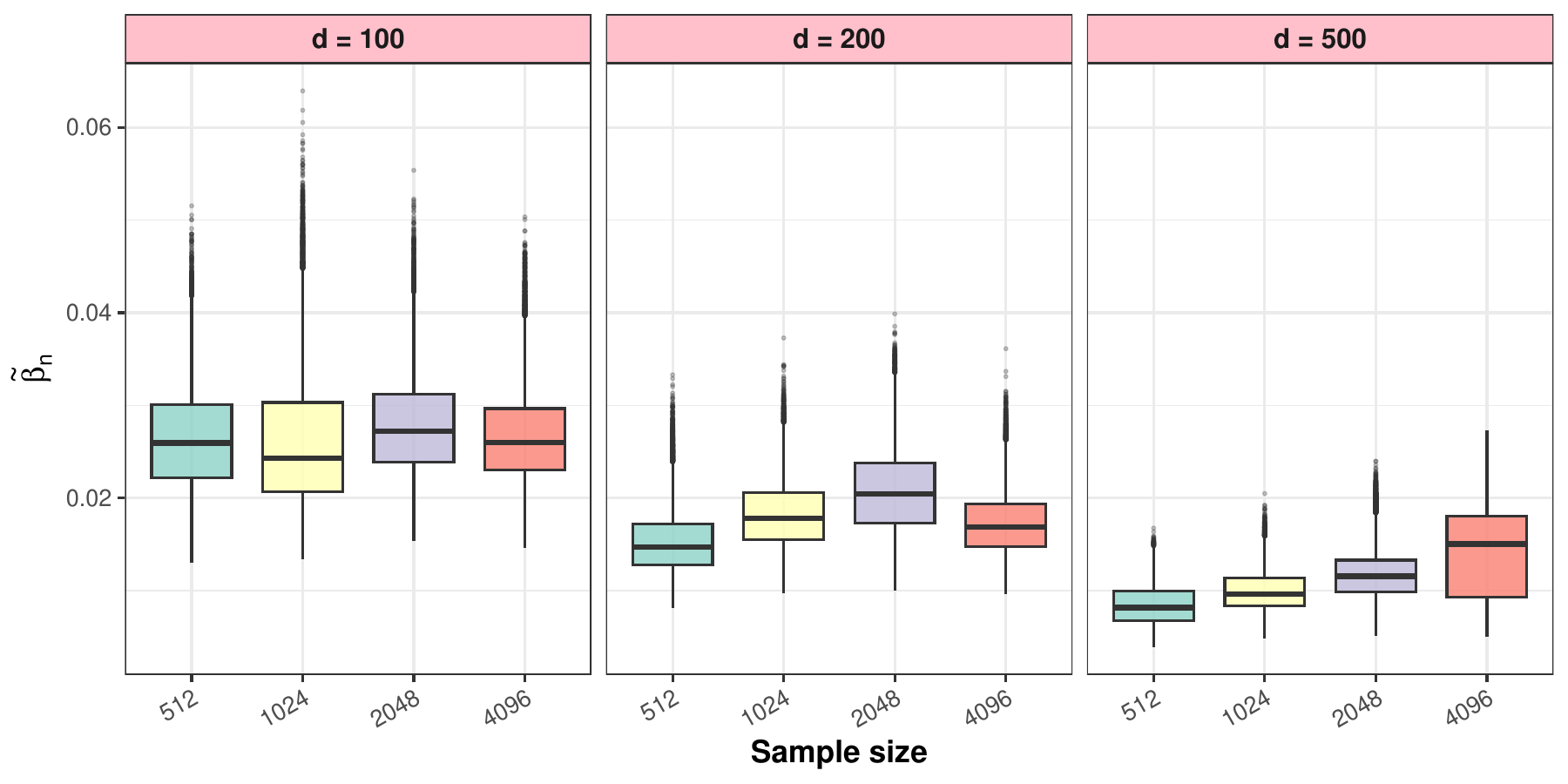}}
    \end{tabular}
    \caption{Boxplots of $\widetilde\beta(\vec Z)$ values in Examples~\ref{example1}, \ref{example3}, \ref{example4}, \ref{example6} and \ref{example7} for varying $n$ and $d$.\label{fig:assumption_appendix}}
\end{figure}

\section{Additional Real Data Analyses}
\label{appendix:Real}

Here, we report some additional results from real data analyses. In the main text, we reported the misclassification rates of different classifiers in six benchmark datasets from the UCR Time Series Classification Archive. In Table~\ref{tab:real_data_results_all}, we report the misclassification rates for the remaining benchmark datasets from the UCR Time Series Classification Archive that we analyzed (datasets with at least $500$ samples, at least $50$ dimensions, and at most $10$ classes). The results for both NN-MADD$_{\rm sc}$ and NN-gMADD$_{\rm sc}$ are reported. In some of the datasets, viz., \textbf{Computers}, \textbf{SonyAIBORobotSurface1}, \textbf{LargekitchenAppliances}, \textbf{CBF}, and \textbf{SemgHandMovementCh2}, one of the proposed classifiers achieves the lowest misclassification rate. In the other datasets, their performance is close to that of the best-performing classifier. Overall, the results further support the effectiveness of the proposed scalable MADD-based classifiers for large data classification.

\begin{table}[h!]
    \centering
    \renewcommand{\se}[1]{\tiny{(\textit{#1})}}
    \setlength{\tabcolsep}{2pt}
    \caption{Average misclassification rates (in \%) of different classifiers for benchmark datasets from UCR time series classification archive. The corresponding standard errors are reported in a smaller font within parentheses. The underlined numbers represent the best result in each setup. \label{tab:real_data_results_all} }
        
    \scriptsize
    \renewcommand{\arraystretch}{0.92}
    \resizebox{\textwidth}{!}{
    \begin{tabular}{l|rrrr|rrrrrrrr}
        \toprule
        Dataset & $J$ & $d$ & $n$ & $n_{\rm test}$ &
        NN & GLMNET & CART & RF & LSVM & NLSVM &
        NN-MADD$_{\rm sc}$ & NN-gMADD$_{\rm sc}$ \\
        \midrule

        Computers & 2 & 720 & 350 & 150 & 41.52 & 42.67 & 39.31 & 36.85 & 44.67 & 39.55 & 40.69 & \best{36.40}  \\
         &  &  &  &  & \se{0.48} & \se{0.80} & \se{0.90} & \se{0.78} & \se{0.84} & \se{0.54}  & \se{0.75} & \se{0.58} \\
        \addlinespace[2pt]

         SonyAIBORobotSurface1 & 2 & 70 & 433 & 188 & 1.36 & 2.23 & 6.98 & 1.40 & 2.36 & 0.57 & \best{0.51} & 0.66 \\
         &  &  &  &  & \se{0.15} & \se{0.21} & \se{0.34} & \se{0.13} & \se{0.19} & \se{0.09} & \se{0.09} & \se{0.10} \\
        \addlinespace[2pt]

        InlineSkate & 7 & 1882 & 456 & 194 & 49.81 & 61.79 & 65.46 & 48.04 & 62.76 & \best{48.00} & 56.37  & 55.46\\
         &  &  &  &  & \se{0.61} & \se{0.67} & \se{0.74} & \se{0.63} & \se{0.71} & \se{0.67} & \se{0.56}  & \se{0.47}\\
        \addlinespace[2pt]

        LargeKitchenAppliances & 3 & 720 & 519 & 231 & 43.84 & 49.75 & 43.39 & 36.48 & 55.88 & 44.64  & 38.82  & \best{36.12}\\
         &  &  &  &  & \se{0.55} & \se{0.73} & \se{0.49} & \se{0.61} & \se{0.53} & \se{0.55} & \se{0.61} & \se{0.50} \\
        \addlinespace[2pt]

        RefrigerationDevices & 3 & 720 & 519 & 231 & 56.38 & 63.81 & 49.40 & \best{41.63} & 63.08 & 60.10 & 59.57  & 48.29\\
         &  &  &  &  & \se{0.56} & \se{0.66} & \se{0.60} & \se{0.61} & \se{0.56} & \se{0.60} & \se{0.54}  & \se{0.56}\\
        \addlinespace[2pt]

        ScreenType & 3 & 720 & 519 & 231 & 55.55 & 55.36 & 56.40 & 52.95 & 59.71 & \best{52.54} & 54.74 & 53.87 \\
         &  &  &  &  & \se{0.61} & \se{0.57} & \se{0.61} & \se{0.59} & \se{0.58} & \se{0.41}  & \se{0.60} & \se{0.62}\\
        \addlinespace[2pt]

        SmallKitchenAppliances & 3 & 720 & 519 & 231 & 61.45 & 43.55 & 38.23 & \best{32.31} & 49.66 & 44.78  & 53.70 & 43.26\\
         &  &  &  &  & \se{0.48} & \se{0.67} & \se{0.56} & \se{0.64} & \se{0.57} & \se{0.73}  & \se{0.61} & \se{0.87} \\
        \addlinespace[2pt]

        ECGFiveDays & 2 & 136 & 618 & 266 & 0.65 & 0.15 & 7.13 & 1.97 & \best{0.05} & 0.33  & 1.52 & 3.11\\
         &  &  &  &  & \se{0.11} & \se{0.04} & \se{0.39} & \se{0.18} & \se{0.02} & \se{0.05}  & \se{0.16} & \se{0.20}\\

        \bottomrule
    \end{tabular}
    }
\end{table}

\begin{table}[h!]
    \centering
    \renewcommand{\se}[1]{\tiny{(\textit{#1})}}
    \setlength{\tabcolsep}{2pt}
    \caption*{Table~\thetable~(continued)}
        
    \scriptsize
    \renewcommand{\arraystretch}{0.92}
    \resizebox{\textwidth}{!}{
    \begin{tabular}{l|rrrr|rrrrrrrr}
        \toprule
        Dataset & $J$ & $d$ & $n$ & $n_{\rm test}$ &
        NN & GLMNET & CART & RF & LSVM & NLSVM &
        NN-MADD$_{\rm sc}$ & NN-gMADD$_{\rm sc}$ \\
        \midrule
        
        SemgHandGenderCh2 & 2 & 1500 & 630 & 270 & 8.46 & 7.59 & 16.16 & \best{6.52} & 14.95 & 7.08  & 7.66 & 24.57\\
         &  &  &  &  & \se{0.41} & \se{0.35} & \se{0.41} & \se{0.35} & \se{0.35} & \se{0.34}  & \se{0.30} & \se{0.54}\\
        \addlinespace[2pt]

        SemgHandSubjectCh2 & 5 & 1500 & 630 & 270 & 16.71 & \best{14.16} & 39.59 & 20.12 & 14.79 & 14.25  & 20.13 & 39.32\\
         &  &  &  &  & \se{0.34} & \se{0.41} & \se{0.60} & \se{0.46} & \se{0.52} & \se{0.51}  & \se{0.47}& \se{0.45} \\
        \addlinespace[2pt]

        CBF & 3 & 128 & 642 & 288 & 1.33 & 1.93 & 11.61 & 0.28 & 2.26 & 0.15 & 0.08  & \best{0.00}\\
         &  &  &  &  & \se{0.13} & \se{0.18} & \se{0.54} & \se{0.06} & \se{0.13} & \se{0.05} & \se{0.04}  & \se{0.00}\\
        \addlinespace[2pt]

         SemgHandMovementCh2 & 6 & 1500 & 642 & 258 & 30.12 & 45.22 & 56.76 & 41.57 & 43.38 & 33.32  & \best{24.25} & 46.65\\
         &  &  &  &  & \se{0.57} & \se{0.62} & \se{0.74} & \se{0.63} & \se{0.51} & \se{0.42} & \se{0.41} & \se{0.53} \\
        \addlinespace[2pt]

        SonyAIBORobotSurface2 & 2 & 65 & 685 & 295 & 1.64 & 5.10 & 9.88 & 3.67 & 5.55 & \best{1.11}  & 2.66 & 2.49\\
         &  &  &  &  & \se{0.15} & \se{0.20} & \se{0.39} & \se{0.24} & \se{0.26} & \se{0.11}  & \se{0.17} & \se{0.19}\\
        \addlinespace[2pt]

        Symbols & 6 & 398 & 712 & 308 & 3.43 & 5.61 & 8.87 & \best{2.87} & 5.18 & 3.69  & 4.71 & 4.17\\
         &  &  &  &  & \se{0.17} & \se{0.22} & \se{0.39} & \se{0.17} & \se{0.17} & \se{0.19} & \se{0.17} & \se{0.18} \\
        \addlinespace[2pt]

        ItalyPowerDemand & 2 & 24 & 766 & 330 & 3.07 & \best{2.69} & 3.18 & 2.70 & 2.85 & 2.72 & 3.54 & 3.88 \\
         &  &  &  &  & \se{0.18} & \se{0.16} & \se{0.15} & \se{0.13} & \se{0.17} & \se{0.13} & \se{0.19} & \se{0.18} \\
        \addlinespace[2pt]

        MedicalImages & 10 & 99 & 789 & 352 & 24.50 & 33.47 & 34.56 & \best{21.20} & 30.58 & 22.99  & 25.44 & 25.84\\
         &  &  &  &  & \se{0.34} & \se{0.34} & \se{0.53} & \se{0.37} & \se{0.29} & \se{0.36} & \se{0.35} & \se{0.39} \\
        \addlinespace[2pt]

        CinCECGTorso & 4 & 1639 & 992 & 428 & 0.21 & 32.32 & 7.20 & 0.78 & 23.42 & \best{0.17}  & 0.49 & 0.96\\
         &  &  &  &  & \se{0.04} & \se{0.44} & \se{0.29} & \se{0.10} & \se{0.45} & \se{0.04}  & \se{0.04} & \se{0.10}\\
        \addlinespace[2pt]

        Mallat & 8 & 1024 & 1608 & 792 & 1.77 & 1.96 & 2.46 & \best{0.54} & 2.01 & 1.17 & 1.48  & 1.63\\
         &  &  &  &  & \se{0.08} & \se{0.10} & \se{0.10} & \se{0.07} & \se{0.08} & \se{0.08}  & \se{0.07} & \se{0.08}\\
        \addlinespace[2pt]

        MixedShapesSmallTrain & 5 & 1024 & 1765 & 760 & 6.49 & 10.38 & 16.82 & 7.41 & 12.38 & \best{5.19} & 7.75 & 7.64 \\
         &  &  &  &  & \se{0.15} & \se{0.17} & \se{0.23} & \se{0.22} & \se{0.21} & \se{0.15}  & \se{0.19} & \se{0.19}\\
        \addlinespace[2pt]

        FreezerSmallTrain & 2 & 301 & 2014 & 864 & 4.32 & 0.23 & 3.53 & \best{0.18} & 0.21 & 0.42 & 2.19 & 2.62  \\
         &  &  &  &  & \se{0.17} & \se{0.04} & \se{0.19} & \se{0.04} & \se{0.03} & \se{0.06}  & \se{0.11}& \se{0.18} \\
        \addlinespace[2pt]

        MixedShapesRegularTrain & 5 & 1024 & 2044 & 881 & 6.55 & 10.32 & 16.21 & 7.06 & 12.00 & \best{4.89} & 7.92 & 7.51\\
         &  &  &  &  & \se{0.13} & \se{0.20} & \se{0.25} & \se{0.21} & \se{0.22} & \se{0.14} & \se{0.15} & \se{0.15} \\
        \addlinespace[2pt]

        FreezerRegularTrain & 2 & 301 & 2100 & 900 & 3.94 & 0.26 & 3.16 & \best{0.17} & 0.22 & 0.39 & 1.80 & 2.40 \\
         &  &  &  &  & \se{0.14} & \se{0.04} & \se{0.16} & \se{0.03} & \se{0.03} & \se{0.04} & \se{0.09} & \se{0.16} \\
        \addlinespace[2pt]

        UWaveGestureLibraryAll & 8 & 945 & 3194 & 1284 & 3.52 & 10.08 & 18.02 & 3.53 & 8.36 & \best{2.11} & 3.74  & 2.82\\
         &  &  &  &  & \se{0.09} & \se{0.09} & \se{0.22} & \se{0.10} & \se{0.10} & \se{0.06}  & \se{0.08} & \se{0.06}\\
        \addlinespace[2pt]

        UWaveGestureLibraryX & 8 & 315 & 3194 & 1284 & 22.68 & 37.47 & 32.00 & 18.84 & 32.12 & \best{17.88}  & 23.26 & 21.53\\
         &  &  &  &  & \se{0.16} & \se{3.22} & \se{0.25} & \se{0.21} & \se{0.14} & \se{0.14}  & \se{0.18} & \se{0.25}\\
        \addlinespace[2pt]

        UWaveGestureLibraryY & 8 & 315 & 3194 & 1284 & 28.86 & 39.43 & 36.60 & 24.52 & 35.29 & \best{24.36}& 33.40  & 30.80 \\
         &  &  &  &  & \se{0.13} & \se{0.28} & \se{0.24} & \se{0.20} & \se{0.23} & \se{0.18} & \se{0.23} & \se{0.23} \\
        \addlinespace[2pt]

        UWaveGestureLibraryZ & 8 & 315 & 3194 & 1284 & 29.30 & 41.25 & 35.94 & \best{23.09} & 37.31 & 23.17 & 30.74  & 27.73\\
         &  &  &  &  & \se{0.18} & \se{0.20} & \se{0.28} & \se{0.17} & \se{0.18} & \se{0.18}  & \se{0.18} & \se{0.19}\\
        \addlinespace[2pt]

        ECG5000 & 4 & 140 & 3464 & 1512 & 5.69 & 4.67 & 5.65 & \best{3.99} & 4.39 & 4.09 & 6.24 & 6.29 \\
         &  &  &  &  & \se{0.10} & \se{0.07} & \se{0.10} & \se{0.08} & \se{0.08} & \se{0.08} & \se{0.11}  & \se{0.10}\\

        \bottomrule
    \end{tabular}
    }
\end{table}

\bibliographystyle{apalike}
\bibliography{ref}

\end{document}